\documentclass[hidelinks,11pt, letterpaper]{article}
\usepackage{fullpage}
\usepackage{amsmath,amssymb,amsfonts,nicefrac}
\usepackage{amsthm}
\usepackage{xspace}
\usepackage{color}
\usepackage{url}
\usepackage{hyperref}
\usepackage[normalem]{ulem}
\usepackage{bm}
\usepackage{bbm}
\usepackage{times}

\usepackage{enumitem}

\newtheorem{theorem}{Theorem}[section]
\newtheorem{thm}{Theorem}[section]

\newtheorem{lemma}[thm]{Lemma}
\newtheorem{corollary}[thm]{Corollary}
\newtheorem{claim}[thm]{Claim}
\newtheorem{proposition}[thm]{Proposition}
\newtheorem{definition}[thm]{Definition}
\newtheorem{remark}[thm]{Remark}

\newtheorem{fact}[thm]{Fact}

\DeclareMathOperator*{\EE}{\mathbb{E}}
\newcommand\card[1]{\left| {#1} \right|}
\newcommand\sett[2]{\left\{ #1 \;\middle\vert\; #2 \right\}}
\newcommand\set[1]{{\left\{ #1 \right\}}}
\newcommand\Prob[2]{{\Pr_{#1}\left[ {#2} \right]}}

\newcommand\cProb[3]{{\Pr_{#1}\left[ #3 \;\middle\vert\; #2 \right]}}

\newcommand\Expect[2]{{\mathop{\mathbb{E}}_{#1}\left[ {#2} \right]}}
\newcommand\Expectnoop[2]{{\mathbb{E}_{#1}\left[ {#2} \right]}}

\newcommand\cExpect[3]{{\mathbb{E}_{#1}\left[ #3 \;\middle\vert\; #2 \right]}}
\newcommand\cExpectop[3]{{\mathop{\mathbb{E}}_{#1}\left[ #3 \;\middle\vert\; #2 \right]}}
\newcommand\norm[1]{\| #1 \|}
\newcommand\power[1]{\set{0,1}^{#1}}
\newcommand\ceil[1]{\lceil{#1}\rceil}

\newcommand\half{{1\over2}}

\newcommand\defeq{\stackrel{\mathit{def}}{=}}
\newcommand\skipi{{\vskip 10pt}}

\newcommand\inner[2]{\langle{#1},{#2}\rangle}
\newcommand\eps{\varepsilon}

\newcommand*\xor{\mathbin{\oplus}}

\renewcommand\geq{\geqslant}
\renewcommand\leq{\leqslant}

\newcommand{\Tdown}[2]{\mathrm{T}_{{#1},{#2}}^{\downarrow}}
\newcommand{\Tup}[2]{\Tdown{#2}{#1}}
\newcommand{\T}{\mathrm{T}}
\providecommand{\distD}{\mathcal{D}}

\renewcommand\epsilon{\eps}

\makeatletter
\newtheorem*{rep@theorem}{\rep@title}
\newcommand{\newreptheorem}[2]{%
\newenvironment{rep#1}[1]{%
\def\rep@title{\bf #2 \ref*{##1} \text{(Restated)} }%
\begin{rep@theorem} }%
{\end{rep@theorem} } }

\newtheorem*{rep@claim}{\rep@title}
\newcommand{\newrepclaim}[2]{%
\newenvironment{rep#1}[1]{%
\def\rep@title{\bf #2 \ref*{##1} \text{(Restated)} }%
\begin{rep@claim} }%
{\end{rep@claim} } }

\newtheorem*{rep@lemma}{\rep@title}
\newcommand{\newreplemma}[2]{%
\newenvironment{rep#1}[1]{%
\def\rep@title{\bf #2 \ref*{##1} \text{(Restated)} }%
\begin{rep@lemma} }%
{\end{rep@lemma} } }

\makeatother
\newreptheorem{theorem}{Theorem}
\newrepclaim{claim}{Claim}
\newreplemma{lemma}{Lemma}

\newtheorem{open-question}{Open Question}

\title{AND Testing and Robust Judgement Aggregation}
\date{\vspace{-5ex}}
\author{Yuval Filmus \thanks{Department of Computer Science, Technion, Israel. This project has received funding from the European Union's Horizon 2020 research and innovation programme under grant agreement No~802020-ERC-HARMONIC.}
\and Noam Lifshitz \thanks{Einstein Institute of Mathematics, Hebrew University of Jerusalem.}
\and Dor Minzer \thanks{Institute for Advanced Study. Supported partially by NSF grant CCF-1412958 and Rothschild Fellowship.}
\and Elchanan Mossel \thanks{Massachusetts Institute of Technology. Department of Mathematics and IDSS. Supported by NSF Award DMS-1737944, by ARO MURI W911NF1910217 and by a Simons Investigator Award in Mathematics (622132)}}

\begin{document}
\maketitle

\begin{abstract}
A function $f\colon\power{n}\to \power{}$ is called an approximate AND-homomorphism if choosing ${\bf x},{\bf y}\in\power{n}$ randomly, we have that
$f({\bf x}\land {\bf y}) = f({\bf x})\land f({\bf y})$ with probability at least $1-\eps$, where $x\land y = (x_1\land y_1,\ldots,x_n\land y_n)$.
We prove that if $f\colon \{0,1\}^n \to \{0,1\}$ is an approximate AND-homomorphism, then $f$ is $\delta$-close to either a constant function or an AND function,
where $\delta(\epsilon) \to 0$ as $\epsilon\to0$. This improves on a result of Nehama, who proved a similar statement in which $\delta$ depends on~$n$.

Our theorem implies a strong result on judgement aggregation in computational social choice. In the language of social choice, our result shows that if $f$ is $\eps$-close to satisfying judgement aggregation, then it is $\delta(\eps)$-close to an oligarchy (the name for the AND function in social choice theory). This improves on Nehama's result,
in which $\delta$ decays polynomially with $n$.

Our result follows from a more general one, in which we characterize approximate solutions to the eigenvalue equation $\T f = \lambda g$,
where $\T$ is the downwards noise operator $\T f(x) = \EE_{{\bf y}}[f(x \land {\bf y})]$, $f$ is $[0,1]$-valued, and $g$ is $\{0,1\}$-valued. We identify
all exact solutions to this equation, and show that any approximate solution in which $\T f$ and $\lambda g$ are close is close to an exact
solution.
\end{abstract}

\section{Introduction}

Which functions $f\colon \{0,1\}^n \to \{0,1\}$ satisfy
\[
 f({\bf x} \land {\bf y}) = f({\bf x}) \land f({\bf y}) \text{ w.p. } 1-\epsilon,
\]
where ${\bf x}, {\bf y}$ are chosen uniformly at random?

If $\epsilon = 0$, it is not hard to check that $f$ is either constant or an AND of a subset of the coordinates. Nehama~\cite{Nehama} showed that when $\epsilon > 0$, $f$ must be $O((n\epsilon)^{1/3})$-close to a constant function or to an AND (in other words, $\Pr[f \neq g] = O((n\epsilon)^{1/3})$, where $g$ is constant or an AND).
The main result in this paper implies, as a corollary, a similar statement, in which the distance between $f,g$ vanishes with $\eps$, without any dependence on $n$.
\begin{theorem} \label{thm:intro-main}
For each $\delta > 0$ there is $\epsilon > 0$ such that if $f\colon \{0,1\}^n \to \{0,1\}$ satisfies
\[ \Pr_{{\bf x}, {\bf y}}[f({\bf x} \land {\bf y}) = f({\bf x}) \land f({\bf y})] \geq 1-\epsilon, \]
then $f$ is $\delta$-close to a constant or an AND.
\end{theorem}
\noindent Our technique is in fact more general, and allows us to study the multi-function version of this problem, in which we are interested in triples of functions $f,g,h\colon\power{n}\to\power{}$ that satisfy $f({\bf x}\land {\bf y}) = g({\bf x})\land h({\bf y})$
with probability at least $1-\eps$. % when ${\bf x},{\bf y}$ are chosen uniformly from $\power{n}$.
%\noindent We also prove a version of this theorem when the three instances of $f$ are allowed to be different functions.

%As an aside, Theorem~\ref{thm:intro-main} implies that the soundness of a property testing algorithm of Parnas, Ron and Samorodnitsky~\cite{PRS02} (whose goal is to test whether the input function is of the form $x_i$), which the authors couldn't analyze.

If we replace $\land$ with $\oplus$ in the above problem, then the result corresponding to Theorem~\ref{thm:intro-main} is the well-known
soundness of the Blum--Luby--Rubinfeld \emph{linearity test} \cite{BLR}, that plays an important role in the construction of PCPs \cite{AS,ALMSS,Hastad}.
By now, many proofs for the soundness of this test are known: self-correction~\cite{BLR}, Fourier analysis~\cite{BCHKS,Hastad}, induction~\cite{DDGKS}.
Unfortunately, all of these proofs rely on $\oplus$ being a group operation (either directly or via Fourier analysis), and hence do not extend to our setting.

Our approach recasts the problem as determining the approximate eigenfunctions of a one-sided noise operator. Define
the operator $\T$ acting on functions $f\colon\power{n}\to\power{}$ in the following way:
\[
 (\T f)(x) = \Expect{\mathbf{y}}{f(x \land \mathbf{y})}.
\]
Using this operator, the premise of Theorem~\ref{thm:intro-main} implies that $f$ is an approximate eigenfunction
of this operator, i.e.\ $\T f \approx \lambda f$, where $\lambda = \EE[f]$ is the average of $f$.
Here, by approximate solution we mean that the $L_1$ distance between the two functions is small:
$\EE_{{\bf x}}[|\T f({\bf x}) - \lambda f({\bf x})|] \leq \epsilon$.

If we replace $\T$ with the usual two-sided noise operator $T_\rho$, then a short spectral argument shows that if $f$ is an approximate eigenfunction than it must be close to an exact eigenfunction of $T_\rho$. Unfortunately, the spectral argument relies on orthogonality of eigenspaces of $T_\rho$, a property which $\T$ doesn't satisfy (its eigenspaces are spanned by ANDs, which aren't orthogonal). Indeed, $\T$ has approximate eigenfunctions beyond ANDs. Here are two examples:
\begin{align*}
f_1(x) &=
\begin{cases}
x_1 \lor x_2 & \text{if } |x| \geq n/3 \\
x_1 \oplus x_2 & \text{if } |x| < n/3	
\end{cases}
&
f_2(x) &=
\begin{cases}
1 & \text{if } |x| \geq n/3 \\
\mathrm{Ber}(\lambda) & \text{if } |x| < n/3	
\end{cases}	
\end{align*}
The second example is probabilistic: $f_2(x) = 1$ with probability $\lambda < 1$ independently for each $|x| < n/3$. These functions satisfy $\T f_1 \approx \frac{1}{2} f_1$, $\T f_2 \approx \lambda f_2$. They are not counterexamples to Theorem~\ref{thm:intro-main} since $\EE[f_1] \approx 3/4 > 1/2$ and $\EE[f_2] \approx 1 > \lambda$.

Note that each one of the functions $f_1,f_2$ is essentially composed of two, completely different ``sub-functions'':
one defined on high Hamming-weight inputs, and another defined on low Hamming-weight inputs.
This suggests decoupling the two functions, and considering the generalized eigenvalue problem
\[
 \T f = \lambda g, \text{ where } f\colon \{0,1\}^n \to [0,1] \text{ and } g\colon \{0,1\}^n \to \{0,1\}.
\]
Here $f$ represents the low-weight part, and $g$ represents the high-weight part. We represent the probabilistic aspect of the low-weight part by allowing $f$ to take on values in the solid interval $[0,1]$.

The two examples above corresponds to \emph{exact} solutions of this generalized problem: $\T(x_1 \oplus x_2) = \tfrac{1}{2}(x_1 \lor x_2)$ and $\T \lambda = \lambda \cdot 1$.
Therefore, as a prerequisite to characterizing approximate eigenfunctions of $\T$, we must first study exact solutions to the more general
two-function version. We show:
\begin{theorem} \label{thm:intro-Tf=g}
If $f\colon \{0,1\}^n \to [0,1]$ and $g\colon \{0,1\}^n \to \{0,1\}$ satisfy $\T f = \lambda g$ then either $f=g=0$ or there exist disjoint subsets $S_1,\ldots,S_m \subseteq [n]$, where $m \leq \log_2 (1/\lambda)$, such that
\begin{align*}
f(x) &= \bigwedge_{i=1}^m \bigoplus_{j \in S_i} x_i, &
g(x) &= \bigwedge_{i=1}^m \bigvee_{j \in S_i} x_i.
\end{align*}
Moreover, if $f$ is monotone then $g$ is an AND and $f = g$.
\end{theorem}

Thus if $\T f = \lambda g$ then $g$ is an ``AND-OR'' and $f$ is the corresponding ``AND-XOR'' (or $f=g=0$). When $f$ is monotone, $g$ must be an AND, and so $f = g$.
Using Theorem~\ref{thm:intro-Tf=g}, we can then actually solve the more general problem of characterizing approximate solutions to the equation $\T f = \lambda g$.
%Namely, we show that:
\begin{theorem} \label{thm:intro-approximate}
If $f\colon \{0,1\}^n \to [0,1]$ and $g\colon \{0,1\}^n \to \{0,1\}$ satisfy $\T f \approx \lambda g$ then either $f\approx g\approx 0$ or $g$ is close to an AND-OR and $f$ is close to an AND-XOR.

Moreover, if $f$ is monotone then $f,g$ are both close to a constant or to an AND.
\end{theorem}
We also show how to deduce Theorem~\ref{thm:intro-main} from Theorem~\ref{thm:intro-approximate}.
We remark that Theorem~\ref{thm:intro-approximate} is stated in a somewhat informal way: the closeness
of the function $f$ to an AND-XOR function has to be stated in a more subtle way (since otherwise it is false),
and we defer this point to the formal statement of the theorems in Section~\ref{sec:main-results}.

\subsection{Other variants}
\subsubsection*{Other noise rates}
Nehama~\cite{Nehama} also considers the more general equation
\[
 f({\bf x_1} \land \cdots \land {\bf x_m}) \approx f({\bf x_1}) \land \cdots \land f({\bf x_m}),
\]
where each one of ${\bf x_1},\ldots,{\bf x_m}$ is sampled uniformly and independently from $\power{n}$. We can reduce
this problem, in a similar manner, to an eigenfunction of an appropriate operator $\T^{(m)}$, defined by
\[
\T^{(m)} f(x) = \EE_{{\bf y_1},\ldots,{\bf y_{m-1}}}{f( x \land {\bf y_1}\land\ldots\land{\bf y_{m-1}})}.
\]
Our techniques also apply to such operators (and in fact to a slightly richer family of noise operators), and
we prove variants of Theorem~\ref{thm:intro-Tf=g} and Theorem~\ref{thm:intro-approximate} in this case as well:
\begin{theorem} \label{thm:intro-m>2}
Let $m > 2$. If $f\colon \{0,1\}^n \to [0,1]$ and $g\colon \{0,1\}^n \to \{0,1\}$ satisfy $\T^{(m)}f = \lambda g$ then either $f=g=0$ or $f=g$ is an AND.

Furthermore, if $\T^{(m)}f \approx \lambda g$ then either $f\approx g\approx 0$ or $f,g$ are close to an AND.
\end{theorem}

\subsubsection*{One-sided error version}
Finally, we consider the one-sided error version of the equation $\T f =\lambda g$. That is, suppose
we have a bounded function $f\colon\power{n}\to[0,1]$, and a Boolean function $g$, such that with probability $1-\eps$ over ${\bf x}$:
(a) if $g({\bf x}) = 1$, then $f({\bf x}\land {\bf y})\geq \lambda$ with constant probability over ${\bf y}$,
and (b) if $g({\bf x}) = 0$ then $f({\bf x}\land {\bf y})\leq \eps$ with probability $1-\eps$ over ${\bf y}$.

We note that this condition is a relaxation of the approximate eigenvalue condition.
In this case, we prove a weaker structural result than in Theorem~\ref{thm:intro-approximate},
namely that $g$ is close to a monotone junta.
\begin{theorem} \label{thm:intro-one-sided}
Suppose that $f\colon \{0,1\}^n \to [0,1]$ and $g\colon \{0,1\}^n \to \{0,1\}$ satisfy the following condition: when $g=0$, $\T f$ is typically small; and when $g=1$, $\T f$ is typically at least $\lambda$.\footnote{In contrast, in Theorem~\ref{thm:intro-approximate} we ask that $\T f$ be typically \emph{close} to $\lambda$.} Then $g$ is close to a monotone Boolean junta.\footnote{A junta is a function depending on a constant number of coordinates.}
\end{theorem}
We remark that while the structural result in this case is weaker, it is for a good reason: for any monotone junta $f$, choosing $g=f$ yields
an approximate, one-sided error solution.

\subsection{Social choice interpretations of approximate eigenfunctions}
%The original motivation in Kornhauser and Sager~\cite{KornhauserSager:86} was
The seminal work of Kornhauser and Sager~\cite{KornhauserSager:86} discusses
a situation where three cases $A,B,C$ are considered in court, and by law, one should rule against $C$ if and only if 
there is a ruling against both $A$ and $B$.
When several judges are involved, their opinions should be aggregated using a function $f$ that preserves this law, that is, satisfies $f(x \land y) = f(x) \land f(y)$; we say that $f$ is an \emph{AND-homomorphism}.
List and Pettit~\cite{LP02,LP04} showed that the only non-constant aggregation functions that are AND-homomorphisms are the AND functions, known in the social choice literature as \emph{oligarchies}. 

Let the individual opinions of the judges are
$x_1,\ldots,x_n$ on $A$, $y_1,\ldots,y_n$ on $B$, and $x_1 \land y_1,x_2 \land y_2,\ldots,x_n \land y_n$ on $C$. 
The characterization of robust judgement aggregation that we prove in this paper (Theorem~\ref{thm:intro-main}) states that if typically 
$f(x \land y) = f(x) \land f(y)$, then $f$ is close to an oligarchy.

The characterization in terms of approximate eigenfunctions (Theorem~\ref{thm:intro-approximate}) actually shows more. Suppose that opinions are aggregated according to a monotone function $f$ which satisfies the following two conditions:
\begin{itemize}
\item 
There is rarely a ruling against $C$ unless there is a ruling against $A$ and a ruling against $B$.%\\ (Formally, $\Pr[f(x \land y) = 1, f(x) \land f(y) = 0]$ is small.)
\item 
Suppose that there is a ruling against $A$. If there is also a ruling against $B$, then with probability roughly $q$, there is a ruling against $C$.\\ (Formally, for typical $x$ such that $f(x) = 1$, we have $\Pr[f(x \land y) = 1 \mid f(y) = 1] \approx q$.)
\end{itemize}  
Then $f$ must be close to an oligarchy or to a constant function, and $q \approx 1$.

In fact, the second condition can be weakened significantly:
\begin{itemize}
\item
Suppose that there is a ruling against $A$. Then with probability roughly $\lambda$, there is a ruling against $C$.
\end{itemize}
Theorem~\ref{thm:intro-approximate} implies that $f$ must be close to an oligarchy or to a constant function, and $\lambda \approx \EE[f]$.

Similarly, Theorem~\ref{thm:intro-one-sided} shows that $f$ has to be close to a monotone junta if the second condition above is replaced with either of the following two conditions:
\begin{itemize}
\item 
Suppose that there is a ruling against $A$. If there is also a ruling against $B$, then with probability \emph{at least} $q$, there is a ruling against $C$.
\item
Suppose that there is a ruling against $A$. Then with probability \emph{at least} $\lambda$, there is a ruling against $C$.
\end{itemize} 

Thus our results do not only strengthen robust judgement aggregation in a quantitative way, but also in a qualitative way. 

\subsection{Our techniques}
Our main result, Theorem~\ref{thm:intro-main}, easily follows from Theorem~\ref{thm:intro-approximate},
which is our main technical result. Below we sketch the proof idea of Theorem~\ref{thm:intro-approximate}
(the proofs of Theorem~\ref{thm:intro-m>2} and Theorem~\ref{thm:intro-one-sided} follow similar lines).

Suppose $f,g$ are functions as in Theorem~\ref{thm:intro-approximate} that satisfy $\T f \approx \lambda g$. The first step of the proof
is to show that the function $g$ is close to a junta $h$, i.e.\ to a function depending only on constantly many variables. To get some
intuition for that, note that if $\T$ was the standard noise operator, then the function $\T f$ has exponentially decreasing tail
and hence it is very concentrated on its low Fourier levels.
When the operator $\T$ is the one-sided noise operator, one can
actually use similar reasoning to claim that $g$ again has an exponentially decaying tail (as observed by Lifshitz~\cite{Lifshitz}).
Since $g$ is Boolean and $g\approx \frac{1}{\lambda} \T f$, this observation
would then allow us to use structure theorems on Boolean functions (more specifically, a result of Bourgain~\cite{Bourgain} or
of Kindler and Safra~\cite{KS}) to conclude that $g$ is (close to) a junta.

Thus, ignoring some (important) technical details, one can think of $g$ as a function of a constant number of variables, and since
the proximity parameter between $\T f$ and $g$ can be taken to be very small (even in comparison to the number of variables
$g$ depends on), one may as well think of it as being $0$. In other words, the problem essentially boils down to studying
exact solutions to the equation $\T f = g$ when $n$ is constant, which is where Theorem~\ref{thm:intro-Tf=g} enters the picture.
Using it, we prove the structural result on $g$; getting the structural result on $f$ then amounts to averaging $f$ over coordinates
that $g$ does not depend on (since those could be thought of as a ``source of randomness'' as in the example of $f_2(x)$ above),
and then inverting the operator $\T$ acting on functions of a constant number of variables.

This ends the informal description of our techniques. We remark that actually composing the two components, namely the
approximation by junta and the solution to the exact equation, is more subtle and requires some care. We also remark
that in the case of one-sided error (Theorem~\ref{thm:intro-one-sided}), the Fourier-analytic argument alluded to above, which implies
that $g$ is close to a junta, does not seem to be applicable. We thus present an alternative, more combinatorial
argument that captures this case as well.
%Our main result, Theorem~\ref{thm:intro-main}, easily follows from Theorem~\ref{thm:intro-approximate}, which is our main technical result. Theorem~\ref{thm:intro-approximate}, in turn, is proved by reduction to Theorem~\ref{thm:intro-Tf=g}, whose proof is combinatorial. Let us focus on the proof of Theorem~\ref{thm:intro-approximate}.
%
%Every function of the form $T_\rho f$, where $f$ is bounded and $T_\rho$ is the usual two-sided noise operator, is noise stable (concentrated on the low Fourier degrees). This follows from the eigenvalues of $T_\rho$. A similar statement holds for functions of the form $\T f$, for a similar reason:
%\[
% \T \chi_S^{(1/4)} = 3^{-|S|/2} \chi_S^{(1/2)},
%\]
%where $\chi_S^{(p)}$ is the $p$-biased Fourier character $\prod_{i \in S} \frac{x_i-p}{\sqrt{p(1-p)}}$ (this property was observed by Lifshitz ~\cite{Lifshitz}). Applying Bourgain's junta theorem, we deduce that if $\T f \approx \lambda g$ then $g$ is close to a junta $h$. Adapting the proof of Theorem~\ref{thm:intro-Tf=g}, we deduce that $h$ is close to an AND-OR, losing a factor depending exponentially on the size of the junta. Altogether, we deduce that $g$ is close to an AND-OR (using a subtle argument), and (inverting $\T$), that $f$ is close to the corresponding AND-XOR (or both are close to the zero function).

% Lying here a bit...

\subsection{Related work}
\subsubsection{Quantitative social choice theory}
Social choice theory studies how to aggregate the opinion of a number of agents.
Already in the 18th century, Condorcet~\cite{Condorcet:85} noted that natural aggregation rules often result in paradoxes. A large body of work has been developed in economics since the middle of the 20th century, in which it was shown that natural aggregation tasks have \emph{no} good aggregation functions.
The two most famous results in this area are Arrow's impossibility theorem~\cite{Arrow:50,Arrow:63} and the Gibbard--Satterthwaite (GS) manipulation theorem~\cite{Gibbard:73,Satterthwaite:75}. The questions of aggregation re-emerged in the context of multi-agent systems in computer science, where the hope was that either the probability of paradoxical outcome is small, or there is computational difficulty in arriving at a paradoxical outcome, see e.g.~\cite{BaToTr:89b} and the survey \cite{Faliszewski_Procaccia_2010}. A sequence of results showed that this is not the case by proving strong and general quantitative versions of both Arrow's Theorem~\cite{Kalai:02,Mossel:12,Keller:12,MoOlSe:13} and the GS Theorem~\cite{FrKaNi:08,FKKN:11,IsKiMo:12,MosselRacz:15}, as well as results interpolating the two theorems~\cite{Falik2014}.

The main motivation for the problem discussed in this paper is Judgement Aggregation.
This problem is considered in a fascinating paper in the Yale Law Review by Kornhauser and Sager~\cite{KornhauserSager:86}. In particular, toward the end of the paper, the authors considered legal cases, where the judgement aggregation function $f$ should satisfy $f(x \land y) = f(x) \land f(y)$. They observe that this does not hold when
$f$ is the majority vote on three opinions.

The failure of Majority, which mirrors the failure of Majority in ranking that was observed by Condorcet, led to work by List and Pettit~\cite{LP02,LP04},
who characterized exactly the functions $f$ that are AND-homomorphisms, i.e., oligarchies.
The question of judgement aggregation has attracted much attention in philosophy, social epistemology,
and artificial intelligence~\cite{Pigozzi:16}.
In the context of multi-agent systems, when the number of agents is large, it makes sense to ask if it is possible to achieve approximate judgment aggregation.
Our results show that this can only be achieved in the obvious way, i.e., by almost-oligarchies.

We note that the %combinatorial 
study of judgement aggregation extends well beyond AND-homomorphisms, to other types of homomorphisms, and indeed such a theory of polymorphisms
is well-developed~\cite{NP08,DH09,DH10a,DH10b,DH10c,SX18}. We leave if for future work to investigate robust versions of these results.
%considered judgement aggregation, which is the following problem. Given a function $\phi\colon \{0,1\}^m \to \{0,1\}$, what functions $f\colon \{0,1\}^n \to \{0,1\}$, possibly satisfying other natural properties, satisfy the identity
%\[
% f(\phi(x_1,\ldots,x_m)) = \phi(f(x_1),\ldots,f(x_m)) \text{ for all } x_1,\ldots,x_m \in \{0,1\}^n?
%\]
%Here, $\phi(x_1,\ldots,x_m)$ denotes elementwise application. Such functions are also known as polymorphisms of $\phi$. Nehring and Puppe~\cite{NP08} solved the case of monotone Boolean $\phi$ (under mild assumptions on $f$), and Dokow and Holzman~\cite{DH09,DH10a,DH10b,DH10c} solved the case of arbitrary Boolean $\phi$, and proved some results for larger alphabets. See also Szegedy and Xu~\cite{SX18}, who brought in the universal-algebraic angle.

\subsubsection{Property testing}
The work of Blum, Luby and Rubinfeld \cite{BLR} has been extended to more general settings by various authors. For example, Moore and Russell~\cite{MR15} and Gowers and Hatami~\cite{GH17} considered approximate representations of finite groups. Other authors had considered infinite groups, see for example the survey of Hyers and Rassias~\cite{HR92}. Theorem~\ref{thm:intro-main} generalizes Blum--Luby--Rubinfeld in a different direction, to approximate polymorphisms,
where there is no group structure.

We remark that Theorem~\ref{thm:intro-main} implies that the soundness of a property testing algorithm of Parnas, Ron and Samorodnitsky~\cite{PRS02}, whose goal is to test whether the input function is a dictatorship (they also consider the more general problem of testing whether the input function is close to an AND function).
The authors proposed the following natural tester, which they were unable to analyze: test that $f$ has expectation $1/2$ and satisfies $f(x \land y) = f(x) \land f(y)$. Instead, they proposed a somewhat less natural tester. Our results
imply that their original tester also works.

It is interesting to explore if there is a relationship between our results and different
notions of approximate polymorphisms that appear in the literature~\cite{BCR15,BG20}, which were used to prove hardness of approximation results.

\paragraph{Organization.} We formally state our results in Section~\ref{sec:main-results}. After some preliminaries in Section~\ref{sec:prel}, we prove the various results in Sections~\ref{sec:main_2func}--\ref{sec:large_noise}. We close the paper by stating some open questions in Section~\ref{sec:open-questions}.

% we have a second Organization paragraph at the end of Section 2, is that a problem?

\section{Main Results} \label{sec:main-results}
Let $\mu_p$ denote the $p$-biased measure on $\power{n}$.
Let $L_2(\power{n}, \mu_p)$ be the space of real-valued functions
on $\power{n}$ equipped with the inner product $\inner{f}{g} = \Expect{{\bf x}\sim\mu_p^n}{f({\bf x})g({\bf x})}$.

\begin{definition}
  For $q\leq p$, the distribution $({\bf y},{\bf x})\sim \mathbb{D}(q,p)$ over $\power{n}\times\power{n}$
  is the distribution in which for each $i\in[n]$ independently, we have $\Prob{}{{\bf y}_i = 1} = q$ and $\Prob{}{{\bf x}_i = 1} = p$,
  and always ${\bf y}_i\leq {\bf x}_i$.
\end{definition}
One way to generate inputs $({\bf y},{\bf x})\sim\mathbb{D}(q,p)$ that will be useful for us is as follows.
Sample ${\bf x}\sim\mu_p$ and ${\bf z}\sim\mu_{q/p}$ independently, and output $({\bf x}\land {\bf z},{\bf x})$. (Here $\land$ refers to the coordinatewise
AND operation, i.e.\ for each $i\in[n]$ we have $({\bf x}\land {\bf z})_i = {\bf x}_i\land {\bf z}_i$.)

For $\rho\in (0,1)$, define the one-sided operator $\Tdown{p}{\rho p}$ as follows.
For any function $f\colon(\power{n},\mu_{\rho p})\to\power{}$,
the function $\Tdown{p}{\rho p} f \colon(\power{n},\mu_p)\to\power{}$ is given by
\[
\Tdown{p}{\rho p} f(x) = \cExpect{({\bf y},{\bf v})\sim \mathbb{D}(\rho p, p)}{{\bf v}=x}{f({\bf y})}.
\]
Equivalently, we have $\Tdown{p}{\rho p} f(x) = \Expect{{\bf z}\sim\mu_{\rho}}{f(x\land {\bf z})}$.

Next, we shall discuss the spectrum
(eigenvectors and eigenvalues) of the operator $\Tdown{p}{\rho p}$. We remark that throughout this section, the parameters $p$ and
$\rho$ should be thought of as constants bounded away from $0,1$.

For each $S\subseteq[n]$, the function ${\sf AND}_S\colon\power{n}\to\power{}$ defined by ${\sf AND}_{S}(x) = \bigwedge_{i\in S}{x_i}$
is an eigenvector of $\Tdown{p}{\rho p}$ with eigenvalue $\rho^{\card{S}}$ (we omit the easy proof).
Moreover, these are the only eigenvectors of $\Tdown{p}{\rho p}$ that are Boolean valued.\footnote{To see that, note that any function $f$ can be written as a linear combination of AND functions, and if $f$ is an eigenvector then all of these
ANDs are of the same size, say with coefficients $\alpha_1,\ldots,\alpha_m$. Considering the value of $f$ on the minterms of these ANDs,
one concludes that all of the $\alpha$'s must be $1$, and considering the value of $f$ on the all-$1$ string, one concludes that $m=1$.}
Our goal in this paper is to find a robust version of this characterization of the Boolean eigenvectors of $\Tdown{p}{\rho p}$.

We say that a function
$f$ is an $\eta$-approximate eigenvector with eigenvalue $\lambda$, if $\norm{\Tdown{p}{\rho p}f - \lambda f}_1\leq \eta$ (here and throughout
the paper we will consider the $\ell_1$ norm with respect to the $\mu_p$ measure).
What structure do Boolean, approximate eigenvectors of $\Tdown{p}{\rho p}$ must have? A natural conjecture
would be that any such function has to be close to an exact eigenvector, which by Booleanity would have to be an AND-function
over $\approx \log_{\rho}(\lambda)$ variables. However, this conjecture turns out to be false, as the following example demonstrates.

Set $p = \rho = \half$, and consider the function $f$ defined by $f(x) = x_1\lor x_2$ for inputs
whose hamming weight is $n/2\pm \sqrt{n\log n}$, and by $x_1\xor x_2$ for the rest of the inputs.
It is easy to see that $f$ is far from any AND function on the $\mu_{1/2}$ measure, and
we argue that $\norm{\Tdown{1/2}{1/4} f - \half f} = o(1)$. By definition, for each $x$,
$\Tdown{1/2}{1/4} f(x)$ is the probability that picking ${\bf z}\sim\mu_{1/2}$, we have $f(x\land {\bf z})=1$.
Except with probability $o(1)$, the hamming weight of $x,x\land {\bf z}$ is roughly $n/2, n/4$ respectively,
and we focus only on this event. In this case if $f(x) = 0$ then $x_1=x_2 = 0$ and thus clearly
$f(x\land {\bf z}) = 0$, and otherwise $(x_1,x_2)\neq (0,0)$ and therefore $f(x\land {\bf z}) = x_1{\bf z}_1\xor x_2{\bf z}_2$
is a uniform bit, i.e.\ $f(x\land {\bf z}) = 1$ with probability $\half$.

\begin{remark}
  It is worth noting that for any constant $\lambda>0$, there are approximate eigenvectors
  of $\Tdown{1/2}{1/4}$ with eigenvalue $\lambda$ (not only for $\lambda = 2^{-k}$). Indeed,
  the function $f$ that is constantly $1$ on inputs with Hamming weight $n/2 \pm \sqrt{n\log n}$,
  and on each other point $x$ independently, we take $f(x) = 1$ with probability $\lambda$, is (with probability $1-o(1)$) an
  approximate eigenvector with eigenvalue $\lambda$.
\end{remark}

\subsection{The basic two-function version}
Since the previous example is essentially composed of two different functions (one around the middle slice and the other around the $n/4$-slice),
it makes sense to consider the two-function version of the approximate eigenvector problem. Namely, let $f\colon(\power{n},\mu_{p\rho})\to\power{}$,
$g\colon(\power{n},\mu_p)\to\power{}$, and $\lambda\in(0,1)$ be such that $\norm{\Tdown{p}{\rho p} f - g}_1\leq \eta$. What
can we say about $f$ and $g$? We note that in this case, even the exact version of the problem, i.e.\ determining
which functions can satisfy $\Tdown{p}{\rho p} = g$, is already unclear (and in fact, as it turns out, understanding
solutions to the exact problem is a key step in solving the approximate problem).

The version of the problem we will consider is actually more general and allows the function $f$ to take values in $[0,1]$.
It turns out that the structure of the solutions heavily depends on $\rho$, and we consider three different regimes:
$0<\rho < \half$, $\rho = \half $ and $\half < \rho < 1$. We remark that all of the results apply in particular for the
original approximate eigenvector problem, i.e.\ the case $f=g$.

The first range, $0 < \rho < \half$, is the simplest, and we have the following result.
\begin{thm}\label{thm:main_2func_smallrho}
  For any $\zeta>0$ there is $J\in\mathbb{N}$ such that for any $\eps>0$ there is $\eta > 0$
  such that the following holds.
  Let $p\in[\zeta,1-\zeta]$, $\rho\in[\zeta,\half-\zeta]$ and $\lambda\in[\zeta,1]$,
  and let $f\colon\power{n}\to[0,1]$ and $g\colon\power{n}\to \power{}$ satisfy $\norm{\Tdown{p}{\rho p} f - \lambda g}_1\leq \eta$.
  Then either $f,g$ are $\eps$-close to the zero function, or there is a set $T\subseteq[n]$ of size at most $J$ such that:
  \begin{itemize}
    \item $g$ is $\eps$-close to ${\sf AND}_T$.
    \item After averaging outside $T$, $f$ is $\eps$-close to $\rho^{-|T|}\lambda\cdot{\sf AND}_T$.
    More precisely, the function $\tilde{f}\colon \power{T}\to[0,1]$ given by $\tilde{f}(x) = \Expect{{\bf y}\sim\mu_{\rho p}^{[n]\setminus T}}{f({\bf y},x)}$
    is $\eps$-close to $\rho^{-|T|}\lambda\cdot{\sf AND}_T$.
  \end{itemize}
\end{thm}
\noindent (This range corresponds to the operators $\T^{(m)}$ mentioned in Theorem~\ref{thm:intro-m>2}.)

In the second range, $\rho = 1/2$, the structure of $f$ and $g$ may be more complicated (we have already seen an example
in this range where $g = {\sf OR}_T$ and $f = {\sf XOR}_T$ for $T$ of size $2$).
\begin{definition}
  A function $g\colon\power{n}\to\power{}$ is called an AND-OR function of width $m$ if
  there are disjoint sets $A_1,\ldots,A_m$ such that $g(x) = \bigwedge_{i\in[m]}\bigvee_{j\in A_i} x_j$.
\end{definition}
\begin{definition}
   A function $g\colon\power{n}\to\power{}$ is called an AND-XOR function of width $m$ if
  there are disjoint sets $A_1,\ldots,A_m$ such that $g(x) = \bigwedge_{i\in[m]}\bigoplus_{j\in A_i} x_j$.
\end{definition}

\begin{thm}\label{thm:main_2func}
  For any $\zeta>0$ there is $m\in\mathbb{N}$ such that for any $\eps>0$ there are $\eta>0$, $J\in\mathbb{N}$ such that the following holds
  for all $p\in[\zeta,1-\zeta]$ and $\lambda\in[\zeta,1]$.
  If $f\colon\power{n}\to[0,1]$ and $g\colon\power{n}\to \power{}$ satisfy $\norm{\Tdown{p}{p/2} f - \lambda g}_1\leq \eta$.
  Then there is a set $T\subseteq[n]$ of size at most $J$, and a partition $T = A_1\cup\dots \cup A_r$ for $r\leq m$ such that either $f,g$ are $\eps$-close to the zero function, or:
  \begin{itemize}
    \item $g$ is $\eps$-close to $\bigwedge_{i\in[r]}\bigvee_{j\in A_i} x_j$ (i.e.\ to an AND-OR function of width at most $m$).
    \item After averaging outside $T$, $f$ is $\eps$-close to a multiple of $\bigwedge_{i\in[r]}\bigoplus_{j\in A_i} x_j$.
    More precisely, the function $\tilde{f}\colon \power{T}\to[0,1]$ given by $\tilde{f}(x) = \Expect{{\bf y}\sim\mu_{p/2}^{[n]\setminus T}}{f({\bf y},x)}$
    is $\eps$-close to $2^r\lambda \cdot\bigwedge_{i\in[r]}\bigoplus_{j\in A_i} x_j$.
  \end{itemize}
\end{thm}
Both Theorem~\ref{thm:main_2func_smallrho} and Theorem~\ref{thm:main_2func} can be shown to be qualitatively tight.
For Theorem~\ref{thm:main_2func}, for example, any pair of functions $f,g$ where $g$ is an AND-OR function
and $f$ is the corresponding AND-XOR function is an exact solution. To see that some averaging is needed
to get a structure for $f$, note that given a pair of approximate solutions $f,g$, one may
sub-sample $f$, i.e.\ change the value on each $x$ such that $f(x) = 1$ with probability $1/2$,
to get a new approximate solution with $\lambda/2$, and $f$ has no apparent structure (other than being a multiple
of AND-XOR after averaging).

Quantitatively, the dependence of $\eta$ on $\eps$ in Theorem~\ref{thm:main_2func_smallrho}
is quasi-polynomial $\eta = \exp(-\Theta_{\zeta}(\log^2(1/\eps)))$.
In contrast, the dependence in Theorem~\ref{thm:main_2func} is exponentially worse,
i.e.\ $\eta = \exp(-\exp(\Theta_{\zeta}(\log^2(1/\eps))))$. The source of this difference is that in
the case of Theorem~\ref{thm:main_2func_smallrho} (and also in Theorem~\ref{thm:basic} and Theorem~\ref{thm:mono_and})
we are able to prove stronger approximation by junta results than in Theorem~\ref{thm:main_2func}.
Namely, we show that there is $J(\zeta)$ (independent of the proximity to junta parameter $\eps$), such
that if $\eta$ is a sufficiently small function of $\eps$, then $g$ is $\eps$-close to a $J$-junta.
In the case of Theorem~\ref{thm:main_2func}, we are forced to allow the size of the junta $J$ to also depend on $\eps$.
As far as we know, in both cases the dependence of $\eta$ on $\eps$ could be much better, perhaps even polynomial.

In the third range of parameters, $\half < \rho < 1$, the solutions to the problem have a richer structure. It can be shown,
for example, that there are $\rho\in(\half,1)$, $\lambda\in(0,1)$ and a function $f\colon\power{n}\to[0,1]$ such that $f$ and $g(x) = {\sf Maj}(x_1,x_2,x_3)$
are an exact solution to $\Tdown{1/2}{\rho/2} f = \lambda g$. In this case we only show a relatively weak structure, namely that $g$ is close
to a monotone junta (see Theorem~\ref{thm:one_sided_error}). We remark that in order to get a stronger structure, one would only need
to classify all exact solutions to the equation $\Tdown{p}{\rho p} f = \lambda g$ for $\rho > 1/2$.

\subsection{Special cases}
We next present our result for a few special cases of interest, in which we are able to prove a stronger structure.
The first result is concerned with the case when the approximate eigenvalue is large:
\begin{thm}\label{thm:basic}
  For every $\zeta,\eps>0$ there is $\eta > 0$
  such that the following holds for any $\rho, p\in[\zeta,1-\zeta]$ and $\lambda\geq \rho + \zeta$.
  If $f\colon(\power{n},\mu_{\rho p})\to[0,1]$ and $g\colon(\power{n},\mu_p)\to\power{}$ satisfy
  $\norm{\Tdown{p}{\rho p} f - \lambda g}_1\leq \eta$, then $g$ is $\eps$-close to a constant function
  $\Gamma \in \{0,1\}$, and $\Expect{{\bf x}\sim\mu_{\rho p}}{f({\bf x})}$ is $\eps$-close to $\lambda \Gamma$.
\end{thm}

Next, we consider the case in which $f$ is a monotone function. In this case (and actually for a more relaxed case
in which $f$ is ``almost monotone''), we show that $g$ must be an AND function and $f$ must be a multiple of
that AND function after averaging. We also get quantitatively stronger relation between $\eps$ and $\eta$.
\begin{thm}\label{thm:mono_and}
  For every $\zeta>0$, $\eps>0$ there exists $\eta > 0$
  such that the following holds for all $p,\rho\in[\zeta,1-\zeta]$ and $\lambda\in[\zeta,1]$.
  If $f\colon(\power{n},\mu_{\rho p})\to[0,1]$ is monotone and $g\colon(\power{n},\mu_p)\to\power{}$ satisfies
  $\norm{\Tdown{p}{\rho p} f - \lambda g}\leq \eta$, then:
  \begin{itemize}
    \item There exists $T\subseteq[n]$ of size at most $\ceil{\log(2/\lambda)}$ and a function $h$
    that is either constant (in which case $T = \emptyset$)  or ${\sf AND}_T$, such that $\norm{g-h}_1\leq\eps$.
    \item $\tilde{f}\colon\power{T}\to[0,1]$ given by $\tilde{f}(x) = \Expect{{\bf y}\sim \mu_{\rho p}}{f(x,{\bf y})}$
    is $\eps$-close in $L_{\infty}$-norm to $\rho^{-\card{T}}\lambda \cdot h$.
  \end{itemize}
\end{thm}
The monotonicity condition in Theorem~\ref{thm:mono_and} can be relaxed to ``almost monotonicity'',
in the sense that flipping any coordinate from $0$ to $1$ cannot decrease the value of the function too
much. To define this relaxation more precisely we need the notion of negative influences:
\begin{definition}\label{def:neg_inf}
  Let $f\colon(\power{n},\mu_p)\to [0,1]$ and let $i\in [n]$. The negative influence of
  a variable $i$ on $f$, denoted by $I_i^{-}[f]$, is defined to be
  \[
  \Expect{{\bf x}\sim\mu_p}{\max(0, f({\bf x}_1,\ldots,{\bf x}_{i-1},0,{\bf x}_{i+1},\ldots,{\bf x}_n) - f({\bf x}_1,\ldots,{\bf x}_{i-1},1,{\bf x}_{i+1},\ldots,{\bf x}_n))}.
  \]
\end{definition}
\noindent (Note that whereas $I_i[f]$ is the average of \emph{squared differences}, $I_i^-[f]$ is an average of \emph{differences}.)

With this definition, Theorem~\ref{thm:mono_and} also holds
when we relax the condition of monotonicity of $f$ to the condition
that all of its individual negative influences are small, i.e.\ $I_i^{-}[f]\leq \eta$
for all $i\in [n]$ (the proof of Theorem~\ref{thm:mono_and} in Section~\ref{sec:almost_mono} achieves this
stronger statement). One benefit of this relaxation is that it is able
to capture the case of ``judgement aggregation'' as an immediate consequence.
\begin{thm}\label{thm:and_homomorphism}
  For all $\zeta,\eps>0$ there is $\eta>0$ such that the following holds for all $p,\rho\in[\zeta,1-\zeta]$.
  If $f\colon(\power{n},\mu_{\rho p})\to\power{}$,
  $g\colon(\power{n},\mu_{p})\to\power{}$ and
  $h\colon(\power{n},\mu_{\rho})\to\power{}$ satisfy
  $\Prob{{\sf x}\sim\mu_p, {\bf y}\sim\mu_{\rho}}{f({\bf x}\land {\bf y}) = g({\bf x})\land h({\bf y})}\geq 1-\eta$, then one of the following cases must happen.
  \begin{enumerate}
    \item $f$ and at least one of the functions $g$ or $h$ are $\eps$-close to the constant $0$ function.
    \item There is a set $T\subseteq[n]$ such that $f,g,h$ are all $\eps$-close to ${\sf AND}_T$ (each with respect to their input distribution).
  \end{enumerate}
\end{thm}

\subsection{One-sided error}
Finally, we consider a more relaxed version of approximate solutions to $\Tdown{p}{\rho p} f = g$.
We say functions $f\colon\power{n}\to[0,1]$ and $g\colon\power{n}\to\power{}$ are one-sided error solutions
with $\lambda>0$ and error $\eta$ if the following two conditions occur:
\begin{enumerate}
  \item $\Tdown{p}{\rho p} f$ is very small on typical inputs $x$ such that $g(x) = 0$:
  \[
  \Expect{{\bf x}\sim \mu_p}{(1-g({\bf x})) \cdot\Tdown{p}{\rho p} f({\bf x})}\leq \eta.
  \]
  \item $\Tdown{p}{\rho p} f$ is bounded away from $0$ on typical inputs $x$ such that $g(x) = 1$:
  \[
  \Prob{{\bf x}\sim\mu_p}{g({\bf x}) = 1, \Tdown{p}{\rho p} f({\bf x})\leq \lambda}\leq \eta.
  \]
\end{enumerate}
\begin{thm}\label{thm:one_sided_error}
  For any $\eps,\zeta>0$ there are $\eta>0$ and $J\in\mathbb{N}$
  such that the following holds for any $p,\rho\in[\zeta,1-\zeta]$
  and $\lambda\in[\zeta,1]$. If $f\colon\power{n}\to[0,1]$ and $g\colon\power{n}\to\power{}$ are one-sided error
  solutions with $\lambda$ and error $\eta$, then $g$ is $\eps$-close to a monotone, Boolean $J$-junta.
\end{thm}
We remark that any monotone junta $g$ is a one-sided error approximate solution (by taking $f=g$),
so Theorem~\ref{thm:one_sided_error} is tight with respect to the structure of $g$.

\paragraph{Organization.}
The proof of Theorem~\ref{thm:main_2func} is given in Section~\ref{sec:main_2func}.
In Section~\ref{sec:special} we prove Theorems~\ref{thm:basic},~\ref{thm:mono_and} and~\ref{thm:and_homomorphism}.
In Section~\ref{sec:one_sided} we prove Theorem~\ref{thm:one_sided_error}, and finally in Section~\ref{sec:large_noise}
we prove Theorem~\ref{thm:main_2func_smallrho}.

\section{Preliminaries} \label{sec:prel}
For any $p\in (0,1)$, we consider the space of function $f\colon(\power{n},\mu_p)\to\mathbb{R}$
equipped with the inner product $\inner{f}{g} = \Expect{{\bf x}\sim\mu_p}{f({\bf x})g({\bf x})}$. We will
use the Fourier--Walsh orthonormal basis $\set{\chi^p_S}_{S\subseteq [n]}$,
where for each $S\subseteq[n]$ we define $\chi^p_S\colon\power{n}\to\mathbb{R}$ by
$\chi^p_S(x) = \prod\limits_{i\in S}\left[(x_i - p)/\sqrt{p(1-p)}\right]$. This way,
we may write the Fourier expansion of a function $f\colon\power{n}\to\mathbb{R}$
by
\[
f(x) = \sum\limits_{S\subseteq[n]}{\widehat{f}_p(S) \chi^p_S(x)}, \qquad\qquad \text{where }\widehat{f}_p(S) = \inner{f}{\chi^p_S}.
\]

Since $\set{\chi^p_S}_{S\subseteq [n]}$ is an orthonormal basis, we have Parseval's identity $\norm{f}_2^2 = \sum\limits_{S\subseteq[n]}\widehat{f}_p(S)^2$.
We will need a few more notions and results from Fourier analysis, such as the Junta Theorems of \cite{Bourgain,KS} and the Sensitivity Conjecture proved
recently by \cite{Huang}, which we present below.

%\paragraph{Noise operator.}
%For $\rho\in [0,1]$, the $\rho$-correlated distribution over $(x,y)\in \power{n}\times\power{n}$
%is defined by taking on each coordinate $i\in[n]$ independently, with probability $\rho$
%$x_i=y_i$ is a uniformly chosen bit, and otherwise $x_i$ and $y_i$ are uniformly chosen independent bits.
%The noise operator $\T_{\rho}$ acts on functions $f\colon\power{n}\to\mathbb{R}$ by
%\[
%\T_{\rho} f(x) = \Expect{(x,y)\sim \rho\text{ correlated}}{f(y)}.
%\]

\subsection{Influences}
For a function $f\colon (\power{n},\mu_p)\to\mathbb{R}$ and a coordinate $i\in[n]$,
we define the $p$-biased influence of variable $i$ to be $I^p_i[f] = \Expect{{\bf x}\sim\mu_p}{(f({\bf x}) - f({\bf x}\xor e_i))^2}$.
When the bias parameter is clear from context, we often write $I_i[f]$.

We will also use the notion of negative influences as given in Definition~\ref{def:neg_inf}. We have the following simple
fact, stating that averaging may only decrease negative influences.
\begin{fact}\label{fact:avg_dec_neg_inf}
  Let $f\colon(\power{n},\mu_p)\to\mathbb{R}$ be a function,
  and let $i\in [n]$. Consider the function $g\colon(\power{n-1},\mu_p)\to\mathbb{R}$
  defined by $g(z) = \cExpect{{\bf x}\sim\mu_p}{{\bf x}_{[n]\setminus\set{i}} = z}{f({\bf x})}$
  (i.e.\ averaging $f$ over the coordinate $i$). Then $I_j^{-}[g]\leq I_j^{-}[f]$ for
  any $j\in [n]\setminus \set{i}$.
\end{fact}
\begin{proof}
  Fix $j\neq i$, and assume without loss of generality that
  $j=n-1$ and $i=n$. We prove that for each $x\in\power{n-2}$,
  the contribution of the edge between $(x,0)$ to $(x,1)$ to
  the negative influence of $g$ is upper-bounded by the contribution
  of the parallelepiped of $x$ in $f$.

  Denote
  $a = f(x,0,0)$, $b = f(x,0,1)$, $c = f(x,1,0)$, $d = f(x,1,1)$,
  $\bar{a} = (1-p)a + p b$ $(=g(x,0))$,
  $\bar{c} = (1-p)c + p d$ $(=g(x,1))$.
  If $\bar{c} > \bar{a}$, the point $x$
  does not contribute to the negative influence of
  $g$ and there is nothing to prove. Otherwise, the contribution
  is
  $\mu_p(x)(\bar{a} - \bar{c})
  =\mu_p(x,0)(a-c) + \mu_p(x,1)(b-d)
  \leq \mu_p(x,0)\max(a-c,0) + \mu_p(x,1)\max(b-d,0)$,
  and the right-hand side is exactly the contribution of
  the parallelepiped of $x$ in $f$.
\end{proof}

We also need the following fact that relates negative influences and
distance from monotonicity.
\begin{fact}\label{fact:gglrs}
  For all $p\in(0,1)$, $n\in\mathbb{N}$ and $\tau > 0$,
  if $f\colon(\power{n},\mu_p)\to\mathbb{R}$ is a function
  such that $I_i^{-}[f]\leq \tau$ for all $i\in [n]$,
  then there is a monotone function $h\colon(\power{n},\mu_p)\to\mathbb{R}$ such that
  $\norm{f-h}_1\leq ((1-p)p)^{-n}n\tau$.
\end{fact}
We remark that the above fact is inspired by
\cite{GGLRS}, wherein a similar statement was proved for Boolean functions for $p=1/2$,
with a better bound ($n\tau$). The argument we present is essentially
the same and can also recover the bound $n\tau$ for $p=1/2$, but since
we will only use this statement with constant $n$ and very small $\tau$, the
bound we achieve is sufficient for our purposes. We defer the proof to the end of the preliminaries.

\subsection{Junta theorems}
We will use Bourgain's Theorem~\cite{Bourgain}; the sharp version below is proved in~\cite{KKO}.
For $k\in\mathbb{N}$, the Fourier tail $W_{\geq k}[f]$ is defined to be $\sum\limits_{\card{S}\geq k}{\widehat{f}_p(S)^2}$.

\begin{thm}\label{thm:Bourgain}
  For any $\zeta>0$ there is a constant $C(\zeta)>0$ such that for any $k\in\mathbb{N}$,
  $\eps>0$
  there are $\tau = (k/\eps)^{-C\cdot k}$, $J = (k/\eps)^{C\cdot k}$ such that the following holds for all $p\in[\zeta,1-\zeta]$.
  If $f\colon(\power{n},\mu_p)\to\power{}$ satisfies $W_{\geq k}[f] \leq \frac{\eps}{C\sqrt{k}\log^{1.5}(k)}$,
  then $f$ is $\eps$-close to a $J$-junta $h$.

  Furthermore, $h$ only depends on variables $i$ such that $I_i[g]\geq \tau$.
\end{thm}

We also need the following result of Kindler and Safra. We remark that its important feature
at is has (lacking from Bourgain's result above) is that the size of the junta only depends
on the level $k$ and not on the closeness parameter $\eps$ that we wish to get.

\begin{thm}[\cite{KS}]\label{thm:KS}
  For any $\zeta > 0$, $m\in \mathbb{N}$ there exists $J(m,\zeta)\in\mathbb{N}$, $C(m,\zeta) > 0$ such that the following holds
  for all $p\in[\zeta,1-\zeta]$.
  For any $\eps>0$ there exists $\delta = C(m,\zeta)\cdot \eps$ such that if
  $f\colon(\power{n},\mu_p)\to\power{}$ is a function such that $W_{\geq m}[f]\leq \delta$,
  then $f$ is $\eps$-close to a junta of size $J(m,\zeta)$.
\end{thm}

\subsection{Degree and sensitivity}
For any $f\colon\power{n}\to\power{}$ and $x\in\power{n}$, the sensitivity of
$f$ at $x$ is equal to the number of coordinates $i\in [n]$ such that $f(x)\neq f(x\xor e_i)$.
The max-sensitivity of a function $f$ is $s(f) = \max_{x} s(f,x)$. The degree of a function
${\sf deg}(f)$ is the maximal size of $S$ such that $\widehat{f}_p(S)\neq 0$ (we remark that
this is easily seen to be independent of $p$).

We will use the following recent result of Huang \cite{Huang}
(formerly known as the sensitivity conjecture \cite{NisanSzegedy}) in our proof.
We remark that quantitatively weaker results there were proven earlier (such as the bound $s(f)\geq \Omega(\log ({\sf deg}(f)))$) would have been enough for us,
but yield to a loss in several parameters.
\begin{thm}[\cite{Huang}]\label{thm:sensitivity}
  For any $f\colon \power{n}\to\power{}$ we have that $s(f) \geq \sqrt{{\sf deg}(f)}$.
\end{thm}

\subsection{Proof of Fact~\ref{fact:gglrs}}\label{sec:proof_gglrs}
For each $i\in [n]$, define the shifting operator $S_i$ on functions $g\colon\power{n}\to\mathbb{R}$
that operates by $(S_i f)(x_{-i},x_i) = \max(f(x_{-i},0),f(x_{-i},1))$ if $x_i=1$
and $(S_i f)(x_{-i},x_i) = \min(f(x_{-i},0),f(x_{-i},1))$ if $x_i=0$.

\begin{claim}\label{claim:inf_remain_small}
  For any $i,j\in[n]$ and $g\colon\power{n}\to\mathbb{R}$
  we have that $I_j^{-}[S_i g]\leq \frac{1}{p(1-p)} I_j^{-}[g]$.
\end{claim}
\begin{proof}
  If $i = j$ the claim is clear as $I_i^{-}[S_i g]=0$, so we assume that $i\neq j$.
  Without loss of generality $j=n-1, i=n$. We show that for every $x\in\power{n-2}$,
  the contribution of the parallelepiped of $x$ in $S_i g$ is at most $\frac{1}{p(1-p)}$
  its contribution in $g$.

  Write $a = g(x,0,0)$, $b = g(x,0,1)$, $c = g(x,1,0)$, $d = g(x,1,1)$,
  and note that the value of $S_i g$ on these points, in the same order,
  is $\min(a,b),\max(a,b),\min(c,d),\max(c,d)$. The contribution of $x$
  to the parallelepiped to $S_i g$ is
  \[
  \mu_p(x)\left[(1-p)\max(\max(a,b) - \max(c,d),0) + p\max(\min(a,b) - \min(c,d),0)\right].
  \]
  Using $\max(a,b)-\max(c,d)\leq \max(a-c,b-d)$ and
  $\min(a,b)-\min(c,d)\leq \max(a-c,b-d)$, we conclude
  that the contribution of the parallelepiped of $x$ to $S_i g$
  is at most $\mu_p(x)\max(a-c,b-d,0)$.

  On the other hand, the contribution of the parallelepiped of $x$ to $g$
  is
  \[
  \mu_p(x)\left[(1-p)\max(a-c,0) + p\max(b-d,0)\right]
  \geq p(1-p)\mu_p(x)\max(a-c,b-d,0),
  \]
  and we are done.
\end{proof}

\begin{claim}\label{claim:dist}
  For any $i\in[n]$ and $g\colon\power{n}\to\mathbb{R}$.
  We have that $\norm{S_i g - g}_1\leq I_i^{-}[g]$.
\end{claim}
\begin{proof}
  Assume without loss of generality that $i=n$.
  We show that for each $x\in\power{n-1}$, the contribution of
  the points $(x,0), (x,1)$ to the left-hand side is at most
  the contribution of the edge between them on the right-hand side.

  Write $a = g(x,0)$, $b = g(x,1)$,
  and note that the value of $S_i g$ on these points, in the same order,
  is $\min(a,b),\max(a,b)$. Consider the contribution of $(x,0)$ and $(x,1)$ to the
  left-hand side; if $b\geq a$ it is $0$ and there is nothing to prove, so assume $b<a$.
  Then the contribution is equal to $\mu_p(x) \card{b-a}$, which is the same as
  the contribution of this edge to $I_i^{-}[g]$, and we are done.
\end{proof}

We are now ready to prove Fact~\ref{fact:gglrs}.
\begin{proof}[Proof of Fact~\ref{fact:gglrs}]
  Let $g$ be a function as in the proof, and define $h_i = S_i\circ S_{i-1}\circ\dots\circ S_1 g$ for each $0\leq i\leq n$
  (where $h_0 = g$). Clearly $h_n$ is monotone, and we show that it is close to $g$.
  By Claim~\ref{claim:inf_remain_small} all negative influences of each $h_i$
  are at most $((1-p)p)^{-n}\tau$, and hence by Claim~\ref{claim:dist} we get that
  $\norm{h_{i} - h_{i-1}}_1\leq ((1-p)p)^{-n}\tau$ for each $i=1,\ldots,n$. By the triangle inequality
  we get that $\norm{g - h_n}_1 = \norm{h_0 - h_n}_1\leq ((1-p)p)^{-n}n\tau$.
\end{proof}

\section{Proof of Theorem~\ref{thm:main_2func}}\label{sec:main_2func}
In this section, we prove Theorem~\ref{thm:main_2func}. Since we will always consider
the downwards noise operator $\Tdown{p}{p/2}$, we denote it succinctly by $\T$.

\subsection{Main Lemma}\label{sec:exact_rho_half}
\begin{lemma}\label{lemma:main_exact}
  For any $\zeta>0$ and $n\in\mathbb{N}$
  there exists $\eta_0>0$ such that the following holds for all $p\in[\zeta,1-\zeta]$, $\lambda \in[\zeta,1]$
  and $\eta\in[0,\eta_0]$.
  If $\norm{\Tdown{p}{p/2} f - \lambda g}_{\infty}\leq \eta$
  then:
  \begin{itemize}
    \item $g$ is an AND-OR function of width $r$, where $r\leq \ceil{\log(2/\zeta)}$.
    \item Let $\phi$ be the corresponding AND-XOR function.
    Then $\norm{f - 2^r\lambda\cdot \phi}_\infty \leq 3^n \eta$.
  \end{itemize}
\end{lemma}
This section is devoted for the proof of this lemma, and the proof is
divided into several claims. It will be convenient for us to identify vectors
in $\power{n}$ with subsets of $[n]$ by identifying a vector with its support,
and consequently think of the inputs of functions as subsets of $[n]$.
The definition of the operator $\T$ to these language is immediate:
$\T f(B) = \Expect{{\bf C}\subseteq B}{f({\bf C})} = 2^{-\card{B}}\sum\limits_{C\subseteq B} f(C)$.

% N.B. Changed $m$ from $\ceil{\log(1/\lambda)}$ to $\floor{\log(2/\lambda)}$.
Fix $\zeta,n$, and choose $\eta = \zeta 4^{-n^2-4n-4}$. Let
$f,g$ be functions as in the statement of the lemma.

\begin{claim}\label{claim:g_mono_rho_half}
  $g$ is monotone.
\end{claim}
\begin{proof}
  Suppose $g$ is not monotone. Then there is an edge $(A,B)$ of the hypercube
  where $A\subsetneq B$ such that $g(A) = 1$, $g(B)=0$. We have that $\T f(B)\leq \lambda g(B) + \eta = \eta$,
  which by definition of $\T$ implies that $\Expect{{\bf C}\subseteq B}{f({\bf C})} \leq \eta$. Denote $\set{i} = B\setminus A$,
  and note that with probability $1/2$ we have $i\not\in {\bf C}$, in which case ${\bf C}\subseteq A$.
  Thus, the non-negativity of $f$ implies that $\Expect{{\bf C}\subseteq A}{f({\bf C})} \leq 2\Expect{{\bf C}\subseteq B}{f({\bf C})}\leq 2\eta$,
  i.e. $\T f(A)\leq 2\eta$. This is in contradiction to $\T f(A)\geq \lambda g(A) - \eta = \lambda - \eta$ (by the choice of $\eta$).
\end{proof}

Since $g$ is monotone, one can discuss its minterms, i.e.\ sets $M\subseteq [n]$ such that $g(M) = 1$ but for all $A\subsetneq M$, $g(A) = 0$.
The following lemma asserts that the value of $f$ on any minterm of $g$ is determined (up to a small error).
\begin{claim}\label{claim:value_in_minterm}
  If $M$ is a minterm of $g$, then $\card{f(M) - \lambda 2^{\card{M}}}\leq 4^{\card{M}}\eta$.
\end{claim}
\begin{proof}
  Since $g(M)=1$, we have that $\card{\T f(M) - \lambda}\leq \eta$, and by definition of
  $\T$ we have $\T f(M) = 2^{-\card{M}}\sum\limits_{A\subseteq M} f(A)$, so by the triangle
  inequality it follows that $\card{f(M) - 2^{\card{M}}\lambda}\leq 2^{\card{M}}\eta + \sum\limits_{A\subsetneq M} f(A)$
  , and to finish the proof, we upper bound the last sum. Note that for every $A\subsetneq M$, choosing
  ${\bf B}\subsetneq M$ randomly of size $\card{M}-1$, we have that $A\subseteq {\bf B}$ with probability at least $1/\card{M}$,
  hence by the non-negativity of $f$ there is $B$ of size $\card{M}-1$ such that $\sum\limits_{A \subsetneq M} f(A)\leq \card{M}\sum\limits_{A\subseteq B} f(A)$,
  and we fix such $B$.\footnote{Alternatively, note that $\sum_{A \subsetneq M} f(A) \leq \sum_{|B|=m-1} \sum_{A \subseteq B} f(A)$, and take $B$ maximizing $\sum_{A \subseteq B} f(A)$.}
  Since $B\subsetneq M$ and $M$ is a minterm of $g$, we have that $g(B)=0$, and therefore
  $\T f(B)\leq \eta$ or equivalently $\sum\limits_{A \subseteq B} f(A)\leq 2^{\card{B}} \eta$. Plugging that in we get
  that $\sum\limits_{A \subsetneq M} f(A)\leq \card{M}2^{\card{M}-1} \eta$ and the claim follows.
\end{proof}

We next wish to argue all minterms of $g$ are of the same size, and towards this end
(and also in other places in the argument) the following proposition will be useful.
\begin{proposition}\label{prop:below_1}
  Let $B,Z\subseteq[n]$ be disjoint such that $g(B) = 1$. Then
  \[
  \card{\sum\limits_{A\subseteq B}{f(A\cup Z)} - 2^{\card{B}}\lambda}\leq 2^{\card{B}} 3^{\card{Z}}\eta.
  \]
\end{proposition}
\begin{proof}
  For any $W\subseteq Z$, we have that $\card{\T f(B\cup W) - \lambda g(B\cup W)}\leq \eta$. Since $g$ is monotone
  and $g(B) = 1$, we must have $g(B\cup Y) = 1$ and we get that
  \[
    \card{\sum\limits_{A\subseteq B, Y\subseteq W} f(A\cup Y) - 2^{\card{B\cup W}}\lambda}\leq 2^{\card{B\cup W}}\eta.
  \]
  Note that for any $Y\subseteq Z$, $\sum\limits_{W\colon Y\subseteq W\subseteq Z}(-1)^{\card{Z\setminus W}} = 0$
  unless $Y=Z$, in which case the sum is $1$, and so we get that
  \[
  \sum\limits_{A\subseteq B} f(A\cup Z)
  =\sum\limits_{\substack{A\subseteq B, Y\subseteq Z\\ W\colon Y\subseteq W\subseteq Z}}(-1)^{\card{Z\setminus W}} f(A\cup Y)
  =\sum\limits_{W\subseteq Z}(-1)^{\card{Z\setminus W}}\sum\limits_{A\subseteq B, Y\subseteq W} f(A\cup Y).
  \]
   Therefore
   the triangle inequality implies that
   \[
   \card{\sum\limits_{A\subseteq B} f(A\cup Z) - \lambda\sum\limits_{W\subseteq Z}(-1)^{\card{Z\setminus W}} 2^{\card{B\cup W}}}
   \leq \sum\limits_{W\subseteq Z}\card{\sum\limits_{A\subseteq B, Y\subseteq W} f(A\cup Y)- 2^{\card{B\cup W}}\lambda}
   \leq \sum\limits_{W\subseteq Z}2^{\card{B\cup W}}\eta,
   \]
   which is equal to $2^{\card{B}} 3^{\card{Z}}\eta$. To complete the proof, we observe that by the binomial
   formula
   \[
   \sum\limits_{W\subseteq Z}(-1)^{\card{Z\setminus W}} 2^{\card{B\cup W}}
   =2^{\card{B}}\sum\limits_{W\subseteq Z}{2^{\card{W}} (-1)^{\card{Z} - \card{W}}} = 2^{\card{B}}(2-1)^{\card{Z}} = 2^{\card{B}}.
   \qedhere
   \]
\end{proof}

We now show two consequences of the above proposition. First,
we show that all minterms of $g$ have the same size.
\begin{claim}
  Let $M,M'$ be two minterms of $g$. Then $\card{M} = \card{M'}$.
\end{claim}
\begin{proof}
  Let $Z=M'\setminus M$. By Proposition~\ref{prop:below_1} we have
  \[
  \sum\limits_{A\subseteq M}{f(A\cup Z)}\leq 2^{\card{M}}\lambda + 2^{\card{M}} 3^{\card{Z}}\eta.
  \]
  By the non-negativity of $f$ the left-hand side is at least the value of $f$ for $A=M\cap M'$, i.e.\ on $A\cup Z = M'$;
  furthermore, by Claim~\ref{claim:value_in_minterm} we have $f(M')\geq \lambda 2^{\card{M'}} - 4^{\card{M'}}\eta$,
  so combining we get
  \[
  \lambda 2^{\card{M'}}\leq 2^{\card{M}}\lambda + \eta(2^{\card{M}} 3^{\card{Z}}+4^{\card{M'}})
  \leq 2^{\card{M}}\lambda + \lambda/2,
  \]
  where the last inequality is by the choice of $\eta$. This implies that $\card{M'}\leq \card{M}$. The second inequality
  is proved analogously.
\end{proof}

Denote the size of a minterm of $g$ by $m$, and note that $m\leq \ceil{\log(2/\lambda)}$.
Indeed, letting $M$ be any minterm of $g$, by Claim~\ref{claim:value_in_minterm} we get that
$\lambda 2^m\leq f(M) + 4^m\eta\leq 2$.

We next show that the value of $f$ in a point $B$ must be either close to $0$ or
close to $2^{m}\lambda$.
\begin{claim}\label{claim:f_2_val}
  For any $B\in\power{n}$, either $f(B)\leq 4^{\card{B}^2}\eta$ or $\card{f(B) - 2^{m} \lambda}\leq 4^{\card{B}}\eta$.
\end{claim}
\begin{proof}
  If $g(B) = 0$, the claim is immediate since $2^{-\card{B}} f(B) \leq \T f(B)\leq \lambda g(B) + \eta = \eta$,
  so we assume that $g(B) = 1$. We prove the claim by induction on $\card{B}$. If $\card{B} = m$, then
  $B$ is a minterm and the claim follows from Claim~\ref{claim:value_in_minterm}. Assume the claim holds for
  all $\card{B}\leq i$, and let $B$ be of size $i+1$. Since $g(B)=1$ we get that there must be a minterm
  $M\subseteq B$ of $g$. Let $Z = B\setminus M$, then by Proposition~\ref{prop:below_1},
  \[
  \card{\sum\limits_{A\subseteq M}{f(A\cup Z)} - 2^{m}\lambda}\leq 2^{m} 3^{\card{Z}}\eta.
  \]
  For each $A\subsetneq M$, by the induction hypothesis $f(A\cup Z)$ is either close to $0$ (i.e.\ at most $4^{\card{A\cup Z}^2}\eta$)
  or close to $2^m\lambda$ (more precisely, up to $\pm 4^{\card{A\cup Z}}\eta$).
  \begin{itemize}
    \item If there is $A^{\star}\subsetneq M$ that falls into
  the second case, then by non-negativity of $f$ we get
  \begin{align*}
  f(B) = f(M \cup Z)
  &\leq \sum\limits_{A\subseteq M}{f(A\cup Z)} - f(A^{\star}\cup Z)\\
  &\leq \card{\sum\limits_{A\subseteq M}{f(A\cup Z)} - 2^m\lambda}
  +\card{f(A^{\star}\cup Z)-2^m\lambda}\\
  &\leq 2^{m} 3^{\card{Z}}\eta + 4^{\card{A^{\star}\cup Z}}\eta
  \leq 4^{\card{B}}\eta,
  \end{align*}
  where in the last inequality we used $\card{A^{\star}}\leq m-1$, and the claim is proved for $B$.
    \item Otherwise, by the triangle inequality
  \begin{align*}
    \card{f(M\cup Z) - 2^m\lambda}
    &\leq \card{\sum\limits_{A\subsetneq M} f(A\cup Z)} + \card{\sum\limits_{A\subseteq M} f(A\cup Z) - 2^{m}\lambda}\\
    &\leq \sum\limits_{A\subsetneq M} 4^{\card{A\cup Z}^2} + 2^{m} 3^{\card{Z}}\eta\\
    &\leq \left(4^{(m+\card{Z}-1)^2 + m} + 3^{m+\card{Z}}\right)\eta,
  \end{align*}
  which is at most $4^{(m+\card{Z})^2}\eta = 4^{\card{B}^2}\eta$. Hence the claim is proved for
  $B$ (as $B = M\cup Z$).
  \qedhere
  \end{itemize}
\end{proof}

We are now able to restate Proposition~\ref{prop:below_1} in a more convenient form.
For each pair of disjoint sets $B,Z\subseteq[n]$ such that $g(B) = 1$,
denote $X(B,Z) = \sett{ A\cup Z }{A\subseteq B}$.
\begin{corollary}\label{corr:prop_reformulate}
  Suppose $B,Z\subseteq[n]$ are disjoint and $g(B) = 1$.
  Then there is a unique $A^{\star}\subseteq B$ such that:
  \begin{itemize}
  \item $\card{f(A^{\star}\cup Z) - 2^m\lambda}\leq 4^{n}\eta$.
  \item For any other $A\subset B$ we have that $f(A\cup Z)\leq 4^{n^2}\eta$.
  \item $g(A^{\star}\cup Z) = 1$ and for any $A\subsetneq A^{\star}$ we have $g(A\cup Z) = 0$.
  \end{itemize}
\end{corollary}
\begin{proof}
  For the first item, if for all $A\subseteq B$ it holds that $f(A\cup Z)\leq 4^{n^2}\eta$, then
  by Proposition~\ref{prop:below_1} we have
  $2^m\lambda\leq \sum\limits_{A\subseteq B}{f(A\cup Z)} + 6^n\eta\leq 4^{n^2+3n}\eta$,
  which contradicts the choice of $\eta$. Therefore, by Claim~\ref{claim:f_2_val}
  there is $A^{\star}\subseteq B$ such that $\card{f(A^{\star}\cup y)-2^m\lambda}\leq 4^n\eta$.

  For the second item, assume towards contradiction there are two such $A_1, A_2$.
  By Proposition~\ref{prop:below_1} we have
  \[
  2\cdot(2^m\lambda - 4^n\eta)\leq f(A_1\cup Z) + f(A_2\cup Z)\leq \sum\limits_{A\subseteq B}{f(A\cup Z)} \leq 2^m\lambda + 6^n\eta,
  \]
  and therefore $2^m\lambda \leq 6^{n+1}\eta$, which is a contradiction to the choice of $\eta$.

  For the third item, note that
  \[
  g(A^{\star}\cup Z)
  \geq \T f(A^{\star}\cup Z) - \eta
  \geq 2^{-\card{A^{\star} \cup Z}}f(A^{\star}\cup Z)-\eta
  \geq 2^{-n}\lambda - (4^n + 1)\eta > 0,
  \]
  and since $g$ is Boolean-valued it follows that $g(A^{\star} \cup Z) = 1$. Also,
  for any $A\subsetneq A^{\star}$ we have
   \[
    g(A\cup Z)\leq \T f(A\cup Z) + \eta
    \leq 4^{n^2}\eta + \eta < 1,
  \]
  where in the second inequality we used the definition of $\T$ and the second item. Since
  $g$ is Boolean we get that $g(A\cup Z) = 0$.
\end{proof}
To simplify notation, for the rest of the section we often say ``the value of $f(S)$ is close to $2^m\lambda$''
to express that $\card{f(S) - 2^m\lambda}\leq 4^{n}\eta$ and ``the value of $f(S)$ is close to $0$'' to express that $f(S)\leq 4^{n^2}\eta$.

\skipi

Consider the $m$-uniform hypergraph $H = ([n],E)$ whose edges are the minterms of $g$. In the remainder of this
section we show that $H$ is a complete $m$-partite hypergraph, which is easily seen to be
equivalent to $g$ being an ANR-OR function of width $m$. Towards this end, we will define a coloring
$\chi\colon [n]\to[m]$ and show that (a) each edge $e\in E$ is rainbow colored (i.e.\ no two vertices
in it are colored in the same color), and (b) any rainbow colored set
$A\subseteq[n]$ is an edge.

Fix a minterm $B\subseteq [n]$, and write $B = \set{b_1,\ldots, b_m}$, where $b_1,\ldots,b_m\in [n]$.
We define $\chi(b_i) = i$. We now define $\chi(v)$ for any $v\in [n]\setminus B$. Fix $v\in [n]\setminus B$,
and consider the set $X(B, \set{v})$; by Corollary~\ref{corr:prop_reformulate} there exists a unique $A\subseteq B$
such that $f(A\cup\set{v})$ is close to $2^m\lambda$, and its $g$-value is $1$. Since $g(A\cup\set{v})=1$,
we must have $\card{A\cup\set{v}}\geq m$,
and there are two options:
\begin{itemize}
  \item If $\card{A\cup\set{v}} = m+1$, i.e. $A = B$, define $\chi(v) = \bot$.
  \item Otherwise, there is $i\in[m]$ such that $A = B\setminus \set{b_i}$,
  and we define $\chi(v) = \chi(b_i) = i$.
\end{itemize}

We first show that each minterm of $g$ is colored using only elements from $[m]$
(as opposed to $\bot$).
\begin{claim}
  Let $M\in E$ be a minterm of $g$. Then for each $v\in M$ we have $\chi(v)\neq \bot$.
\end{claim}
\begin{proof}
  Assume towards contradiction that this is not the case, and let $v\in M$ be such that $\chi(v) = \bot$.
  Then by definition of $\chi$ this means that $f(B\cup\set{v})$ is close to $2^m\lambda$, and since $B$ is a minterm of $g$
  we also know, by Claim~\ref{claim:value_in_minterm}, that $f(B)$ is close to $2^m\lambda$.
  This gives us two points in $X(M, B\setminus M)$ whose $f$-value is close to $2^m \lambda$,
  in contradiction to Corollary~\ref{corr:prop_reformulate}.
\end{proof}

Next, we show that each minterm of $g$ is rainbow colored by $\chi$.
\begin{claim}\label{claim:minterm_rainbow}
  Let $M\in E$ be a minterm of $g$. Then $M$ is rainbow colored.
\end{claim}
\begin{proof}
  Write $M = \set{v_1,\ldots,v_m}$, and assume towards contradiction the statement is false.
  Then there are $v_i,v_j$ that are assigned the same color by $\chi$, and without loss of generality
  we may assume $\chi(v_1) = \chi(v_2) = m$. By definition of $\chi$ it follows that
  $f(\set{v_1}\cup (B\setminus\set{b_m}))$ and $f(\set{v_2}\cup(B\setminus\set{b_m}))$ are both close to $2^m\lambda$.
  However, note that these are two distinct points in $X(M,B\setminus M)$, and thus we get a contradiction
  to the second item in Corollary~\ref{corr:prop_reformulate}.
\end{proof}
Note that the definition of the coloring $\chi$ may depend on the minterm $B$ chosen initially to define
it. The following claim shows that this is actually not the case --- and more precisely that if we use a
different minterm $B'$ to define a coloring $\chi'$, then there is a permutation $\pi$ on $[m]$
such that $\chi' = \pi\circ\chi$.

\begin{claim}\label{claim:coloring_independent}
  Let $B' = \set{b_1',\ldots,b_m'}$ be any minterm of $g$, and let
  $\chi'$ be a coloring defined as above using $B'$ in place of $B$.
  Then there exists $\pi\in S_m$ such that $\chi = \pi\circ \chi'$.
\end{claim}
\begin{proof}
  Since $B$ is a minterm of $g$, it follows by Claim~\ref{claim:minterm_rainbow} that
  it is rainbow colored by both $\chi$ and $\chi'$, so we define $\pi\in S_m$ by $\pi(\chi'(b_i)) = \chi(b_i)$.
  We define $\tilde{\chi} = \pi\circ \chi'$ and show that $\chi = \tilde{\chi}$.

  Let $v\not\in B$, and assume without loss of generality $\chi(v) = m$.
  Then by definition of $\chi$ we must have that $\set{b_1,\ldots,b_{m-1},v}$
  is a minterm of $g$, and hence by Claim~\ref{claim:minterm_rainbow} it must be
  rainbow colored by $\tilde{\chi}$. Since $\tilde{\chi}$ agrees with $\chi$ on
  $b_1,\ldots,b_{m-1}$, we must have $\tilde{\chi}(v) = m$, and we are done.
\end{proof}

Lastly, we show that each rainbow colored set of size $m$ is a minterm of $g$.
\begin{claim}
  For any minterm $B$ of $g$ and a coloring $\chi$ defined by it, the following holds.
  If $M\subseteq[n]$ of size $m$ is rainbow colored by $\chi$, then $g(M)=1$.
  Consequently, $M$ is a minterm of $g$.
\end{claim}
\begin{proof}
  We prove the statement for all $B,\chi,M$ by induction on $\card{B\cap M}$.

  Write $M = \set{v_1,\ldots,v_m}$, $B = \set{b_1,\ldots,b_m}$, and assume without loss of generality that $\chi(v_i) = \chi(b_i)$.
  The base case is $\card{B\cap M} = m$, in which case $M=B$ and the
  claim is obvious.

  Let $k\leq m-1$, assume the statement is true whenever $\card{B\cap M}\geq k+1$, and let $M$ be such that $\card{B\cap M} = k$.
  Without loss of generality we may assume that $v_i = b_i$ for all $1\leq i\leq k$. Since $\chi(v_{k+1}) = \chi(b_{k+1})$
  we get that $B' = (B\setminus\set{b_{k+1}})\cup\set{v_{k+1}}$ is a minterm of $g$. Let $\chi'$ be the coloring defined by
  $B'$.
  By Claim~\ref{claim:minterm_rainbow} we get that $\chi' = \pi\circ\chi$ for some $\pi\in S_m$,
  and in particular as $M$ is rainbow colored by $\chi$ it is also rainbow colored by $\chi'$.
  Since $\card{B' \cap M} = k+1$ we may apply the induction hypotehsis on $B'$ with the coloring $\chi'$
  to conclude that $M$ is a minterm of $g$, as required.
\end{proof}
It follows that the function $g$ is the function $\bigwedge_{i\in[m]}\bigvee_{j\in A_i} x_j$
where $A_i = \chi^{-1}(i)$. To complete the proof of Lemma~\ref{lemma:main_exact} we
must establish the structural result for $f$, which we do by inverting $f$.
\begin{claim}\label{claim:aux_mobius}
The operator $\T$ has an inverse $\T^{-1}$ given by
$\T^{-1} h(B) = \sum\limits_{A\subseteq B}{(-1)^{\card{B\setminus A}} 2^{\card{A}}h(A)}$.
\end{claim}
\begin{proof}
Let $h$ be in the image of $\T$,
i.e.\ $h(B) = \T \tilde{h}(B) = 2^{-\card{B}}\sum\limits_{A\subseteq B} \tilde{h}(A)$ for some $\tilde{h}$.
We prove by induction on $B$ that $\tilde{h}(B) = \sum\limits_{A\subseteq B}{(-1)^{\card{B\setminus A}} 2^{\card{A}}h(A)}$.

The base case $\card{B} = 0$ is clear. Assume the statement holds for all $B$ such that $\card{B}\leq k$, and let $B$
be of size $k+1$. By definition of $h$ we have that $2^{\card{B}} h(B) = \sum\limits_{A\subsetneq B} \tilde{h}(A) + \tilde{h}(B)$.
Using the induction hypothesis on each $A\subsetneq B$ we get that
\[
\sum\limits_{A\subsetneq B} \tilde{h}(A)
=\sum\limits_{A\subsetneq B}\sum\limits_{C\subseteq A} (-1)^{\card{A\setminus C}} 2^{\card{C}} \tilde{h}(C)
=\sum\limits_{C}2^{\card{C}} \tilde{h}(C) \sum\limits_{\substack{A\colon \\ C\subseteq A\subsetneq B}}{(-1)^{\card{A\setminus C}}}
=\sum\limits_{C}2^{\card{C}} \tilde{h}(C) (-1)^{\card{B\setminus C}+1},
\]
where in the last equality we used the fact that adding the summand corresponding to $A=B$, the sum would be $0$.
Plugging that into the previous equality and rearranging finishes the inductive step.
\end{proof}

Define $\psi = \lambda\sum\limits_{A\subseteq B}{(-1)^{\card{B\setminus A}} 2^{\card{A}} g(B)}$.
By Claim~\ref{claim:aux_mobius} we have that $\lambda g = \T\psi$.
\begin{claim}\label{claim:f_psi_dist}
  $\norm{f - \psi}_{\infty}\leq 3^n\eta$.
\end{claim}
\begin{proof}
Let $h = \T f$. Using the formula for $h$ from Claim~\ref{claim:aux_mobius}
and the definition of $\psi$ we have that for all $B\subseteq[n]$,
\[
\card{f(B) - \psi(B)}
\leq\sum\limits_{A\subseteq B}{2^{\card{A}}\card{h(A) - \lambda g(A)}}
\leq\sum\limits_{A\subseteq B}{2^{\card{A}}\eta}
= 3^{\card{B}}\eta.\qedhere
\]
\end{proof}
Let $\phi = \bigwedge_{i=1}^{m}\left(\bigoplus_{j\in A_i} x_j\right)$.
We show that $\T\phi = 2^{-m} g$. Since $\T$ is invertible by Claim~\ref{claim:aux_mobius} and $\T\psi = \lambda g$, we get $\psi = 2^m\lambda\phi$, and hence Claim~\ref{claim:f_psi_dist}
implies that $f$ is $3^n\eta$-close to $2^m\lambda \phi$ in $L_{\infty}$,
as required.

\begin{claim}\label{claim:aux_OR_XOR}
Let $A_1,\ldots,A_m$ be disjoint, non-empty sets,
and let $\phi = \bigwedge_{i=1}^{m}\bigoplus_{j\in A_i} x_j$,
$g = \bigwedge_{i=1}^{m}\bigvee_{j\in A_i} x_j$.
Then $\T \phi = 2^{-m} g$.
\end{claim}
\begin{proof}
Fix $B\subseteq [n]$, and let $B_i = B\cap A_i$ for each $i$.

If $g(B) = 0$, then $B_i = \emptyset$ for some $i$, without loss
of generality $i=1$. Thus, for any $C\subseteq B$ we have that $C\cap A_1 = \emptyset$,
and hence $\phi(C) = 0$, so $\T\phi(B)=0$.

If $g(B) = 1$, then $B_i\neq\emptyset$ for all $i$. Let ${\bf C}\subseteq B$ be chosen
uniformly at random, and denote ${\bf C}_i = {\bf C}\cap A_i$. Note that the distribution
of ${\bf C}_1,\ldots,{\bf C}_m$ is of independent uniform subsets of $B_1,\ldots,B_m$,
and as such the parity of the size of each ${\bf C}_i$ is a uniform and independent bit.
Thus, $\T\phi(B) = \Prob{{\bf C}\subseteq B}{\phi({\bf C}) = 1} = \Prob{{\bf C}\subseteq B}{\card{{\bf C}_i} \equiv 1\pmod{2} \text{ for all } i \in [m]} = 2^{-m}$.
\end{proof}

This completes the proof of Lemma~\ref{lemma:main_exact}.
\qed

\subsection{Deducing Theorem~\ref{thm:main_2func}}
In this section we use Lemma~\ref{lemma:main_exact} to deduce Theorem~\ref{thm:main_2func},
and we first sketch the argument. Given an approximate solution $f,g$, we first observe
that the function $g$ is noise insensitive --- that is, has a small Fourier tail --- and hence deduce
from Theorem~\ref{thm:Bourgain} that it is close to a junta. We then show that for almost
all restrictions $\beta$ outside the junta variables, we can associate a bounded function $\tilde{f}_{\beta}$
such that $\tilde{f}_{\beta},g_{\beta}$ are a solution to the equation in $L_{\infty}$, and we may
deduce some structure for $g_{\beta}$ and $\tilde{f}_{\beta}$. Using the fact that the restricted
variables barely affect the function $g$ (since it is junta) one can thus deduce the necessary
AND-OR structure from $g$. To get the structural result for the function $f$, slightly more work is needed.
We show that eliminating ORs that are too wide from the $g_{\beta}$'s, almost all of them become the same
function, and we show that after averaging over the removed variables, $f$ is close to a multiple
of the corresponding AND-XOR function.

We first give several statements that will be useful for us in the proof.
The following lemma from~\cite{Lifshitz} shows the
effect of the operator $\Tdown{p}{p\rho}$ on the Fourier
expansion of a function.
\begin{lemma} \label{lem:T-ev}
 If $f = \sum_S \hat{f}(S) \chi_S^{(\rho p)}$ then $\Tdown{p}{\rho p} f = \sum_S \left(\frac{(1-p)\rho}{1-\rho p}\right)^{|S|/2} \hat{f}(S) \chi_S^{(p)}$.
\end{lemma}

Using the previous lemma we may show that if $f,g$ are approximate solutions,
then $g$ has an exponentially small tail.
\begin{lemma}\label{lem:vanishing_tail}
  Let $f\colon(\power{n},\mu_{\rho p})\to[0,1]$, $g\colon(\power{n},\mu_p)\to\power{}$ be functions
  such that $\norm{\Tdown{p}{\rho p} f - \lambda g}\leq \eta$. Then for any $k\in\mathbb{N}$
  we have that $W_{\geq k}[g] \leq 2\lambda^{-2}(\eta + \rho^{k})$.
\end{lemma}
\begin{proof}
Since $f,g$ are bounded between $0$ and $1$, $\T f - \lambda g$ is bounded between $-1,1$
at each point and therefore we get that
$\norm{\T f - \lambda g}_2^2\leq \norm{\T f - \lambda g}_1\leq \eta$. Using Parseval's inequality
(and Lemma~\ref{lem:T-ev} to get the Fourier coefficients of $\T f$ on $\mu_p$),
we get that, denoting $\rho_2 = \frac{(1-p)\rho}{1-\rho p}\leq \rho$, we have
\[
\sum\limits_{S}(\rho_2^{\card{S}/2}\widehat{f}_{p\rho}(S) - \lambda\widehat{g}_{p}(S))^2 \leq \eta.
\]
Therefore, using $(a+b)^2\leq 2(a^2+b^2)$ we get that
\[
\lambda^2\sum\limits_{\card{S}\geq k}\widehat{g}_{p}(S)^2
\leq 2\eta + 2\sum\limits_{\card{S}\geq k}(\rho_2^{\card{S}/2}\widehat{f}_{p\rho}(S))^2
\leq 2\eta + 2\cdot \rho_2^{k}\sum\limits_{\card{S}\geq k}\widehat{f}_{p\rho}(S)^2
\leq 2\eta + 2\cdot \rho^{k},
\]
where in the last inequality we used Parseval to bound the sum of Fourier coefficients of $f$ by $1$ and $\rho_2\leq \rho$.
\end{proof}

Second, the following will be useful for us in the pruning process of the
wide ORs.
\begin{lemma} \label{lem:truncation}
 Suppose that $g_1,g_2$ are AND-OR functions of width at most $d$ that are $\eps$-close with respect to $\mu_p$,
 and let $\gamma = 2p^{-d} \eps$. Let $\psi_1,\psi_2$ be the truncations of $g_1,g_2$ respectively resulting
 by removing all ORs containing more than $\log_{1/(1-p)}(1/\gamma)$ variables.

 Then $\psi_1 = \psi_2$, and furthermore this function is $d\gamma$-close to $g_1$.
\end{lemma}
\begin{proof}
 We say an OR of $g_1$ is \emph{small} if it contains at most $\log_{1/(1-p)}(1/\gamma)$ variables,
 and let $A_1$ be a small OR of $g_1$. We claim that there is a small OR of $g_2$, which will be denoted by $A_2$, that contains $A_1$.
 Assume towards contradiction that this is not the case. Thus, restricting $A_1$-variables to $0$,
 the restricted function $(g_2)_{A_1\rightarrow 0}$ does not become identically $0$ and it is
 still an AND-OR function of width at most $d$, and therefore it gets the value $1$ with probability at least $p^d$.
 Since the probability that all of the variables in $A_1$ get the value $0$ is at least $(1-p)^{\log_{1/(1-p)}(1/\gamma)} = \gamma$, we get that
 $\Prob{{\bf x}}{{\sf OR}_{A_1}({\bf x}) = 0, g_2({\bf x})=1}\geq \gamma\cdot p^d$. However, note that on any such $x$ we have
 $g_1(x) = 0$ and $g_2(x) = 1$, and by assumption the probability mass on such $x$'s is at most $\eps$, so we get that
 $\eps\geq \gamma p^d$ and contradiction.

 Therefore, for each small OR of $g_1$ there is a small OR in $g_2$ containing it and vice versa.
 As the ORs in each function are disjoint in variables, it follows that each small OR of $g_1$ appears in $g_2$
 and vice versa, so in other words $\psi_1 = \psi_2$.

 Finally, since $g_1$ and $\psi_1$ may differ only when there is an OR of size at least $\log_{1/(1-p)}(1/\gamma)$
 in $g_1$ that evaluates to~$0$, and there are at most $d$ such clauses, it follows from the union bound that
 $\Prob{x}{g_1(x)\neq \psi_1(x)} \leq d\gamma$.
\end{proof}

We are now ready to prove Theorem~\ref{thm:main_2func}.
\begin{proof}[Proof of Theorem~\ref{thm:main_2func}]
Fix $\zeta,\eps>0$ from Theorem~\ref{thm:main_2func} (we assume $\eps>0$ is small enough) and choose $m=\ceil{\log(2/\zeta)}$.
Let $C=C(\zeta)$ be from Theorem~\ref{thm:Bourgain},
choose $\eta_1 = \zeta^2\eps^2/(4C\log(1/\eps))$, and pick $\tau,J$ from Theorem~\ref{thm:Bourgain} for
$\eps$ and $k = \ceil{\log(1/\eta_1)}$.
Later in the proof we will also define $\eta_2$ and subsequently take $\eta = \min(\eta_1,\eta_2)$.

Let $f,g$ be functions as in the statement of Theorem~\ref{thm:main_2func},
and set $k = \ceil{\log(1/\eta_1)}$. From Lemma~\ref{lem:vanishing_tail} we have
that $W_{\geq k}[g]\leq \lambda^{-2}(\eta + 2^{-k}) \leq \eps/(C k)$,
and Theorem~\ref{thm:Bourgain} implies that there is $T\subseteq[n]$ of size $J$ such that $g$ is $\eps^2$-close to a $T$-junta.

Take $\eta_{\ref{lemma:main_exact}}$ from Lemma~\ref{lemma:main_exact} for $\zeta$
and $n=J$, and set $\eta_2 = 6^{-J}\eta_{\ref{lemma:main_exact}}\eps^2$.
We write points $x\in \power{n}$ as $(\alpha,\beta)$ where $\alpha\in\power{T}$ and $\beta\in\power{[n]\setminus T}$.
For each $\beta\in\power{[n]\setminus T}$, define $\tilde{f}_{\beta}\colon\power{T}\to[0,1]$ and $g_\beta\colon \power{T} \to \power{}$ by
\[
\tilde{f}_{\beta}(\alpha) = \Expect{\bm{\beta'}\leq\beta}{f(\alpha,\bm{\beta'})}, \quad
g_\beta(\alpha) = g(\alpha,\beta),
\]
and let
$B = \sett{\beta\in\power{[n]\setminus T}}{\norm{\tilde{f}_{\beta} - \lambda g_{\beta}}_{\infty}\leq 3^{-J}\eps\eta_{\ref{lemma:main_exact}}}$.
Since for any $\alpha,\beta$ we have $\T\tilde f_{\beta}(\alpha) = \T f(\alpha,\beta)$, it follows that
$\Expect{\bm{\beta}}{\norm{\T\tilde{f}_{\bm{\beta}} - \lambda g_{\bm{\beta}}}_1} = \norm{\T f - \lambda g}_1\leq \eta_2$.
Therefore Markov's inequality
implies that with probability at least $1-\eps$ over $\bm{\beta}\sim\mu_p$ we have
$\norm{\T\tilde{f}_{\bm{\beta}} - \lambda g_{\bm{\beta}}}_1\leq 6^{-J}\eps\eta_{\ref{lemma:main_exact}}$,
in which case $\bm{\beta}\in B$. In particular, we conclude that $\Prob{\bm{\beta}\sim\mu_p^{[n]\setminus T}}{\bm{\beta}\in B}\geq 1-\eps$.

For each $\beta\in B$, Lemma~\ref{lemma:main_exact} implies that $g_{\beta}$ is an AND-OR function of width $r(\beta)$ which is at most $m(\zeta) = O(\log(1/\zeta))$,
and that $\tilde{f}_{\beta}$ is $\eps$-close to $2^r\lambda\cdot \text{AND-XOR}$ in $L_{\infty}$ for the corresponding AND-XOR function; we will use that only later when we establish
the structure for $f$. Since $g$ is $\eps^2$-close to a $T$-junta, if we choose $\bm{\beta},\bm{\beta'} \in B$ independently (according to $\mu_p$)
then on average $g_{\bm{\beta}}$ and $g_{\bm{\beta'}}$ are $\delta$-close, where $\delta \leq 2\eps^2/\Pr[B] \leq 2\eps^2$.
Thus there is $\beta^{\star} \in B$ such that $\Expect{\bm{\beta} \in B}{\norm{g_{\bm{\beta}} - g_{\beta^{\star}}}_1}\leq 2\eps^2$,
so by Markov's inequality, defining $B'\subseteq B$ by $B' = \sett{\beta\in B}{\norm{g_{\beta} - g_{\beta^{\star}}}_1\leq \eps}$ we have that
$\Prob{\bm{\beta}\sim\mu_p^{[n]\setminus T}}{\bm{\beta}\in B'}\geq 1-\frac{2\eps}{\Pr[B]} \geq 1-4\eps$.
We may already argue that $g$ is close to the AND-OR function $g_{\beta^{\star}}$, however that will not be strong enough to establish the
structural result for $f$ and hence we prove a stronger statement. Namely, we show that if we truncate $g_{\beta}$ by removing the wide ORs,
then almost all of them will produce the same AND-OR function $\psi$.

\paragraph{Proving the structural result for $g$.}
For each $\beta\in B'$, let $\psi_{\beta}$ be the AND-OR function $g_{\beta}$ where we remove from it all ORs whose width exceeds
$\log_{1/(1-p)}(1/(2p^{-m}\eps))$. From Lemma~\ref{lem:truncation} we get that there is an AND-OR function $\psi$ (namely, $\psi_{\beta^{\star}}$)
such that $\psi_{\beta} = \psi$ and $\norm{g_{\beta} - \psi}_1\leq 2mp^{-m}\eps = O_{\zeta}(\eps)$ for each $\beta\in B'$. Therefore, it follows that
\[
\Prob{\bm{\alpha},\bm{\beta}}{g(\bm{\alpha},\bm{\beta})\neq \psi(\bm{\alpha})}
\leq \Prob{\bm{\beta}}{\bm{\beta}\not\in B'} + \cProb{\bm{\alpha},\bm{\beta}}{\bm{\beta}\in B'}{g(\bm{\alpha},\bm{\beta})\neq \psi(\bm{\alpha})}
 = O_{\zeta}(\eps),
\]
i.e.\ $g$ is close to the AND-OR function $\psi$. Let $T'\subseteq T$ be the set of variables that appear in $\psi$,
and write $\alpha\in\power{T}$ as $\alpha = (\alpha_1,\alpha_2)$, where $\alpha_1\in \power{T'}$ and $\alpha_2\in\power{T\setminus T'}$.

\paragraph{Proving the structural result for $f$.} We show that averaging outside $T'$ makes $f$ close to
a multiple of $\phi$, where $\phi$ is the AND-XOR function corresponding to $\psi$. For each $\beta\in B'$
we denote by $\phi_{\beta}$ the AND-XOR function corresponding to $g_{\beta}$.

For each $\beta\in B'$, define $A_{\beta}(\alpha_1) =  \Expect{\bm{\alpha_2'}\sim\mu_{p\rho}^{T\setminus T'}}{\tilde{f}_{\beta}(\alpha_1,\bm{\alpha_2'})}$,
and $A\colon \power{T}\to [0,1]$ by
$A(\alpha_1) = \Expect{\bm{\beta}\sim\mu_{p}^{[n]\setminus T}}{A_{\bm{\beta}}(\alpha_1)}$.
For each $\beta \in B'$, Lemma~\ref{lemma:main_exact} implies that $\tilde{f}_{\beta}$ is $\eps$-close to
$2^{r(\beta)}\lambda\cdot \phi_{\beta}$ in $L_{\infty}$, hence by averaging over $\alpha_2$ we conclude that $A_{\beta}(\alpha_1)$
is $\eps$-close to $2^{r(\beta)}\lambda\cdot \Expect{\bm{\alpha_2}\sim\mu_{p\rho}^{T\setminus T'}}{\phi_{\beta}(\alpha_1,\bm{\alpha_2})} = c(\beta)\cdot\phi(\alpha_1)$,
where $c(\beta) \leq 2^m$.
%where $c(\beta)$ is some constant less than $1$ (it is $2^{r(\beta)}\lambda$ times the probability all of the wide XORs are $1$ under $\mu_{p\rho}$).

Set $K = \Expect{\bm{\beta}}{c(\bm{\beta}) 1_{\bm{\beta}\in B'}}$. Then
\begin{multline*}
\norm{A - K\phi}_1
\leq
%\Expect{\bm{\beta}}{\norm{A_{\bm{\beta}}} 1_{\bm{\beta}\not\in B'}}
2^m \Prob{\bm{\beta}}{\bm{\beta}\notin B'} + \Expect{\bm{\beta}}{\norm{A_{\bm{\beta}} - c(\bm{\beta})\cdot \phi}_1 1_{\bm{\beta}\in B'}}
\leq 2^m \Prob{\bm{\beta}}{\bm{\beta}\notin B'} +3^J\eta
\\ \leq 2^m O(\eps) + \eps = O_{\zeta}(\eps).
\end{multline*}
Since $A(\alpha_1) = \Expect{\bm{\alpha_2},\bm{\beta'}\sim\mu_{p\rho}}{f(\alpha_1,\bm{\alpha_2},\bm{\beta'})}$,
this shows that after averaging outside $T'$, the function $f$ is $O_{\zeta}(\eps)$-close to $K\phi$,
and we next show that one may replace $K$ by $\lambda 2^r$, where $r$ is the width of $\phi$, and retain
this closeness.

If $\psi = 0$ then $\phi=0$ so the value of $K$ does not matter, so we assume henceforth that $\psi \neq 0$, in which case we clearly have $\norm{\psi}_1 \geq p^m$.
Since $\T$ is a contraction, $\norm{\T A - K\cdot\T \phi}_1\leq \norm{A - K\phi}_1  = O_{\zeta}(\eps)$. Since by Claim~\ref{claim:aux_OR_XOR} we have
$\T\phi = 2^{-r}\psi$, it follows that $\norm{\lambda g - K \cdot 2^{-r} \psi}_1 = O_{\zeta}(\eps)$. Since $g$ is $O_{\zeta}(\eps)$-close to $\psi$,
this gives $\card{\lambda - K2^{-r}}\norm{\psi}_1 = O_{\zeta}(\eps)$, and so $\card{\lambda - K2^{-r}} = O(p^{-m} \eps) = O_{\zeta}(\eps)$,
implying that $\card{K - 2^r\lambda} = O_{\zeta}(2^r\eps)=O_{\zeta}(\eps)$. Thus, $K\cdot \phi$ is $O_{\zeta}(\eps)$-close
to $2^r\lambda\cdot \phi$, and by the triangle inequality we conclude that $\norm{A - 2^r\lambda\cdot\phi}_1\leq O_{\zeta}(\eps)$, finishing the proof.
\end{proof}

\section{Stronger structural results for special cases}\label{sec:special}
In this section we prove Theorems~\ref{thm:basic},~\ref{thm:mono_and} and~\ref{thm:and_homomorphism}.
The key idea in the proofs of these results is to use an appropriate natural coupling of downwards random walks.

%The arguments we present holds for any value of $p,\rho$ so long as they are bounded away from $0$ and $1$, and to
%simplify presentation we present them in the case $p = 1/2, \rho = 1/2$.

\subsection{Proof of Theorem~\ref{thm:basic}}\label{sec:large_lambda}
Throughout this section we fix $\zeta,p,\rho,\eps,\lambda$ as in Theorems~\ref{thm:basic}.
Let $f\colon\power{n}\to[0,1]$ and $g\colon\power{n}\to\power{}$ and $\eta>0$ (to be determined) be such that
$\norm{\Tdown{p}{\rho p} f - \lambda g}_1\leq \eta$.
\begin{claim}\label{claim:low_inf}
  Suppose that $\lambda>\rho$. Then for every $\tau>0$, there exists $\eta_1 >0$ such that if
  $\norm{\Tdown{p}{\rho p} f - \lambda g}_1\leq \eta_1$ then $I_i[g] < \tau$ for all $i$.
\end{claim}
\begin{proof}
 We prove the statement for $\eta_1 = \frac{(\lambda - \rho)\zeta(1-\zeta)}{4} \tau$.

 Assume towards contradiction that there is $i$ such that $I_i[g] \geq \tau$. This means that if we sample $\mathbf{x} \sim \mu_p^{[n] \setminus \{i\}}$ then $\Pr[g(\mathbf{x},0) \neq g(\mathbf{x},1)] \geq \tau$. Denote this event by $E$, and name the endpoints $\mathbf{x(0)},\mathbf{x(1)}$ so that $g(\mathbf{x(0)}) = 0$, $g(\mathbf{x(1)}) = 1$.

 %Consider the following experiment. Sample $\mathbf{x}$ according to $\mu_p^{[n] \setminus \{i\}}$ conditioned on the event $E$, sample $\mathbf{z} \sim \mu_\rho$, and let $\mathbf{y_0}=\mathbf{x_0} \land \mathbf{z}$, $\mathbf{y_1}=\mathbf{x_1} \land \mathbf{z}$. Since $\mathbf{y_0} \neq \mathbf{y_1}$ iff $\mathbf{z}_i = 1$, we have $\Pr[\mathbf{y_0} \neq \mathbf{y_1}] = \rho$.

 %On the other hand,
 Since $\norm{\Tdown{p}{\rho p} f - \lambda g} \leq \eta_1$ and $\mu_p(x_0),\mu_p(x_1) \geq \zeta(1-\zeta) \mu_p(x)$, we have
\begin{gather*}
 \sum_{x \in E} \mu_p(x) \Expect{\mathbf{z} \sim \mu_\rho}{f(x(0) \land \mathbf{z})} \leq \frac{1}{\zeta(1-\zeta)}
 \sum_{x \in E} \mu_p(x(0)) \Expect{\mathbf{z} \sim \mu_\rho}{f(x_0 \land \mathbf{z})} \leq \frac{\eta_1}{\zeta(1-\zeta)}, \\
 \sum_{x \in E} \mu_p(x) \left({\lambda - \Expect{\mathbf{z} \sim \mu_\rho}{f(x(1) \land \mathbf{z})}}\right) \leq \frac{1}{\zeta(1-\zeta)}
 \sum_{x \in E} \mu_p(x(1)) \card{\Expect{\mathbf{z} \sim \mu_\rho}{f(x(1) \land \mathbf{z})} - \lambda} \leq \frac{\eta_1}{\zeta(1-\zeta)}.
\end{gather*}
 In particular, if we choose $\mathbf{x} \sim \mu_p^{[n] \setminus \{i\}}$ subject to $\mathbf{x} \in E$, then
\[
 \Expect{\mathbf{x},\mathbf{z}}{f(\mathbf{x(0)} \land \mathbf{z})} \leq \frac{\eta_1}{\zeta(1-\zeta)\tau}, \quad
 \lambda - \Expect{\mathbf{x},\mathbf{z}}{f(\mathbf{x(1)} \land \mathbf{z})} \leq \frac{\eta_1}{\zeta(1-\zeta)\tau},
\]
and so
\[
 \lambda - \frac{2\eta_1}{\zeta(1-\zeta)\tau} \leq \Expect{\mathbf{x},\mathbf{z}}{f(\mathbf{x(1)} \land \mathbf{z}) - f(\mathbf{x(0)} \land \mathbf{z})}
 \leq \Prob{\mathbf{x},\mathbf{z}}{\mathbf{x(1)} \land \mathbf{z} \neq \mathbf{x(0)} \land \mathbf{z}} = \rho,
\]
 since $\mathbf{x_0} \neq \mathbf{x_1}$ iff $\mathbf{z}_i = 1$. We conclude that $\lambda - \rho \leq \frac{2\eta_1}{\zeta(1-\zeta)\tau} = \frac{\lambda-\rho}{2}$, reaching a contradiction.
\end{proof}

Theorem~\ref{thm:basic} now follows by combining Claim~\ref{claim:low_inf} with Lemma~\ref{lem:vanishing_tail}
and Bourgain's Theorem~\ref{thm:Bourgain}.
\begin{proof}[Proof of Theorem \ref{thm:basic}]
  Let $C = C(\zeta)$ be from Theorem~\ref{thm:Bourgain}, and set
  $\eta_2 = \zeta^2\eps/(4C\log(1/\eps))$, $k = 1+\ceil{\frac{\log(1/\eta_2)}{\log(1/\rho)}}$.
  Choose $\tau = \half(k/\eps)^{-C\cdot k}$, pick $\eta_1$ for $\tau$ from Claim~\ref{claim:low_inf},
  and finally set $\eta = \min(\eta_1,\eta_2)$.

  By Lemma~\ref{lem:vanishing_tail} we have that
  $W_{\geq k}[g]\leq 2\lambda^{-2}(\eta + \rho^{k})\leq \frac{\eps}{C\cdot \sqrt{k}\log^{1.5}(k)}$,
  and therefore by Theorem~\ref{thm:Bourgain} the function $g$ is $\eps$-close to a junta $h$, where $h$ may only depend on variables $i$
  such that $I_i[g]\geq (k/\eps)^{-C k}$. However, by Claim~\ref{claim:low_inf}, as $\eta\leq \eta_1$,
  all individual influences of $g$ are smaller than $\tau$, and hence $h$ must be constant.
  Let $\Gamma \in \{0,1\}$ be the closest Boolean constant function; then $g$ is $2\eps$-close to $\Gamma$.
  By the triangle inequality we have $\card{\Expect{{\bf x}\sim\mu_{\rho p}}{f}-\lambda \Expect{{\bf x}\sim\mu_p}{g}} \leq \eta$,
  and so $\card{\Expect{{\bf x}\sim\mu_{\rho p}}{f}-\lambda \Gamma} \leq \eta + 2\eps \leq 3\eps$.
\end{proof}

\subsection{Preliminary results about almost monotone functions}\label{sec:almost_mono}
Throughout this section and the subsequent one, we fix $p = \rho = 1/2$ to simplify notations.
The proofs carry over easily, as in the previous section, to any $p$ and $\rho$ that are bounded away from $0$ and $1$.

The proof of Theorem~\ref{thm:mono_and} consists of two parts. First, in this section we show that if $f,g$ are approximate solutions,
then $g$ is close to a junta, i.e., it is close to a function depending on constantly many coordinates.
Importantly, we prove that
the number of variables in the junta only depends on $\lambda$ (and doesn't increase as $\eps$, the distance from the junta, gets smaller), which
allows us to deduce that for any restriction of the junta variables, $g$ becomes nearly constant.%
\footnote{Proving
that $g$ is close to a junta, where the size of the junta may depend on $\eps$, is easy; it follows immediately from
Lemma~\ref{lem:vanishing_tail} and Bourgain's Theorem~\ref{thm:Bourgain}.} In Section~\ref{sec:mono_and} we use this
to show that $g$ is close to an AND function.

\skipi
Contrary to the previous section, when $\lambda \leq \rho$ the function $g$ may have influential variables,
and we do not know how to directly bound their number. Instead, we prove below a replacement for Claim~\ref{claim:low_inf} that considers
influences of sets of variables. We phrase the claim for the more general case in which the monotonicity
condition of $f$ is relaxed to ``almost-monotonicity'' (as discussed in the paragraph subsequent to Theorem~\ref{thm:mono_and}).
But first, we observe that all negative influences of $g$ must be small.
\begin{claim}\label{claim:low_inf_neg}
  For every $\zeta,\tau > 0$ there is $\eta_1 > 0$ such that the following holds for all $\lambda\in[\zeta,1]$.
  If $\norm{\Tdown{1/2}{1/4} f - \lambda g}_1\leq \eta_1$, then for any $i\in[n]$ we have $I_i^{-}[g]\leq \tau$.
\end{claim}
\begin{proof}
We prove the statement for $\eta_1 = \zeta\tau/4$.

Sample an edge $({\bf x},{\bf x'})$ in the $i$th direction such that ${\bf x}_i = 0$, let $E$ be the event that
it is negatively influential for $g$, that is $g({\bf x}) = 1$, $g({\bf x'}) = 0$, and assume that $\Prob{}{E}\geq \tau$.
Denote by $X(E),X'(E)$ the set of $x,x'$, respectively, for which $E$ holds.
Let $({\bf y},{\bf y'})$ be a coupled random walk downwards from ${\bf x},{\bf x'}$: ${\bf y} = {\bf x}\land {\bf z}$, ${\bf y'} = {\bf x'}\land {\bf z}$ for an independently chosen ${\bf z}\sim\mu_{1/2}$.

Consider the conditioning ${\bf z}_i = 0$. Since ${\bf x}_i = 0$, the distribution of ${\bf y}$ is unaffected by
this conditioning, and since ${\bf x},{\bf x'}$ only differ in coordinate $i$, we have that ${\bf y'} = {\bf y}$. It follows
that
\begin{equation}\label{eq12}
\cExpect{{\bf x}}{{\bf x}\in X(E)}{f({\bf y'})} \geq \Prob{{\bf z}}{{\bf z}_i = 0} \cExpect{{\bf x}}{{\bf x}\in X(E)}{f({\bf y})} = \half \cExpect{{\bf x}}{{\bf x}\in X(E)}{f({\bf y})}.
\end{equation}
On the other hand, by the approximate eigenvector condition on $x'$, we have that
\[
\sum\limits_{x\in X'(E)}{\mu_{1/2}(x')\cdot \cExpect{({\bf y'},{\bf v})\sim\mathbb{D}(1/4,1/2)}{{\bf v}=x'}{f({\bf y'})}}\leq \eta_1,
\]
which implies, since $\mu_{1/2}(X'(E))\geq \tau$, that $\cExpect{{\bf x}}{{\bf x}\in X(E)}{f({\bf y'})}\leq \eta_1/\tau$.
Similarly, the approximate eigenvector condition on $x$ implies that $\cExpect{{\bf x}}{{\bf x}\in X(E)}{f({\bf y})}\geq \lambda - \eta_1/\tau$.
Plugging both inequalities to~\eqref{eq12} yields, after rearranging, that $\eta_1\geq \lambda\tau/3$, in contradiction
to the choice of $\eta_1$.
\end{proof}

\begin{claim}\label{claim:low_inf_gen}
  For every $\zeta > 0$, there is $m\in \mathbb{N}$, such that for all $\tau > 0$ there is $\eta_2 > 0$ such that the following holds
  for all $\lambda\in[\zeta,1]$.
  Let $g\colon(\power{n},\mu_{1/2})\to\power{}$, $f\colon(\power{n},\mu_{1/4})\to[0,1]$ be such that
  $\norm{\Tdown{1/2}{1/4} f - \lambda g}_1\leq \eta_2$, and assume that for all $i\in[n]$ we have $I_i^{-}[f]\leq \eta_2$.
  Then for each $M\subseteq[n]$ of size $m$,
  \[
    I_M[g] \defeq \sum\limits_{S\supseteq M}{\widehat{g}^2(S)}\leq \tau.
  \]
\end{claim}
\begin{proof}
  Let $m = \ceil{\log(2/\zeta)}^2$
  and choose $\eta_1$ from Claim~\ref{claim:low_inf_neg} for $\zeta$ and $\zeta\tau/12$.
  We prove the statement for $\eta_2 = \min(\eta_1, \frac{\zeta}{2^{2m+7}(\sqrt{m}+1)}\tau)$.

  Assume towards contradiction that there exists $M$ of size $m$ such that $I_M[g]\geq \tau$. Consider the function
  $h(x) = \widehat{g_{[n]\setminus M\rightarrow x}}(M)$. Then $h$ can have values between
  $-1$ and $1$, and $\norm{h}_2^2 = I_M[g] \geq \tau$, and therefore $\Prob{{\bf x}}{h({\bf x})\neq 0} \geq \tau$.
  By Theorem \ref{thm:sensitivity}, whenever $h({\bf x})\neq 0$ the function
  $g_{[n]\setminus M\rightarrow {\bf x}}$ has a point $a$ with at least $\sqrt{m}$ sensitive coordinates. It follows that
  \[
    \Prob{{\bf x}}{\exists a\in \power{M}, A\subseteq M, \card{A}\geq \sqrt{m}, \text{ all coordinates in }A \text{ are sensitive
    on $({\bf x},a)$ for $g$}}\geq \tau .
  \]
  Thus, there exist a set $A$ of size $\sqrt{m}$ and $a\in \power{M}$, such that
  \[
    \Prob{{\bf x}}{\text{all coordinates in $A$ are sensitive
    on $({\bf x},a)$ for $g$}}\geq \frac{\tau}{2^{2m}} \defeq t.
  \]

  Assume without loss of generality that $A = \set{1,\ldots,\sqrt{m}}$.
  We now consider a coupled random walk as in Claim~\ref{claim:low_inf_neg}.
  Sample ${\bf x}\sim\mu_{1/2}$, let $\bf{x'(1)},\dots,{\bf x'(\sqrt{m})}$ be given by ${\bf x'(i)} = {\bf x}\xor e_i$, and consider
  a coupled downwards random walk from ${\bf x}$ and the ${\bf x'(i)}$'s to
  ${\bf y} = {\bf x}\land {\bf z},{\bf y'(1)} = {\bf x'(1)}\land {\bf z},\dots,{\bf y'(\sqrt{m})} = {\bf x'(\sqrt{m})}\land {\bf z}$, where ${\bf z}\sim \mu_{1/2}$.
  Then with probability at least $t$ we have that $g({\bf x})\neq g({\bf x'(1)})=\dots=g({\bf x'(\sqrt{m})})$, and conditioned
  on that there are two cases: (a) $g({\bf x}) = 1$ with probability at least $\half$, or (b) $g({\bf x}) = 0$ with probability at least $\half$.

  \paragraph{Case (a).}
    Denote by $E$ the event that $g({\bf x}) = 1$ and all of $A$'s coordinates are sensitive on $x$,
    and let $X(E)$ denote the set of $x$'s for which this event holds.
    The idea is that while by the approximate eigenvalue condition, random walks down from ${\bf x}$ should reach a point with $f$-value
    bounded away from $0$ and random walks down ${\bf x'(i)}$'s should reach points with $f$-value close to $0$, some walk from
    the latter group is likely to collide with the former.
    \skipi
    Formally, by the approximate eigenvalue condition on $x$ we have that
    \[
    \sum\limits_{x\in X(E)}\mu_{1/2}(x) \card{\cExpect{({\bf y},{\bf v})\sim \mathbb{D}(1/4,1/2)}{{\bf v}=x}{f({\bf y})}-\lambda}\leq \eta_2,
    \]
    implying (since in this case $\mu_{1/2}(X)\geq t/2$) that $\cExpect{{\bf x},{\bf y}}{{\bf x}\in X(E)}{f({\bf y})}\geq \lambda - \frac{\eta_2}{t/2}$.
    Similarly, by the approximate eigenvalue condition on ${\bf x'(i)}$ we have that $\cExpect{{\bf x},{\bf y'(i)}}{{\bf x}\in X(E)}{f({\bf y'(i)})}\leq \frac{\eta_2}{t/2}$.
    Therefore we get that
    \begin{align*}
    \lambda - (\sqrt{m} +1)\frac{\eta_2}{t/2}
    &\leq \cExpect{{\bf x},{\bf y},{\bf y'(i)}\text{'s}}{{\bf x}\in X(E)}{f({\bf y})-\sum\limits_{i\in A}{f({\bf y'(i)})}}\\
    &\leq \cProb{{\bf x},{\bf y},{\bf y'(i)}\text{'s}}{{\bf x}\in X(E)}{{\bf y}\neq {\bf y'(i)} ~\forall i\in A}\\
    &=\Prob{{\bf z}}{{\bf z}_i = 1 ~\forall i\in A}= 2^{-\sqrt{m}}\leq \lambda/2,
    \end{align*}
    where in the second transition we used the fact that $f$ is bounded between $0$ and $1$. Rearranging and plugging
    in expression for $t$ yields a contradiction for the choice of $\eta_2$.

   \paragraph{Case (b).}
   Let $E_0$ be the event that $g({\bf x}) = 0$ and all of $A$'s coordinates are sensitive on ${\bf x}$. We would like to identify a large subevent $E \subseteq E_0$ such that ${\bf x}_1 = {\bf x}_2 = 0$ and $g(1,1,{\bf x}_3,\dots,{\bf x}_n) = 1$.

   Let $E'\subseteq E_0$ be the event that ${\bf x}_1 = 1$ or ${\bf x}_2 = 1$.
   For any $x$ satisfying the event $E'$, we have that $(x'(1),x)$ or $(x'(2),x)$
   is a negatively influential edge in direction $i=1$ or $i=2$. It follows that
   $\Prob{}{E'}\leq I_1^{-}[g]+I_2^{-}[g]$, and since $\eta_2\leq \eta_1$ we have by Claim~\ref{claim:low_inf_neg}
   that $I_1^{-}[g],I_2^{-}[g]\leq \tau/12$. Therefore $\Prob{}{E'}\leq \tau/6$.

   Let $E'' \subseteq E_0 \setminus E'$ be the event that $g({\bf x}'(1,2)) = 0$, where ${\bf x}'(1,2) = (1,1,{\bf x}_3,\dots,{\bf x}_n)$. For any $x$ satisfying the event $E''$, we have that $(x'(1),x'(1,2))$ is a negatively influential edge in direction $i=2$. It follows that $\Prob{}{E''} \leq I_2^-[g]$. As above, $\Prob{}{E''} \leq \tau/12$. It follows that the probability of the event $E := E_0 \setminus (E' \cup E'')$ is at least $\tau/2 - \tau/4 = \tau/4$.

   Write inputs $x\in\power{n}$ as $x = (b_1,b_2,\beta)$ for
   $b_1,b_2\in\power{}$ and $\beta\in\power{n-2}$, and let
   $B = \sett{\beta}{(0,0,\beta)\in E}$; then $\mu_{1/2}(B)\geq \tau/4$.
   Let ${\bf z}\sim\mu_{1/2}^{n}$ and write ${\bf z}=({\bf z(1)},{\bf z(2)})$, where ${\bf z(1)}\in\power{2}$ and ${\bf z(2)}\in\power{n-2}$,
   and denote ${\bf y} = \bm{\beta}\land {\bf z(2)}$.
   Using the approximate eigenvector condition on the points $(0,0,\beta)$ for $\beta\in B$
   we get that $\cExpect{\bm{\beta},{\bf z}, {\bf y}}{\bm{\beta}\in B}{f(0,0,{\bf y})}\leq \frac{\eta_2}{\tau/4}$,
   and by the approximate eigenvector condition on points $(1,0,\beta)$ for $\beta\in B$ we get that
   $\cExpect{\bm{\beta},{\bf z},{\bf y}}{\bm{\beta}\in B}{f((1,0)\land {\bf z(1)}, {\bf y})}\geq \lambda - \frac{\eta_2}{\tau/4}$.
   Note that when $z(1)_1 = 0$, the distribution of $((1,0)\land {\bf z(1)}, {\bf y})$ is identical to the distribution of
   $(0,0,{\bf y})$, and thus the expected value of $f$ on these points is at most $\frac{\eta_2}{\tau/4}$. Thus we have
   $\Prob{}{{\bf z(1)_1}=1}\cExpect{\bm{\beta},{\bf z},{\bf y}}{\bm{\beta}\in B}{f(1,0, {\bf y})}\geq \lambda - 2\frac{\eta_2}{\tau/4}$,
   which, since $\Prob{{\bf z}}{z(1)_1=1} = \half$, implies that
   $\cExpect{\bm{\beta},{\bf z},{\bf y}}{\bm{\beta}\in B}{f(1,0, {\bf y})}\geq 2\lambda - 4\frac{\eta_2}{\tau/4}$. Analogously,
   we have $\cExpect{\bm{\beta},{\bf z},{\bf y}}{\bm{\beta}\in B}{f(0,1, {\bf y})}\geq 2\lambda - 4\frac{\eta_2}{\tau/4}$.
   Since $I_1^{-}[f]\leq \lambda\tau/12$, the latter inequality also implies that
   $\cExpect{\bm{\beta},{\bf z},{\bf y}}{\bm{\beta}\in B}{f(1,1, {\bf y})}\geq 2\lambda - 4\frac{\eta_2}{\tau/4}-\frac{\lambda\tau/12}{\tau/4}
   \geq \lambda - 4\frac{\eta_2}{\tau/4}$.

   Combining everything, we get that
   \begin{align*}
   \cExpect{\bm{\beta}\sim\mu_{1/2}}{\bm{\beta}\in B}{\Tdown{1/2}{1/4} f(1,1,\bm{\beta})}
   &\geq \frac{1}{4}
   \cExpect{\bm{\beta},{\bf z},{\bf y}}{\bm{\beta}\in B}{f(1,1, {\bf y})+f(1,0, {\bf y})+f(0,1, {\bf y})}\\
   &\geq \frac{5}{4}\lambda-3\frac{\eta_2}{\tau/4}
    \geq \frac{9}{8}\lambda,
   \end{align*}
   where we used the choice of $\eta_2$. This is a contradiction, since by the approximate solution condition we have that
   \[
   \cExpect{\bm{\beta}\sim\mu_{1/2}}{\bm{\beta}\in B}{\Tdown{1/2}{1/4} f(1,1,\bm{\beta})}
   \leq \cExpect{\bm{\beta}\sim\mu_{1/2}}{\bm{\beta}\in B}{\lambda g(1,1,\bm{\beta})} + \frac{\eta_2}{\tau/4}
   \leq \lambda + \frac{\lambda}{16} < \frac{9}{8}\lambda,
   \]
   and we are done.
\end{proof}

The second ingredient we need is a version of Bourgain's Theorem~\ref{thm:Bourgain}
where the role of influences of variables is played by influences of sets of variables.
Let $C=C(1/2)$ be the constant from Theorem~\ref{thm:Bourgain} for $\zeta = 1/2$.
\begin{lemma}\label{lem:Bourgain_gen}
  For every $\delta>0$, $m,k\in\mathbb{N}$ there is $\tau > 0$ such that if $g\colon \power{n}\to\power{}$
  has $W_{\geq m}[g]\geq \delta$ and $I_{J}[f]\leq \tau$ for all $J\subseteq [n]$ of size $m$, then
  $W_{\geq k}[g] \geq \frac{\delta}{2C\sqrt{k}\log^{1.5}(k)}$.
\end{lemma}
\begin{proof}
  We prove the statement for $\tau= \frac{\delta^{1+m C\cdot k}}{4 (2k)^{m C\cdot k}}$.

  Assume towards contradiction that $W_{\geq k}[g] < \frac{\delta}{2C\sqrt{k}\log^{1.5}(k)}$.
  By Theorem~\ref{thm:Bourgain}, $g$ is $\delta/2$-close to a junta $h$ depending on $T\subseteq[n]$
  whose size is at most $J = (2k/\delta)^{C\cdot k}$. Thus, $\sum\limits_{S:S\not\subseteq T}{\widehat{g}(S)^2}\leq \delta/2$,
  and since $W_{\geq m}[g]\geq \delta$ it follows that
  \[
  \delta/2\leq \sum\limits_{\substack{S\colon\card{S}\geq m\\ S\subseteq T}}{\widehat{g}(S)^2}
  \leq \sum\limits_{\substack{M\subseteq T\\ \card{M} = m}}{I_M[g]} \leq \card{T}^m \tau\leq J^m \tau,
  \]
  and we get a contradiction to the choice of $\tau$.
\end{proof}

We are now ready to prove the approximation by junta result for $g$.
\begin{lemma}\label{lem:strong_junta_approx}
  For every $\zeta>0$ there is $J\in\mathbb{N}$ such that the following holds for all $\lambda\in[\zeta,1]$.
  For every $\eps>0$ there exists $\eta > 0$
  such that if $f\colon(\power{n},\mu_{1/4})\to[0,1]$, $g\colon(\power{n},\mu_{1/2})\to\power{}$ satisfy
  $\norm{\Tdown{1/2}{1/4} f - \lambda g}\leq \eta$ and $\max_{i} I_i^{-}[f]\leq \eta$, then
  $g$ is $\eps$-close to a junta $h$ that depends on $J$ variables.
\end{lemma}
\begin{proof}
  Pick $m = m(\zeta)$ from Claim~\ref{claim:low_inf_gen},
  and let $J$, $\delta$ be from Theorem~\ref{thm:KS} for $m$, $\zeta = 1/2$ and $\eps$.
  Let $C = C(1/2)$ be from Theorem~\ref{thm:Bourgain}, and set $\eta_3 = \zeta^2\delta/(16C\log(1/\delta))$,
  $k = \ceil{\log(1/\eta_3)}$. Take $\tau = \tau(\delta,m,k)$ from Lemma~\ref{lem:Bourgain_gen}
  and $\eta_2 = \eta_2(\zeta,m,\tau)$ from Claim~\ref{claim:low_inf_gen}. We prove the statement for $\eta = \min(\eta_2,\eta_3)$.

  Let $f,g$ be as in the statement of the lemma. By
  Claim~\ref{claim:low_inf_gen} we have $I_{M}[g]\leq \tau$ for all $\card{M} = m$, and by Lemma~\ref{lem:vanishing_tail}
  we have $W_{\geq k}[g]\leq 2\lambda^{-2}(\eta + 2^{-k}) < \frac{\delta}{2C\sqrt{k}\log^{1.5}(k)}$. Therefore, Lemma~\ref{lem:Bourgain_gen}
  implies that we must have $W_{\geq m}[g]< \delta$ (otherwise $g$ would be a counterexample to the lemma), and therefore by Theorem~\ref{thm:KS}
  we get that $g$ is $\eps$-close to a $J$-junta.
\end{proof}

\subsection{Proof of Theorem~\ref{thm:mono_and}}\label{sec:mono_and}
In this section we prove Theorem~\ref{thm:mono_and}. As discussed in the paragraph below
the theorem statement, we will actually prove the statement under the more relaxed
assumption that all individual negative influences of $f$ are small.

\skipi

Fix $\zeta,\eps>0$
and let $J = J(\zeta)$ be from Lemma~\ref{lem:strong_junta_approx}.
Set $\eps' = 2^{-6J - 4}\zeta\eps$, and pick $\eta'$ from Lemma~\ref{lem:strong_junta_approx}
for $\zeta,\eps'$. Let $\eta_1$ be from Claim~\ref{claim:low_inf_neg} for $\tau = 2^{-J-2}$.
We prove the statement for $\eta = \min(\eta',\eta_1,\eps')$.

\paragraph{Proving the structure for $g$.}
By Lemma~\ref{lem:strong_junta_approx} we get that $g$ is $\eps'$-close to
a junta $h$ depending on at most $J$ variables, say on $T\subseteq[n]$.
First, observe that $h$ must be monotone. Indeed, otherwise there is $x\in\power{T}$ and $i\in T$
such that $x_i = 0$ and $h(x) = 1 > 0 = h(x\oplus e_i)$. Since $g$ and $h$ are $\eps'$-close, it follows
that the restriction $g_{T\rightarrow x}$ is $2^J \eps'$-close to being constant $1$
and that the restriction $g_{T\rightarrow x\oplus e_i}$ is $2^J \eps'$-close to being constant $0$.
In particular, we conclude that the negative influence of $i$ is at least $2^{-J}(1 - 2\cdot 2^J \eps')\geq 2^{-J-1}$.
This is a contradiction to Claim~\ref{claim:low_inf_neg} since $\eta \leq \eta_1$.

Therefore, we may discuss the minterms of $h$. We prove that $h$ has at most one minterm, which implies that
it is either an AND function or a constant function, and either way we establish the structure for $g$.
Let us write inputs as $x = (\alpha,\beta)$ for $\alpha\in\power{T}$, $\beta\in\power{[n]\setminus T}$.

\begin{proposition}\label{prop:basic_mono_and}
  For any $\alpha\in\power{T}$, we have
  $\Expect{\bm{\beta}}{\card{\Tdown{1/2}{1/4} f(\alpha,\bm{\beta}) - \lambda h(\alpha)}}\leq 2^{2J}\eps'$.
\end{proposition}
\begin{proof}
  By the triangle inequality we have
  \[
  \Expect{\bm{\beta}}{\card{\Tdown{1/2}{1/4} f(\alpha,\bm{\beta}) - \lambda h(\alpha)}}
  \leq
  \Expect{\bm{\beta}}{\card{\Tdown{1/2}{1/4} f(\alpha,\bm{\beta}) - \lambda g(\alpha,\bm{\beta})}}
  +\Expect{\bm{\beta}}{|\lambda g(\alpha,\bm{\beta})-\lambda h(\alpha)|}.
  \]
  The first expectation is at most $2^{J}\eta$ by the approximate solution condition (and the fact we restricted at most $J$ variables),
  and the second expectation is at most $2^J\eps'$ by the closeness between $g$ and $h$.
\end{proof}
Assume towards contradiction that $\alpha_1,\alpha_2$ are two distinct minterm of $h$.
By definition, we have
\begin{equation}\label{eq1}
\Expect{\bm{\beta}}{\Tdown{1/2}{1/4} f(\alpha_1,\bm{\beta})} =
2^{-\card{\alpha_1}}\Expect{\bm{\beta}}{\Tdown{1/2}{1/4} f_{T\rightarrow \alpha_1}(\bm{\beta})}
+2^{-\card{\alpha_1}}\sum\limits_{\alpha < \alpha_1} \Expect{\bm{\beta}}{\Tdown{1/2}{1/4} f_{T\rightarrow \alpha}(\bm{\beta})}.
\end{equation}
We want to upper bound the rightmost average over $\alpha$, and for that we may pick $\alpha < \alpha_1$ such that
$\Expect{\bm{\beta}}{\Tdown{1/2}{1/4} f_{T\rightarrow \alpha}(\bm{\beta})}$ is at least half that average. Since $h(\alpha) = 0$, Proposition~\ref{prop:basic_mono_and}
implies that we have $\Expect{\bm{\beta}}{\Tdown{1/2}{1/4} f(\alpha,\bm{\beta})}<2^{2J}\eps'$. Since $h(\alpha_1)=1$, Proposition~\ref{prop:basic_mono_and}
implies that $\Expect{\bm{\beta}}{\Tdown{1/2}{1/4} f(\alpha_1,\bm{\beta})}\geq \lambda - 2^{2J}\eps'$. Plugging these two
bounds into~\eqref{eq1} gives us that
\begin{equation}\label{eq2}
\Expect{\bm{\beta}}{\Tdown{1/2}{1/4} f_{T\rightarrow \alpha_1}(\bm{\beta})}\geq 2^{\card{\alpha_1}}\lambda - 2^{3J + 2}\eps'.
\end{equation}
Similarly, $\Expect{\bm{\beta}}{\Tdown{1/2}{1/4} f_{T\rightarrow \alpha_2}(\bm{\beta})}\geq 2^{\card{\alpha_2}}\lambda - 2^{3J + 2}\eps'$.

\begin{proposition}\label{prop:mono_and_abv_minterm}
  Let $\alpha,\gamma\in\power{T}$ be such that $\gamma\geq \alpha$.
  Then
  \[
  \Expect{\bm{\beta}}{\Tdown{1/2}{1/4} f_{T\rightarrow \gamma}(\bm{\beta})}\geq
  \Expect{\bm{\beta}}{\Tdown{1/2}{1/4} f_{T\rightarrow \alpha}(\bm{\beta})} - J 4^J \eta.
  \]
\end{proposition}
\begin{proof}
  Let $\alpha_0 = \alpha\rightarrow \alpha_1\rightarrow \cdots \to\alpha_r = \gamma$ be an upwards walking
  path from $\alpha$ to $\gamma$, where $r\leq J$. Clearly by the triangle inequality, it is enough to show
  \[
  \Expect{\bm{\beta}}{\Tdown{1/2}{1/4} f_{T\rightarrow \alpha_{j+1}}(\bm{\beta})}
  -\Expect{\bm{\beta}}{\Tdown{1/2}{1/4} f_{T\rightarrow \alpha_{j}}(\bm{\beta})}
  \geq - 4^J \eta.
  \]
  To see that, let $i$ be the index on which the two inputs $\alpha_{j+1},\alpha_j$ differ, and expand out the definition of $\Tdown{1/2}{1/4}$
  on the left-hand side to write it as
  \[
  \Expect{\bm{\beta},{\bf z}\sim\mu_{1/2}}{f(\alpha_{j+1},\bm{\beta}\land {\bf z})
  -f(\alpha_{j},\bm{\beta}\land {\bf z})}
  = 4^{J}\Expect{(\bm{\alpha},\bm{\beta})\sim \mu_{1/4}}{(f(\bm{\alpha}\oplus e_i,\bm{\beta}) -f(\bm\alpha,\bm{\beta}))1_{\bm{\alpha} = \alpha_j}}.
  \]
  As the last expectation is at least $-I_i^{-}[f]\geq -\eta$, the proof is concluded.
\end{proof}
Consider $\gamma = \alpha_1\lor \alpha_2$, choose ${\bf z}\sim\mu_{1/2}^{T}$, and let $\bm{\alpha} = \gamma \land {\bf z}$.
Let $E_1$ be the event that $\bm{\alpha}\geq \alpha_1$, and let $E_2$ be the event that $\bm{\alpha} < \alpha_1$ and $\bm{\alpha}\geq \alpha_2$.
Clearly, $\Prob{}{E_1} = 2^{-\card{\alpha_1}}$, and since $\alpha_1,\alpha_2$ are incomparable (as they are distinct minterms)
we have $\Prob{}{E_2} \geq 2^{-J}$. Therefore, we get that
\begin{align*}
\Expect{\bm{\beta}}{\Tdown{1/2}{1/4} f(\gamma,\bm{\beta})}
=\Expect{\bm{\alpha}}{\Expect{\bm{\beta}}{\Tdown{1/2}{1/4} f_{T\rightarrow \bm{\alpha}}(\bm{\beta})}}
&\geq 2^{-\card{\alpha_1}}\cExpectop{\bm{\alpha}}{E_1}{\Expect{\bm{\beta}}{\Tdown{1/2}{1/4} f_{T\rightarrow \bm{\alpha}}(\bm{\beta})}}\\
&+2^{-J}\cExpectop{\bm{\alpha}}{E_2}{\Expect{\bm{\beta}}{\Tdown{1/2}{1/4} f_{T\rightarrow \bm{\alpha}}(\bm{\beta})}}.
\end{align*}
Consider the right-hand side. For the first expectation, for any $\alpha$ for which $E_1$ holds,
we have by Proposition~\ref{prop:mono_and_abv_minterm} (using also~\eqref{eq2})
that $\Expect{\bm{\beta}}{\Tdown{1/2}{1/4} f_{T\rightarrow \alpha}(\bm{\beta})}\geq 2^{\card{\alpha_1}}\lambda - 2^{5J+2}\eps'$.
Similarly, for the second expectation, for any $\alpha$ for which $E_2$ holds we have
$\Expect{\bm{\beta}}{\Tdown{1/2}{1/4} f_{T\rightarrow \alpha}(\bm{\beta})}\geq 2^{\card{\alpha_2}}\lambda - 2^{5J+2}\eps'$.
Plugging these two bounds we conclude that
\[
\Expect{\bm{\beta}}{\Tdown{1/2}{1/4} f(\gamma,\bm{\beta})}
\geq (1 + 2^{\card{\alpha_2} - J})\lambda - 2^{5J+3}\eps'.
\]
On the other hand, by the approximate solution condition we have that
$\Expect{\beta}{\Tdown{1/2}{1/4} f(\gamma,\beta)}\leq \lambda + 2^J\eta \leq \lambda +2^J\eps'$.
Combining the two inequalities and rearranging, we conclude that
$\lambda \leq 2^{J-\card{\alpha_2}}(2^{5J+3} + 2^J)\eps'\leq 2^{6J+3}\eps'$,
contradicting the definition of $\eps'$.

Therefore, $h$ is either constant or an AND function. In the latter case, if $\alpha_1$ is a minterm for $h$
then from~\eqref{eq2} and $f\leq 1$ we get that $2^{\card{\alpha_1}}\lambda \leq 2$, implying that $\card{\alpha_1}\leq \log(2/\lambda)$.

\paragraph{Proving the structure for $f$.}
Assume first that $h = {\sf AND}_T$ (possible $T = \emptyset$, in which case $h = 1$).
Define $\tilde{f}(\alpha) = \Expect{\bm{\beta}\sim\mu_{1/4}}{f(\alpha,\bm{\beta})}$
and
$\tilde{g}(\alpha) = \Expect{\bm{\beta}\sim\mu_{1/2}}{g(\alpha,\bm{\beta})}$.
Then by the triangle inequality and the approximate solution condition, we get that
$\norm{\Tdown{1/2}{1/4} \tilde{f} - \lambda\tilde{g}}_1\leq \eta$.

For any $\alpha$ such that $h(\alpha) = 0$, we have that $\tilde{g}(\alpha)\leq 2^J\eps'$,
and therefore $\Tdown{1/2}{1/4} \tilde{f}(\alpha) \leq 2^J\eps' + \eta \leq 2^{J+1}\eps'$.
On the other hand, clearly $\Tdown{1/2}{1/4} \tilde{f}(\alpha)\geq 4^{-J}\tilde{f}(\alpha)$,
so we get that $\tilde{f}(\alpha)\leq 2^{3J+1}\eps'$.

For $\alpha = e_T$ (note that $h(\alpha) = 1$), we have $\tilde{g}(\alpha)\geq 1-2^J\eps'$,
hence we get that $\card{\Tdown{1/2}{1/4} \tilde{f}(\alpha) - \lambda}\leq 2^{J+1}\eps'$. Note that
\[
\Tdown{1/2}{1/4} \tilde{f}(\alpha)
=2^{-\card{T}} \tilde{f}(\alpha) + 2^{-\card{T}}\sum\limits_{\alpha' < \alpha} \tilde{f}(\alpha'),
\]
and as for any $\alpha'<\alpha$ it holds that $h(\alpha') = 0$, we have from the previous paragraph that
$\tilde{f}(\alpha')\leq 2^{3J+1}\eps'$, and in conclusion we get that $\card{\Tdown{1/2}{1/4} \tilde{f}(\alpha) - 2^{-\card{T}} \tilde{f}(\alpha)}\leq 2^{3J+1}\eps'$.
Therefore, by the triangle inequality we get that $\card{2^{-\card{T}} \tilde{f}(\alpha)-\lambda}\leq 2^{3J+2}\eps'$, or in other words,
that $\card{\tilde{f}(\alpha) - 2^{\card{T}}\lambda}\leq 2^{4J+2}\eps'$. It follows that
$\norm{\tilde{f} - 2^{\card{T}}\lambda\cdot {\sf AND}_{T}}_{\infty}\leq 2^{4J+2}\eps'\leq \eps$, and we are done.

When $h = 0$, the same argument shows that $\|\tilde f\|_\infty \leq \epsilon$. \qed

\subsection{Proof of Theorem~\ref{thm:and_homomorphism}} \label{sec:thm:and_homomorphism}
In this section we show that Theorem~\ref{thm:mono_and} implies Theorem~\ref{thm:and_homomorphism}.
Let $\zeta,\eps>0$ be as in the theorem statement. If either $g$ or $h$ are $10\eps$-close to the
constant $0$ function, then we immediately get that $f$ is $11\eps$-close to the constant $0$ function (assuming $\eta \leq \epsilon$),
and the first item is proved (up to a scaling of $\epsilon$). We assume henceforth that $g,h$ are $10\eps$-far from the constant $0$
function, and in particular their averages are at least $10\eps$.

Set $\zeta'=\min(\eps,\zeta)$, and take $\eta_1$ from Theorem~\ref{thm:mono_and} for $\zeta'$ and $\eps^2\zeta^3$,
and prove the statement for $\eta = \eta_1\zeta^2\eps^3/6$.

We argue that $h$ is $\eps$-close to ${\sf AND}_T$ for some $T\subseteq[n]$. Let $\lambda_{h} = \Expect{{\bf y}\sim\mu_\rho}{h({\bf y})}\geq \eps$.
We note that for any $x\in\power{n}$ we have that $\card{\Tdown{p}{\rho p} f(x) - \lambda_h g(x)}\leq \Prob{{\bf y}}{f(x\land {\bf y}) \neq g(x) h({\bf y})}$,
and therefore $\norm{\Tdown{p}{\rho p} f - \lambda_h g}_{1}\leq \eta \leq \eta_1$. To use Theorem~\ref{thm:mono_and}, we next argue that all individual
negative influences of $f$ are at most $\eta_1$.

Assume towards contradiction that there exists $i\in[n]$, without loss of generality $i=n$,
such that $I_{i}^{-}[f]\geq \eta_1$, and sample ${\bf x}\sim\mu_p^{n-1}$, ${\bf y}\sim\mu_{\rho}^{n-1}$,
so that $({\bf x}\land {\bf y},0)$, $({\bf x}\land {\bf y},1)$ is an edge for $f$ in direction $i$; then with probability
at least $\tau$ we have $f({\bf x}\land {\bf y},0) = 1$, $f({\bf x}\land {\bf y},1) = 0$. We show that this must lead to violation
of the AND-homomorphism condition. Considering the pairs of inputs $({\bf x},0), ({\bf y},1)$
and $({\bf x},1),({\bf y},0)$ we conclude that unless the AND-homomorphism condition breaks in one of
them, we have that $g({\bf x},1) = 1$, $h({\bf y},1) = 1$, but then the AND-homomorphism condition fails for the
pair of inputs $({\bf x},1),({\bf y},1)$. Therefore there is one of these three pairs that fails the AND-homomorphism condition
with probability at least $\eta_1/3$, and in any case we get that it fails with probability at least $\min(p,1-p)^2\eta_1/3 > \eta$,
and contradiction.

Therefore, we may apply Theorem~\ref{thm:mono_and} to conclude that there is $T_1\subseteq[n]$ such that $g$ is $\eps^2\zeta^3$-close
to ${\sf AND}_{T_1}$. Analogously, the same argument shows there is $T_2\subseteq[n]$ such that $h$ is $\eps^2\zeta^3$-close to ${\sf AND}_{T_2}$,
and to finish the proof we argue that $T_1 = T_2$. Since $g,h$ have averages at least $10\eps$ and are $\eps^2\zeta^3$-close to ${\sf AND}_{T_1},{\sf AND}_{T_2}$
respectively, the latter functions have expectation at least $9\eps$.

Assume towards contradiction that $T_1\neq T_2$, say there is $i\in T_1\setminus T_2$.
Sample ${\bf x}\sim\mu_p$ and ${\bf y}\sim\mu_{\rho}$ conditioned on ${\bf y}_{T_2} = \vec{1}$, ${\bf y}_i = 0$,
${\bf x}_{T_1\setminus\set{i}} = 1$, ${\bf x}_i = 1$
(note that by the above, the event we condition ${\bf y}$ on has probability $\geq 9\eps\zeta$, as we condition on ${\sf AND}_{T_2}$
being $1$ and on the value of an additional variable; similarly the event we condition ${\bf x}$ on has probability $\geq 9\eps\zeta$).
With probability
at least $1-\eta/(9\eps\zeta)^2\geq 1-\eps/3$ we get that $f({\bf x}\land {\bf y}) = g({\bf x})\land h({\bf y})$ as well as
$f({\bf x}\land {\bf y}) = g({\bf x}\oplus e_i)\land h({\bf y})$ (noting that $({\bf x}\oplus e_i)\land {\bf y} = {\bf x}\land {\bf y}$ as ${\bf y}_i = 0$).
Also, with probability at least $1-3\frac{\eps^2\zeta^3/\zeta}{(9\eps\zeta)^2} = 1-\eps/3$, we get that $h({\bf y}) = {\sf AND}_{T_1}({\bf y}) = 1$,
$g(x) = {\sf AND}_{T_2}({\bf x}) = 1$ and $g({\bf x}\oplus e_i) = {\sf AND}_{T_2}({\bf x}\oplus e_i) = 0$ (using $\mu_p(x \oplus e_i) \leq \mu_p(x)/\zeta$). Therefore with probability
at least $1-2\eps/3 > 0$ we get that
\[
f({\bf x}\land {\bf y}) = g({\bf x})\land h({\bf y}) = 1\neq 0 = g({\bf x}\oplus e_i)\land h({\bf y}) = f({\bf x}\land {\bf y}),
\]
and contradiction. Therefore $T_1 = T_2$. It now follows immediately that $f$ is $2\eps^2\zeta^3+\eta\leq \eps$-close to
${\sf AND}_{T_1}$, and we are done.\qed

\section{Structural results for the one-sided error version}\label{sec:one_sided}
In this section we prove Theorem~\ref{thm:one_sided_error}. To do so,
we show that if $f,g$ are approximate solutions with one-sided error,
then $g$ has a small Fourier tail, in which case we can apply Bourgain's
Theorem~\ref{thm:Bourgain} to conclude that it is junta. Finally, we show
that since the negative influences of $g$ are all small (individually),
$g$ must actually be close to a monotone junta.

Most of the effort in the proof is devoted into showing that $g$
has a small Fourier tail. In the approximate solutions case, there
is a quick Fourier-analytic proof of this fact that we have already
used, namely Lemma~\ref{lem:vanishing_tail}. In the one-sided error
case, however, we are not aware of any such quick Fourier-analytic
proof. Instead, we us combinatorial/probabilistic arguments
similar to the ones we used in Section~\ref{sec:special}.

\subsection{A tail bound for one-sided error solutions}
It will be useful for us to consider a more combinatorial notion closely related to the Fourier tail,
namely the notion of noise sensitivity, and prove that if $f,g$ are approximate solutions with
one-sided error, then $g$ is very stable with respect to small enough noise rate.

Let $p,\nu\in(0,1)$.
The $p$-biased, $(1-\nu)$-correlated distribution over $(x,y)\in \power{n}\times\power{n}$
is defined as follows: sample $x\sim\mu_p$, and for each $i\in[n]$ independently set $y_i = x_i$
with probability $1-\nu$, and otherwise sample $y_i$ to be an independent $p$-biased bit.
\begin{definition}
  Let $g\colon(\power{n},\mu_p)\to\power{}$, and let $\nu\in (0,1)$.
  The noise sensitivity of $g$ at $\nu$ is ${\sf NS}_{\nu}(g) = \Prob{\text{$({\bf x},{\bf y})$ is $(1-\nu)$ correlated}}{g({\bf x}) \neq g({\bf y})}$.
\end{definition}

The following lemma is the main result of this section.
\begin{lemma}\label{lem:one_sided_NS}
  For every $\zeta>0$, there is $\nu_0>0$ such that the following holds
  for any $\rho, p\in[\zeta,1-\zeta]$ and $\lambda\in[\zeta,1]$. For any $\nu\in (0,\nu_0)$
  there is $\eta>0$ such that if $f\colon(\power{n},\mu_{\rho p})\to[0,1]$,
  $g\colon(\power{n},\mu_p)\to\power{}$ are one-sided error solutions with $\lambda$ and error $\eta$, then ${\sf NS}_{\nu}[g]\leq \nu^{3/4}$.
\end{lemma}
\begin{remark}
  The bound $\nu^{3/4}$ in Lemma~\ref{lem:one_sided_NS} can be improved to any
  bound that is $o(\nu)$.
\end{remark}
To make notations easier, we prove Lemma~\ref{lem:one_sided_NS} in the case that $p = \rho = 1/2$;
the proof carries over to the general case with minor adjustments.

We next set up some tools that we need for the proof of Lemma~\ref{lem:one_sided_NS}.
For a parameter $\nu\in(0,1)$ we define a distribution $\mathcal{D}_\nu$
over quadruples $(y,m,x,z)$, that will allow us to couple a $(1-\nu)$-correlated
pair of inputs $x,z$ with points $y,m$ that are below them and marginally correspond
to the downwards-walk distribution from each $x,z$ (for two different walk lengths).
\begin{definition}[The Distribution $\mathcal{D}_\nu$]\label{def:dist_coupled}
  Set $\theta = \frac{\nu}{2+\nu}$.
  The distribution $\mathcal{D}_\nu$ over $y,m,x,z$ in $\power{n}$ is defined in the following manner:
\begin{itemize}
    \item Pick $({\bf y},{\bf m})\sim \mathbb{D}(1/4,1/2-\nu/4)$.
    \item For each $i=1,\ldots,n$ independently do the following:
    \begin{itemize}
      \item If ${\bf m}_i = 1$, set ${\bf x}_i = {\bf z}_i = 1$.
      \item If ${\bf m}_i = 0$, set ${\bf x}_i = 1, {\bf z}_i = 0$ or ${\bf x}_i = 0, {\bf z}_i = 1$, each with probability $\theta$,
    and otherwise ${\bf x}_i = {\bf z}_i = 0$.
    \end{itemize}
  \end{itemize}
\end{definition}
\noindent When $\nu$ is clear from context, we omit the subscript and just use $\mathcal{D}$.

Inspecting the marginal distributions, we see that ${\bf y}\sim \mu_{1/4}$, ${\bf m}\sim \mu_{1/2-\nu/4}$, and ${\bf x},{\bf z}\sim\mu_{1/2}$, since $\Pr[{\bf x}_i=1] = (1/2-\nu/4) + (1/2+\nu/4) \cdot \theta = 1/2$.
Also, for each $i\in [n]$, the probability that ${\bf x}_i={\bf z}_i$ is equal to
$(1/2 - \nu/4) + (1/2+\nu/4)\cdot (1-2\theta)=1 - \nu/2$, and these events are independent, thus the distribution of
$({\bf x},{\bf z})$ is $(1-\nu)$-correlated over $\mu_{1/2}$.
Finally, notice that ${\bf y}_i \leq {\bf m}_i \leq {\bf x}_i,{\bf z}_i$, implying that the marginal distribution of ${\bf y},{\bf x}$ or ${\bf y},{\bf z}$ is $\mathbb{D}(1/4,1/2)$.

Let $E$ be the event that $g({\bf x}) = 1, g({\bf z}) = 0$ when $({\bf y},{\bf m},{\bf x},{\bf z})\sim\mathcal{D}_{\nu}$,
consider the distribution of ${\bf x}$ conditioned on $E$, given by
$p_a = \cProb{({\bf y},{\bf m},{\bf x},{\bf z})\sim \mathcal{D}_\nu}{E}{{\bf x}=a}$, and let $A = \sett{a}{p_a\geq 1/2^{n+1}}$.

It will be important for us to understand the distribution of $y$ when we sample $({\bf y},{\bf m},{\bf x},{\bf z})\sim \mathcal{D}_\nu$
conditioned on $E$. More precisely, for any $a\in A$, we consider the following two distributions:
  \begin{enumerate}
    \item $\mathcal{D}^a_1$: sample $({\bf y},{\bf m},{\bf x},{\bf z})\sim \mathcal{D}_{\nu}$ conditioned on ${\bf x}=a$ and $E$, and output ${\bf y}$.
    \item $\mathcal{D}^a_2$: sample $({\bf y},{\bf m},{\bf x},{\bf z})\sim \mathcal{D}_{\nu}$ conditioned on ${\bf x}=a$, and output ${\bf y}$.
  \end{enumerate}
  We remark that clearly $\mathcal{D}^a_2$ is the distribution of $u$ in $({\bf u},{\bf v})\sim \mathbb{D}(1/4,1/2)$ conditioned on ${\bf v}=a$,
  that is, uniform over ${\bf y}\in \power{{\sf supp}(a)}$. We wrote it in this form to be suggestive of the following proposition,
  that asserts that the distributions $\mathcal{D}^a_1$ and $\mathcal{D}^a_2$ are somewhat close.
  \begin{proposition}\label{prop:1}
    For each $\xi>0$ there are $\nu_0,c>0$
    such that the following holds. If $\Prob{}{E}\geq \half\nu_0^{3/4}$, then for each $a\in A$
    we have $\Prob{{\bf y}\sim \mathcal{D}^a_2}{\mathcal{D}^a_1({\bf y}) \geq c \mathcal{D}^a_2({\bf y})}\geq 1-\xi$.
  \end{proposition}
  Before proving Proposition \ref{prop:1}, let us
  show that it implies Lemma~\ref{lem:one_sided_NS}.

\begin{proof}[Proof of Lemma~\ref{lem:one_sided_NS}]
  Set $\xi = \lambda/2$, and let $\nu_0,c>0$ be from Proposition \ref{prop:1}.
  Fix $\nu\in(0,\nu_0)$, and prove the statement for $\eta = c\zeta \nu^{3/2}/ 100$.

  Assume towards contradiction that ${\sf NS}_\nu[g]\geq \nu^{3/4}$.
  Throughout the proof we will take $({\bf y},{\bf m},{\bf x},{\bf z})\sim \mathcal{D}_{\nu}$;
  thus, the event $E = \sett{(x,z)}{g(x)=1, g(z) = 0}$ has probability
  $\geq \half \nu^{3/4}$.

  We consider the quantity $\Expect{{\bf y},{\bf x},{\bf z}}{1_E({\bf x},{\bf z})\cdot f({\bf y})}$, and show
  a lower bound as well as an upper bound on it, which combined together
  lead to a contradiction. For the upper bound we note that
  \begin{equation}\label{eq:up_bd}
  \Expect{{\bf y},{\bf x},{\bf z}}{1_E({\bf x},{\bf z})\cdot f({\bf y})}
  \leq \Expect{{\bf y},{\bf z}}{(1-g({\bf z})) f({\bf y})}
  = \Expect{{\bf z}}{(1-g({\bf z})) \Tdown{1/2}{1/4} f({\bf z})}
  \leq \eta,
  \end{equation}
  where the last inequality is by the condition that $f,g$ are one-sided error solutions with error $\eta$.

  The lower bound requires more effort. Recall that the distribution $p_a$ is defined by $p_a = \cProb{}{E}{{\bf x}=a}$,
  and the set $A$ consists of $a$'s such that $p_a\geq \frac{1}{2^{n+1}}$.
  We need the following easy proposition.
  \begin{proposition}
    $\mu_{1/2}(A)\geq \frac{1}{2}\Prob{}{E}$.
  \end{proposition}
  \begin{proof}
    Note that for every $a$, $p_a \leq \frac{\Prob{{\bf x}}{{\bf x}=a}}{\Prob{}{E}} = \frac{2^{-n}}{\Prob{}{E}}$.
    Thus,
    \[
    1 = \sum\limits_{a}{p_a}
    =\sum\limits_{a\in A}{p_a} + \sum\limits_{a\not\in A}{p_a}
    \leq \card{A}2^{-n}\frac{1}{\Prob{}{E}} + \half
    =\frac{\mu_{1/2}(A)}{\Prob{}{E}} + \half,
    \]
    and the result follows by rearrangement.
  \end{proof}

  We are now ready for the lower bound. We have
  \begin{equation}\label{eq4:lvl0}
    \Expect{{\bf y},{\bf x},{\bf z}}{1_E({\bf x},{\bf z})\cdot f({\bf y})}
    = \Prob{{\bf x},{\bf z}}{E}\cdot \cExpect{{\bf y},{\bf x},{\bf z}}{E}{f({\bf y})},
  \end{equation}
  and we bound the conditional expectation. We have
  \begin{align}\label{eq9}
  \cExpect{{\bf y},{\bf x},{\bf z}}{E}{f({\bf y})}\notag
   = \sum\limits_{a,y}{p_a \mathcal{D}^a_1(y) f(y)}
  &\geq \frac{1}{2^{n+1}}\sum\limits_{a\in A, y}{\mathcal{D}^a_1(y) f(y)}\\
  &\geq \frac{c}{2^{n+1}}\sum\limits_{a\in A}\sum\limits_{y\in Y_a}{\mathcal{D}^a_2(y)f(y)},
  \end{align}
  where $Y_a$ is the set of $y$'s on which $\mathcal{D}^a_1(y)\geq c \mathcal{D}^a_2(y)$.
  By Proposition \ref{prop:1}, $Y_a$ contains all but  $\xi$ of the mass, so since $f\leq 1$
  for each $a\in A$ we get that $\sum\limits_{y\in Y_a}{\mathcal{D}^a_2(y) f(y)} \geq \sum\limits_{y}{\mathcal{D}^a_2(y) f(y)} - \xi$.
  Also, note that $\sum\limits_{y}{\mathcal{D}^a_2(y) f(y)} = \Tdown{1/2}{1/4} f(a) $, so plugging that into~\eqref{eq9} we get
  \begin{equation}\label{eq11}
    \eqref{eq9}\geq \frac{c}{2}\left(2^{-n}\sum\limits_{a\in A}\Tdown{1/2}{1/4} f(a) - \mu_{1/2}(A)\xi\right).
  \end{equation}
  Next, we lower-bound the sum on the right-hand side.
  Let $A'\subseteq A$ be the set of $a\in A$ such that $\Tdown{1/2}{1/4} f(a) > \lambda$.
  Since the probability over $a\sim\mu_{1/2}$ that $g(a) = 1$
  and $\Tdown{1/2}{1/4} f(a) \leq \lambda$ is at most $\eta$, we get that $\mu_{1/2}(A')\geq\mu_{1/2}(A) - \eta$,
  and therefore
  \[
    \frac{1}{2^n}\sum\limits_{a\in A}\Tdown{1/2}{1/4} f(a) \geq
    \frac{1}{2^n}\sum\limits_{a\in A'}\Tdown{1/2}{1/4} f(a)
    >\mu_{1/2}(A')\lambda
    \geq \mu_{1/2}(A)\lambda - \eta.
  \]
  Plugging that into \eqref{eq11},
  \[
  \eqref{eq9}\geq
  \frac{c}{2}(\mu_{1/2}(A)\lambda - \mu_{1/2}(A)\xi-\eta)\geq
  \frac{c}{2}\left(\frac{\lambda}{4}\Prob{}{E} - \eta\right),
  \]
  where we used the choice $\xi = \lambda/2$ and the fact that $\mu_{1/2}(A)\geq \half \Prob{}{E}$.
  By the choice of $\eta$, we have $\eta\leq \frac{\nu^{3/4}\lambda}{16}\leq \frac{\lambda}{8}\Prob{}{E}$,
  hence we get that $\eqref{eq9}\geq \frac{c\lambda }{16}\Prob{}{E}$.
  Plugging all the way back into~\eqref{eq4:lvl0}, we get
  \begin{equation}\label{eq:lower_bd}
   \Expect{{\bf y},{\bf x},{\bf z}}{1_E({\bf x},{\bf z})\cdot f({\bf y})}
  \geq \frac{c\lambda }{16}\Prob{}{E}^2
  \geq \frac{c\lambda \nu^{3/2}}{64}.
  \end{equation}
  %where in the last inequality we used $\Prob{}{E}\geq \half\nu^{3/4}$.
  Combining inequalities \eqref{eq:lower_bd} and \eqref{eq:up_bd}
  we get that $\eta\geq \frac{c\lambda \nu^{3/2}}{64}$, contradicting
  the choice of $\eta$.
  \end{proof}

\subsection{Proof of Proposition~\ref{prop:1}}
  In this section we prove Proposition \ref{prop:1}.
  We will need the following piece of notation: for a function $g\colon\power{n}\to\mathbb{R}$ and a probability measure $\mu$ on $\power{n}$,
  we denote $\mu(g) = \Expect{{\bf x}\sim\mu}{g({\bf x})}$.
  We begin with an auxiliary lemma that will be helpful.

  \begin{lemma}\label{lem:2}
    For every $\gamma,\xi>0$ there is $\nu>0$ such that the following holds.
    Let $g\colon (\power{n}, \mu_{\nu}) \to \power{}$ be a function such that $\mu_{\nu}(g) \geq \gamma\cdot \nu^{3/4}$.
    Then except with probability $\xi$, for ${\bf x}\sim \mu_{1/2}^n$ we have $\Tup{\nu}{1/2} g({\bf x})\geq \nu$.
  \end{lemma}
  \begin{proof}
    Fix $\gamma,\xi$.
    Without loss of generality, we restrict ourselves to the case $\xi\leq \min(\gamma,1/10)$ (otherwise we set $\xi$ to be this
    minimum and prove a stronger statement).
    We prove the statement with $\nu = \xi^8$.

    Set $k=\frac{1}{2\nu}$, and consider the following distribution over ${\bf y(1)},\dots,{\bf y(k)},{\bf x}$:
    choose ${\bf y(1)},\dots,{\bf y(k)} \sim \mu_\nu^n$, and set ${\bf z} = {\bf y(1)}\lor {\bf y(2)}\lor \dots\lor {\bf y(k)}$.
    Pick ${\bf x}$ by taking ${\bf x}_i = 1$ if ${\bf z}_i=1$, else
    take ${\bf x}_i = 1$ with probability\footnote{This probability is positive since $(1-\nu)^k \geq 1-\nu\cdot k= 1/2$.}
    \[
      r = \frac{(1-\nu)^k-1/2}{(1-\nu)^k}.
    \]
    Hence the marginal distribution of ${\bf x}$ is $\mu_{1/2}$.

    Since ${\bf y(1)},\dots,{\bf y(k)} $ are chosen independently,
    \[
    \Prob{{\bf y(1)},\dots,{\bf y(k)},{\bf x}}{g({\bf y(i)}) = 0 ~\forall i=1,\ldots,k}\leq (1-\gamma\cdot \nu^{3/4})^k \leq e^{-\gamma / 2\nu^{1/4}} \leq e^{-1/2\xi}
     \leq \xi/2,
    \]
     using $\xi \leq 1/10$.
     Define
    \[
      A = \sett{a}{\cProb{{\bf y(1)},\dots,{\bf y(k)},{\bf x}}{{\bf x}=a}{g({\bf y(i)}) = 0 ~\forall i=1,\ldots,k}\leq \half},
    \]
    then by an averaging argument
    we get that $\mu_{1/2}(A) \geq 1-\xi$.
    We show that every $a\in A$ has a high $\Tup{\nu}{1/2} g(a)$ value. Indeed, we have
    that
    \begin{align*}
    \half
    &\leq \cProb{{\bf y(1)},\dots,{\bf y(k)},{\bf x}}{{\bf x}=a}{\exists i ~g({\bf y(i)})= 1}\\
    &\leq \cExpect{{\bf y(1)},\dots,{\bf y(k)},{\bf x}}{{\bf x}=a}{g({\bf y(1)}) + \dots + g({\bf y(k)})}\\
    &= k\cExpect{{\bf y(1)},\dots,{\bf y(k)},{\bf x}}{{\bf x}=a}{g({\bf y(1)})},
    \end{align*}
    where we used the fact that $y_1,\dots,y_k$ are identically distributed. We note that the distribution
    of $y_1$ is the same as the distribution of $u$, when $({\bf u},{\bf v})\sim\mathbb{D}(\nu,1/2)$ conditioned on ${\bf v}=a$,
    thus by definition of $\Tup{\nu}{1/2}$, the value of the expectation is exactly $\Tup{\nu}{1/2} g(a)$.
    Therefore, we get that $\Tup{\nu}{1/2} g(a)\geq \frac{1}{2k} =\nu$.
  \end{proof}

  %\subsection{Proof of Proposition \ref{prop:1}}
   We now turn to the proof of Proposition~\ref{prop:1}.
   Fix $\xi>0$, and choose $\nu>0$ from Lemma \ref{lem:2} for $\gamma=1/8$ and $\xi/2$.
   We prove Proposition \ref{prop:1} with $\nu_0 = \nu$ and $c =\nu^{7/4}/16$.

   Fix $a,g$ as in the statement of Proposition \ref{prop:1}.
   %We assume with some loss of generality that $\card{a} = \half n$
%   (the proof below goes through with minor adjustments, specifically in the definition of
%   $h_{S^+,s}$ below, $\half$ needs to be replaced with $r=1 - \frac{\card{a}}{n}$, as well as subsequent biases for $z$).
   Consider the following distributions over $x,y$:
   Let ${\bf S}$ be a random subset of $\sett{i}{a_i = 0}$ in which each element is picked with probability $\nu$,
   and choose ${\bf s}\in\power{{\bf S}}$ uniformly at random. Let $I=\sett{i}{a_i = 1}$, and define
   the function $g_{{\bf S},{\bf s}}\colon \power{I}\to\power{}$ by $g_{{\bf S},{\bf s}}(z) = 1 - g(a\xor z\xor {\bf s})$.

   \begin{proposition}
     $\Expect{{\bf S},{\bf s}}{\mu_{\nu/2}(g_{{\bf S},{\bf s}})} \geq \frac{1}{4}\nu^{3/4}$.
   \end{proposition}
   \begin{proof}
     By definition,
     \[
       \Expect{{\bf S},{\bf s}}{\mu_{\nu/2}(g_{{\bf S},{\bf s}})} = \Expect{{\bf S},{\bf s}}{\Expect{{\bf z}\sim \mu_{\nu/2}^{I}}{1-g(a\xor {\bf z}\xor {\bf s})}}.
     \]
     Note that ${\bf s}$ is distributed according to $\mu_{\nu/2}^{[n]\setminus I}$,
     and clearly ${\bf z}$ is distributed according to $\mu_{\nu/2}^{I}$ independently, so the distribution of ${\bf s}\xor {\bf z}$
     is $\mu_{\nu/2}^n$. Thus, the distribution of $a\xor {\bf z}\xor {\bf s}$ is of a $(1-\nu)$-correlated point with $a$, and therefore we get that
     \begin{align*}
       \Expect{{\bf S},{\bf s}}{\mu_{\nu/2}(g_{{\bf S},{\bf s}})} &= \Prob{{\bf z}\sim(1-\nu)\text{ correlated with } a}{g({\bf z}) = 0} \\
                                  &= 2^n \Prob{({\bf x},{\bf z}) \sim (1-\nu) \text{ correlated}}{{\bf x}=a, g({\bf x}) = 1, g({\bf z}) = 0}\\
                                  &= 2^n \Prob{({\bf x},{\bf z}) \sim (1-\nu) \text{ correlated}}{E} \cdot p_a
                                  \geq \frac{1}{4}\nu^{3/4}.
     \end{align*}
     In the last inequality, we used $p_a\geq 1/2^{n+1}$ and $\Prob{}{E}\geq \half\nu^{3/4}$.
   \end{proof}
   Thus, with probability at least $\frac{1}{8}\nu^{3/4}$ over ${\bf S},{\bf s}$, we have that
   $\mu_{\nu/2}(g_{{\bf S},{\bf s}})\geq \frac{1}{8}\nu^{3/4}$; denote the set of these tuples by $G$.
   For each $(S,s)$, let
   \[ h_{S,s}(z) = \Tup{\nu/2}{1/2} g_{S,s}(z), \]
   and define their average $h(z) = \Expect{{\bf S},{\bf s}}{h_{{\bf S},{\bf s}}(z)}$.
   By Lemma \ref{lem:2}, for every $(S,s)\in G$ we have that for ${\bf z}\sim\mu_{1/2}^I$,
   $h_{S,s}({\bf z}) \geq \nu$ except with probability $\xi/2$, thus for ${\bf z}\sim\mu_{1/2}^I$, $\Expect{({\bf S},{\bf s}) \in G}{h_{{\bf S},{\bf s}}({\bf z})} \geq \nu/2$ except with probability $\xi$. Therefore $h({\bf z})\geq \frac{\nu^{3/4}}{8} \cdot \frac{\nu}{2} = c$ except with probability $\xi$.
%   Denote the set of $z$'s satisfying this by $Z$.
   \begin{proposition}
     For each $y\in\power{I}$, $\mathcal{D}_1^a(y) \geq h(a\xor y) \mathcal{D}_2^a(y)$.
   \end{proposition}
   \begin{proof}
     By definition,
     \begin{align}\label{eq8}
       h(a\xor y) &= \Expect{S,s}{\cExpect{(w,u)\sim \mathbb{D}(\nu/2,1/2)}{u=a\xor y}{g_{S,s}(w)}}\notag\\
                  &= \Expect{S,s}{\cExpect{(w,u)\sim \mathbb{D}(\nu/2,1/2)}{u=a\xor y}{1-f(a\xor w\xor s)}}.
     \end{align}

     Consider the distribution over $(y',m',x',z')\sim \distD_\nu$ conditioned on $y' = y, x' = a$.
     Note that the distribution of $w$ in~\eqref{eq8} is $\mu_\nu^{J}$, where $J = \sett{i}{a_i = 1, y_i = 0} \subseteq I$,
     which is the same distribution as $x'\xor m' = a\xor m'$, and so the distribution of $a\xor w$ is the same as the distribution of $m'$.
     Since $s$ is independently distributed according to $\mu_{\nu/2}^{[n]\setminus I}$, the distribution of $a\xor w\xor s$ is the same as
     the distribution of $m'\xor s$, which is the same distribution as of $z'$.
     Thus,
     \[
       \eqref{eq8} = \cExpect{(y',m',x',z')\sim \distD_\nu}{y' = y, x' = a}{1 - f(z')}
                 = \frac{\Prob{(y',m',x',z')\sim \distD_\nu}{f(z') = 0 \land y' = y \land x' = a}}
                 {\Prob{(y',m',x',z')\sim \distD_\nu}{y' = y \land x' = a}}.
     \]
     The numerator of the last fraction is equal to
     \begin{align*}
        &\Prob{(y',m',x',z')\sim \distD_\nu}{x' = a, f(z') = 0}\cdot \cProb{(y',m',x',z')\sim \distD_\nu}{x' = a, f(z') = 0}{y' = y}\\
        &=\Prob{(y',m',x',z')\sim \distD_\nu}{x' = a, f(z') = 0}\cdot \distD_1^a(y).
     \end{align*}

     As for the denominator, it is equal to
     \[
     \Prob{(y',m',x',z')\sim \distD_\nu}{x' = a}\cProb{(y',m',x',z')\sim \distD_\nu}{x' = a}{y' = y}
     = \Prob{(y',m',x',z')\sim \distD_\nu}{x' = a}\cdot \distD_2^a(y).
     \]

     %Plugging the upper bound on the numerator, and the latter term for the denominator, we see that
     Plugging these expressions, we see that
     \begin{align*}
       \eqref{eq8}
       &= \frac{\distD_1^a(y)}{\distD_2^a(y)}
       \frac{\Prob{(y',m',x',z')\sim \distD_\nu}{x' = a, f(z') = 0}}{\Prob{(y',m',x',z')\sim \distD_\nu}{x' = a}}\\
       %&= \frac{\distD_1^a(y)}{\distD_2^a(y)}\cProb{(y',m',x',z')\sim \distD}{x' = a}{f(z')=0}\\
       &\leq \frac{\distD_1^a(y)}{\distD_2^a(y)}. \qedhere
     \end{align*}
   \end{proof}
   We now finish the proof of Proposition~\ref{prop:1}. Sampling ${\bf y}\sim \mu_{1/2}^{I}$, we see that ${\bf z}=a\xor {\bf y}$ is distributed according to
   $\mu_{1/2}^I$, thus with
   probability at least $1-\xi$ we have $h({\bf z})\geq c$. In this case, we would get by the previous proposition that
   $\mathcal{D}_1^a({\bf y}) \geq c \mathcal{D}_2^a({\bf y})$, as desired.
   \qed

\subsection{Proof of Theorem~\ref{thm:one_sided_error}}
In this section we prove Theorem~\ref{thm:one_sided_error}. Again, to make notations easier we write the proof in the case that $p = \rho = 1/2$,
and the proof carries over to the general case with minor adjustments.

We need the following claim that relates the tail of a function and its noise sensitivity.
\begin{claim}\label{claim:ns_tail}
  For any $g\colon\power{n}\to\power{}$ and $k\in\mathbb{N}$ we have
  that $W_{\geq k}[g]\leq {\sf NS}_{1/k}[g]$.
\end{claim}
\begin{proof}
  It is well known (e.g.~\cite[Theorem 2.49]{Odonnell}; our definition of ${\sf NS}_{\delta}$
  corresponds to $4{\sf NS}_{\delta/2}$ in the notation therein) that
  ${\sf NS}_{\delta}[g] = 2\sum\limits_{S\subseteq [n]}{(1-(1-\delta)^{\card{S}})\widehat{g}(S)^2}$,
  and since for $\card{S} \geq 1/\delta$ we have that $(1-\delta)^{\card{S}}\leq e^{-1}$,
  we conclude that ${\sf NS}_{\delta}[g]\geq 2(1-e^{-1}) W_{\geq 1/\delta}[g] \geq W_{\geq 1/\delta}[g]$.
  The claim now follows by choosing $\delta = 1/k$.
\end{proof}
We also need the analog of Claim~\ref{claim:low_inf_neg} for the one-sided error case.
\begin{claim}\label{claim:low_inf_neg_onesided}
  For every $\zeta,\tau > 0$ there is $\eta_2 > 0$ such that the following holds for $\lambda\in[\zeta,1]$.
  Suppose $f\colon(\power{n},\mu_{1/4})\to[0,1]$,
  $g\colon(\power{n},\mu_{1/2})\to\power{}$ are
  one-sided error solutions with $\lambda$ and error $\eta_2$,
  then for any $i\in[n]$ we have $I_i^{-}[g]\leq \tau$.
\end{claim}
\begin{proof}
  The proof is identical to the proof of Claim~\ref{claim:low_inf_neg}.
\end{proof}

We are now ready to prove Theorem~\ref{thm:one_sided_error}.
\begin{proof}[Proof of Theorem~\ref{thm:one_sided_error}]
Fix $\eps,\zeta>0$ as in the theorem statement, and choose $\nu_0$ from Lemma~\ref{lem:one_sided_NS} for $\zeta$.
Let $C = C(\zeta)$ be from Theorem~\ref{thm:Bourgain}, and choose
$\nu = C^{-5} \eps^9\nu_0$, $k = \ceil{1/\nu}$. Pick $\eta_1$ from
Lemma~\ref{lem:one_sided_NS} for $\zeta$ and $\nu$, pick $J = J(k,\zeta)$ from Theorem~\ref{thm:Bourgain}
and set $\tau = 2^{-10 J}\eps$. Finally, take $\eta_2$ from Claim~\ref{claim:low_inf_neg_onesided} for
$\zeta,\tau$; we prove the statement for $\eta = \min(\eta_1,\eta_2)$

Since $\eta \leq \eta_1$ and $\nu <\nu_0$, Lemma~\ref{lem:one_sided_NS} implies that
${\sf NS}_{\nu}[g]\leq \nu^{3/4}$, and therefore by Claim~\ref{claim:ns_tail}
we have
\[
W_{\geq k}[g] \leq \nu^{3/4}\leq \frac{\eps^2\sqrt{\nu}}{C \log^{1.5}(1/\nu)} \leq \frac{\eps^2}{C \sqrt{k}\log^{1.5}(k)}.
\]
Theorem~\ref{thm:Bourgain} now implies that there is a set $T\subseteq[n]$ of size $J$ such that
$g$ is $\eps^2$-close in $L_1$ to a $T$-junta $g'\colon\power{n}\to\power{}$. Write inputs $x\in\power{n}$ as $x = (\alpha,\beta)$, where
$\alpha\in\power{T},\beta\in\power{[n]\setminus T}$, and define $\tilde{g}\colon\power{n}\to [0,1]$ by
$\tilde{g}(\alpha,\beta') = \Expect{\bm{\beta}}{g(\alpha,\bm{\beta})}$. Note that $\tilde{g}$ is the closest function to $g$ in $L_2$ that depends
only on coordinates from $T$,
so $\norm{g-\tilde{g}}_2\leq \norm{g-g'}_2\leq \sqrt{\norm{g-g'}_1}\leq \eps$ (we used the fact that $\card{g-g'}\leq 1$).
Therefore, $\norm{g - \tilde{g}}_1\leq \eps$.

By Claim~\ref{claim:low_inf_neg_onesided}, all of the negative
influences of $g$ are at most $\tau$, and so by Fact~\ref{fact:avg_dec_neg_inf} all of the negative influences of $\tilde{g}$
are also at most $\tau$. Therefore, by Fact~\ref{fact:gglrs} there is a monotone function $h\colon\power{n}\to\mathbb{R}$
such that $\norm{\tilde{g} - h}_1\leq 4^J J \cdot \tau\leq \eps$, hence by the triangle inequality
$\norm{g - h}_1\leq \norm{g - \tilde{g}}_1+\norm{\tilde{g} - h}_1\leq 2\eps$. Finally, for $h'(x) = 1_{h(x)\geq 1/2}$
we have that $h'$ is a monotone, Boolean-valued $J$-junta and $\norm{g - h'}_1\leq 2\norm{g - h}_1\leq 4\eps$, and the proof is concluded.
\end{proof}

\section{Results for large noise rates}\label{sec:large_noise}
The proof of Theorem~\ref{thm:main_2func_smallrho} is composed of two parts.

In the first part, Lemma~\ref{lem:strong_junta_approx_smallrho}, we show that if $f,g$ are approximate solutions, then $g$ can be approximated
by a junta in a stronger manner, similarly to Lemma~\ref{lem:strong_junta_approx}.
The proof of Lemma~\ref{lem:strong_junta_approx_smallrho} is very similar to the proof of Lemma~\ref{lem:strong_junta_approx}:
the only place where monotonicity/almost-monotonicity is used is in Claim~\ref{claim:low_inf_gen}, and
we show it is not needed in the case that $\rho\geq \half + \zeta$.

This approximation by junta result allows us to reduce the problem to the case that $n$ is constant,
which we then prove in a similar way to the proof of Lemma~\ref{lemma:main_exact} (though much simpler).

\subsection{Strong approximation by junta}
We begin with the the following claim asserting that the negative influences of $g$ are all small.
\begin{claim}\label{claim:low_inf_neg_smallrho}
  For every $\zeta,\tau > 0$ there is $\eta_1 > 0$ such that the following holds for $\lambda\in[\zeta,1]$,
  $p\in[\zeta,1-\zeta]$ and $\rho\in [0,1-\zeta]$.
  If $\norm{\Tdown{p}{\rho p} f - \lambda g}_1\leq \eta_1$, then for any $i\in[n]$ we have $I_i^{-}[g]\leq \tau$.
\end{claim}
The proof is an easy adaptation of the proof of Claim ~\ref{claim:low_inf_neg} and is omitted.
The next claim is an adaptation of Claim~\ref{claim:low_inf_gen}.

\begin{claim}\label{claim:low_inf_gen_smallrho}
  For every $\zeta > 0$ there is $m\in \mathbb{N}$ such that for all $\tau > 0$ there is $\eta_2 > 0$ such that the following holds
  for all $\lambda\in[\zeta,1]$, $p\in[\zeta,1-\zeta]$ and $\rho\in[0,\half - \zeta]$.
  Let $g\colon(\power{n},\mu_{p})\to\power{}$, $f\colon(\power{n},\mu_{\rho p})\to[0,1]$ be such that
  $\norm{\Tdown{p}{p\rho} f - \lambda g}_1\leq \eta_2$.
  Then for each $M\subseteq[n]$ of size $m$,
  \[
    I_M[g] = \sum\limits_{S\supseteq M}{\widehat{g}^2(S)}\leq \tau.
  \]
\end{claim}
\begin{proof}
  Let $m = \ceil{\log(2/\zeta)}^2$
  and choose $\eta_1$ from Claim~\ref{claim:low_inf_neg_smallrho} for $\tau/12$.
  We prove the statement for $\eta_2 = \min(\eta_1, \frac{\zeta^3(1-\zeta)^2}{2^{2m+3}(\sqrt{m}+1)}\tau)$.

  Assume we have a set $M$ of size $m$ violating this condition. As in the proof of Lemma~\ref{lem:strong_junta_approx},
  we conclude that there is $A\subseteq M$ of size $\sqrt{m}$, without loss of generality $A = \set{1,\ldots,\sqrt{m}}$,
  such that
    \[
    \Prob{{\bf x}}{\text{all coordinates in $A$ are sensitive
    on $({\bf x},a)$ for $g$}}\geq \frac{\tau}{2^{2m}} \defeq t.
    \]
  Sample ${\bf x}\sim\mu_{p}$, let $\bf{x'(1)},\dots,{\bf x'(\sqrt{m})}$ be given by ${\bf x'(i)} = {\bf x}\xor e_i$, and consider
  a coupled downwards random walk from ${\bf x}$ and the ${\bf x'(i)}$'s to
  ${\bf y} = {\bf x}\land {\bf z},{\bf y'(1)} = {\bf x'(1)}\land {\bf z},\ldots,{\bf y'(\sqrt{m})} = {\bf x'(\sqrt{m})}\land {\bf z}$, where ${\bf z}\sim \mu_{\rho}$.
  Then with probability at least $t$ we have that $f({\bf x})\neq f({\bf x'(1)})=\dots=f({\bf x'(\sqrt{m})})$, and conditioned
  on that there are two cases: (a) $g({\bf x}) = 1$ with probability at least $\half$, or (b) $g({\bf x}) = 0$ with probability at least $\half$.

  \paragraph{Case (a).}
    This case is the same as in Lemma~\ref{lem:strong_junta_approx}.
    Denote by $E$ the event that $g({\bf x}) = 1$ and all of $A$'s coordinates are sensitive on $x$,
    and let $X(E)$ denote the set of $x$'s for which this event holds. By the approximate eigenvalue condition on $x$ we have that
    \[
    \sum\limits_{x\in X(E)}\mu_{p}(x) \card{\cExpect{({\bf y},{\bf v})\sim \mathbb{D}(\rho p,p)}{{\bf v}=x}{f({\bf y})}-\lambda}\leq \eta_2,
    \]
    implying (since in this case $\mu_{p}(X)\geq t/2$) that $\cExpect{{\bf x},{\bf y}}{{\bf x}\in X(E)}{f({\bf y})}\geq \lambda - \frac{\eta_2}{t/2}$.
    By the approximate eigenvalue condition on ${\bf x'(i)}$ we have
    \[
    \sum\limits_{x\in X(E)}\mu_{p}(x'(i)) \cExpect{({\bf y},{\bf v})\sim \mathbb{D}(\rho p,p)}{{\bf v}=x}{f({\bf y'(i)})} \leq \eta_2,
    \]
    and since $\mu_p(x'(i))\geq \zeta(1-\zeta)\mu_p(x)$, we conclude that $\cExpect{{\bf x},{\bf y'(i)}}{{\bf x}\in X(E)}{f({\bf y'(i)})}\leq \frac{\eta_2}{\zeta(1-\zeta)t/2}$.
    Therefore we get that
    \begin{align*}
    \lambda - (\sqrt{m} +1)\frac{\eta_2}{\zeta(1-\zeta)t/2}
    &\leq \cExpect{{\bf x},{\bf y},{\bf y'(i)}\text{'s}}{{\bf x}\in X(E)}{f({\bf y})-\sum\limits_{i\in A}{f({\bf y'(i)})}}\\
    &\leq \cProb{{\bf x},{\bf y},{\bf y'(i)}\text{'s}}{{\bf x}\in X(E)}{{\bf y}\neq {\bf y'(i)} ~\forall i\in A}\\
    &=\Prob{{\bf z}}{{\bf z}_i = 1 ~\forall i\in A}= 2^{-\sqrt{m}}\leq \lambda/2,
    \end{align*}
    where in the second transition we used the fact that $f$ is bounded between $0$ and $1$. Rearranging and plugging
    in the expression for $t$ yields a contradiction for the choice of $\eta_2$.

   \paragraph{Case (b).}
   Let $E_0$ be the event that $g({\bf x}) = 0$ and all of $A$'s coordinates are sensitive on ${\bf x}$. We would like to identify a large subevent $E \subseteq E_0$ such that ${\bf x}_1 = {\bf x}_2 = 0$ and $g(1,1,{\bf x}_3,\dots,{\bf x}_n) = 1$.

   Let $E'\subseteq E_0$ be the event that ${\bf x}_1 = 1$ or ${\bf x}_2 = 1$.
   For any $x$ satisfying the event $E'$, we have that $(x'(1),x)$ or $(x'(2),x)$
   is a negatively influential edge in direction $i=1$ or $i=2$. It follows that
   $\Prob{}{E'}\leq I_1^{-}[g]+I_2^{-}[g]$, and since $\eta_2\leq \eta_1$ we have by Claim~\ref{claim:low_inf_neg_smallrho}
   that $I_1^{-}[g],I_2^{-}[g]\leq \tau/12$. Therefore $\Prob{}{E'}\leq \tau/6$.

   Let $E'' \subseteq E_0 \setminus E'$ be the event that $g({\bf x}'(1,2)) = 0$, where ${\bf x}'(1,2) = (1,1,{\bf x}_3,\dots,{\bf x}_n)$. For any $x$ satisfying the event $E''$, we have that $(x'(1),x'(1,2))$ is a negatively influential edge in direction $i=2$. It follows that $\Prob{}{E''} \leq I_2^-[g]$. As above, $\Prob{}{E''} \leq \tau/12$. It follows that the probability of the event $E := E_0 \setminus (E' \cup E'')$ is at least $\tau/2 - \tau/4 = \tau/4$.

   Write inputs $x\in\power{n}$ as $x = (b_1,b_2,\beta)$ for
   $b_1,b_2\in\power{}$ and $\beta\in\power{n-2}$, and let
   $B = \sett{\beta}{(0,0,\beta)\in E}$; then $\mu_p(B)\geq \tau/4$.
   Let ${\bf z}\sim\mu_{\rho}^{n}$ and write ${\bf z}=({\bf z(1)},{\bf z(2)})$, where ${\bf z(1)}\in\power{2}$ and ${\bf z(2)}\in\power{n-2}$,
   and denote ${\bf y} = \bm{\beta}\land {\bf z(2)}$.
   Using the approximate eigenvector condition on points $(0,0,\beta)$ for $\beta\in B$
   we get that $\cExpect{\bm{\beta},{\bf z}, {\bf y}}{\bm{\beta}\in B}{f(0,0,{\bf y})}\leq \frac{\eta_2}{\tau/4}$,
   and by the approximate eigenvector condition on points $(1,0,\beta)$ for $\beta\in B$ we get that
   $\cExpect{\bm{\beta},{\bf z},{\bf y}}{\bm{\beta}\in B}{f((1,0)\land {\bf z(1)}, {\bf y})}\geq \lambda - \frac{\eta_2}{\zeta(1-\zeta)\tau/4}$.
   Note that in case $z(1)_1 = 0$, the distribution of $((1,0)\land {\bf z(1)}, {\bf y})$ is identical to the distribution of
   $(0,0,{\bf y})$, and thus the expected value of $f$ on these points is at most $\frac{\eta_2}{\tau/4}$. We conclude that
   $\Prob{}{{\bf z(1)_1}=1}\cExpect{\bm{\beta},{\bf z},{\bf y}}{\bm{\beta}\in B}{f(1,0, {\bf y})}\geq \lambda - 2\frac{\eta_2}{\zeta(1-\zeta)\tau/4}$,
   which, since $z(1)_1=1$ with probability $\rho$, implies that
   $\cExpect{\bm{\beta},{\bf z},{\bf y}}{\bm{\beta}\in B}{f(1,0, {\bf y})}\geq \frac{\lambda}{\rho} - 2\frac{\eta_2}{\rho\zeta(1-\zeta)\tau/4}$. Analogously,
   we have $\cExpect{\bm{\beta},{\bf z},{\bf y}}{\bm{\beta}\in B}{f(0,1, {\bf y})}\geq \frac{\lambda}{\rho} - 2\frac{\eta_2}{\rho\zeta(1-\zeta)\tau/4}$.

   We show a lower bound and an upper bound on $\Expect{\bm{\beta}}{\Tdown{p}{\rho p} f(1,1,\bm{\beta})}$, which together give a contradiction.
   By definition and non-negativity of $f$ it is equal to
   \begin{align*}
   \cExpect{\bm{\beta}}{\bm{\beta}\in B}{\Expect{{\bf z}}{f((1,1)\land{\bf z(1)},\bm{\beta}\land {\bf z(2)})}}
   &\geq
   \Prob{{\bf z(1)}}{{\bf z(1)} = (1,0)}
   \cdot \cExpect{\bm{\beta},{\bf y}}{\bm{\beta}\in B}{f(1,0,{\bf y})}\\
   &+ \Prob{{\bf z(1)}}{{\bf z(1)} = (0,1)}
   \cdot \cExpect{\bm{\beta},{\bf y}}{\bm{\beta}\in B}{f(0,1,{\bf y})}.
   \end{align*}
   The first expression is equal to
   $\rho(1-\rho)\cdot \cExpect{\bm{\beta},{\bf y}}{\bm{\beta}\in B}{f(1,0,{\bf y})}\geq \rho(1-\rho)(\frac{\lambda}{\rho} - 2\frac{\eta_2}{\rho\zeta(1-\zeta)\tau/4})$,
   and we have the same lower bound for the second expression. Hence we get
   $\Expect{\bm{\beta}}{\Tdown{p}{\rho p} f(1,1,\bm{\beta})}\geq 2(1-\rho)\lambda - 4\frac{(1-\rho)\eta_2}{\zeta(1-\zeta)\tau/4}$,
   which by the assumption that $\rho \leq \half - \zeta$ and choice of $\eta_2$ is at least $(1+2\zeta)\lambda - \zeta^2 \geq (1+\zeta)\lambda$.

   On the other hand, using the approximate eigenvector condition on points of the form $(1,1,\beta)$ for $\beta\in B$, we get
   \[
    \sum\limits_{\beta\in B}\mu_{p}(1,1,\beta) \card{\Tdown{p}{\rho p} f(1,1,\beta) - \lambda g(1,1,\beta)}\leq \eta_2.
   \]
   As $\mu_{p}(1,1,\beta)\geq \zeta^2(1-\zeta)^2\mu_p(\beta)$ and $\mu_p(B)\geq \tau/4$, we get that
   \[
   \cExpect{\bm{\beta}}{\bm{\beta}\in B}{\card{\Tdown{p}{\rho p} f(1,1,\bm{\beta}) - \lambda g(1,1,\bm{\beta})}}
   \leq \frac{\eta_2}{\zeta^2(1-\zeta)^2\tau/4}\leq \lambda\zeta/2,
   \]
   which by $g\leq 1$ implies $\cExpect{\bm{\beta}}{\bm{\beta}\in B}{\Tdown{p}{\rho p} f(1,1,\bm{\beta})}\leq (1+\zeta/2)\lambda$, and contradiction.
\end{proof}

We can now prove the approximation by junta result we will need.
The proof is the same as the proof of Lemma~\ref{lem:strong_junta_approx} except that we use
the above replacement instead of Claim~\ref{claim:low_inf_gen}, and is omitted.
\begin{lemma}\label{lem:strong_junta_approx_smallrho}
  For every $\zeta>0$ there is $J\in\mathbb{N}$ such that for all $\eps>0$ there is $\eta >0$ such that
  the following holds for all $\lambda\in[\zeta,1]$, $p\in[\zeta,1-\zeta]$
  and $\rho\in[0,1/2-\zeta]$. If $f\colon(\power{n},\mu_{\rho p})\to[0,1]$, $g\colon(\power{n},\mu_{p})\to\power{}$ satisfy
  $\norm{\Tdown{p}{\rho p} f - \lambda g}_1\leq \eta$, then
  $g$ is $\eps$-close to a $J$-junta.
\end{lemma}
%\begin{proof}
%  Pick $m(\zeta)$ from Claim~\ref{claim:low_inf_gen_smallrho},
%  and let $J = J(m)$, $\delta = \delta(m,\eps)$ be from Theorem~\ref{thm:KS}.
%  Let $C = C(\zeta)$ be from Theorem~\ref{thm:Bourgain}, and set $\eta_3 = \zeta^2\delta/(16C\log(1/\delta))$,
%  $k = \ceil{\log(1/\eta_3)}$. Finally, take $\tau = \tau(\delta,\zeta,m,k)$ from Lemma~\ref{lem:Bourgain_gen}
%  and $\eta_2 = \eta_2(\zeta,m,\tau)$ from Claim~\ref{claim:low_inf_gen_smallrho}. We prove the statement for $\eta = \min(\eta_2,\eta_3)$.
%
%  Let $f,g$ be as in the statement of the lemma. By
%  Claim~\ref{claim:low_inf_gen_smallrho} we have $I_{M}[g]\leq \tau$ for all $\card{M} = m$, and by Lemma~\ref{lem:vanishing_tail}
%  we have $W_{\geq k}[g]\leq 2\lambda^{-2}(\eta + \rho^{-k}) < \frac{\delta}{2C\sqrt{k}\log^{1.5}(k)}$. Therefore, Lemma~\ref{lem:Bourgain_gen}
%  implies that we must have $W_{\geq m}[g]\leq \delta$ (otherwise $g$ would be a counter example to it), and therefore by Theorem~\ref{thm:KS}
%  we get that $g$ is $\eps$-close to a $J$-junta.
%\end{proof}

\subsection{The case $n$ is constant}
This section is devoted for the proof of the following lemma.
\begin{lemma}\label{lemma:main_exact_smallrho}
  For any $\zeta > 0$ and $n\in\mathbb{N}$
  there exists $\eta_0>0$ such that the following holds for all $\lambda \in[\zeta,1]$,
  $p\in[\zeta,1-\zeta]$, $\rho\in[\zeta,1/2-\zeta]$
  and $0<\eta\leq\eta_0$.
  If $\norm{\Tdown{p}{\rho p} f - \lambda g}_{\infty}\leq \eta$
  then:
  \begin{itemize}
    \item Either $g$ is constant or $g$ is an AND function of width $r$ where $r\leq \ceil{\log(2/\zeta)}$.
    \item $\norm{f - 2^r\lambda\cdot g}_{\infty}\leq \zeta^{-3n^2} \eta$.
  \end{itemize}
\end{lemma}
The proof is similar to the proof in Section~\ref{sec:exact_rho_half}, but somewhat simpler. As there,
it will be more convenient for us to identify $\power{n}$ with subsets of $[n]$ in this section, and think of
the functions $f,g$ as being functions over subsets of $[n]$. Throughout this section we denote $\T = \Tdown{p}{\rho p}$.

Fix $\zeta,n$, and choose $\eta = \zeta^{3n^2+3n+4}$. Let
$f,g$ be functions as in the statement of the lemma.
\begin{claim}
  $g$ is monotone.
\end{claim}
\begin{proof}
  Repeat the proof of Claim~\ref{claim:g_mono_rho_half} verbatim,
  using the fact that $\Prob{}{i\not\in C} = 1-\rho\geq \half$.
\end{proof}
If $g\equiv 0$, we are done proving the structure for $g$ (and we prove the
structure of $f$ below), and otherwise $g$ has a minterm.
The following claim shows that the value of $f$ on a minterm must
be very close to a specific value (similarly to Claim~\ref{claim:value_in_minterm}).
\begin{claim}\label{claim:value_in_minterm_rhosmall}
  If $M$ is a minterm of $g$, then $\card{f(M) - \lambda \rho^{-\card{M}}}\leq \rho^{-3\card{M}}\eta$.
\end{claim}
\begin{proof}
  Since $g(M)=1$, we get $\card{\T f(M) - \lambda}\leq \eta$, and since $\T f(M) = \sum\limits_{A\subseteq M}\rho^{\card{A}}(1-\rho)^{\card{M\setminus A}} f(A)$,
  the triangle inequality implies that $\card{f(M) - \rho^{-\card{M}}\lambda}\leq \rho^{-\card{M}}\eta + \rho^{-\card{M}}\sum\limits_{A\subsetneq M} \rho^{\card{A}}({1-\rho})^{\card{M\setminus A}} f(A)$; to finish the proof we upper-bound the last sum by $\card{M}\eta$. Note that for every $A\subsetneq M$, choosing
  ${\bf B}\subsetneq M$ randomly of size $\card{M}-1$, we have that $A\subseteq {\bf B}$ with probability at least $1/\card{M}$,
  hence by the non-negativity of $f$ there is $B$ of size $\card{M}-1$ such that
  \[
  \sum\limits_{A \subsetneq M} \rho^{\card{A}}({1-\rho})^{\card{M\setminus A}}f(A)\leq \card{M}\sum\limits_{A\subseteq B} \rho^{\card{A}}({1-\rho})^{\card{M\setminus A}}f(A)
  \leq \card{M}\T f(B),
  \]
  and we fix such $B$.
  Since $B\subsetneq M$ and $M$ is a minterm of $g$, we have that $g(B)=0$, implying that
  $\T f(B)\leq \eta$ and therefore $ \sum\limits_{A \subsetneq M} \rho^{\card{A}}({1-\rho})^{\card{M\setminus A}}f(A)\leq \card{M} \eta$.
\end{proof}

\begin{claim}\label{claim:value_in_zero_rhosmall}
  If $g(M) = 0$, then $f(M)\leq \rho^{-\card{M}}\eta$.
\end{claim}
\begin{proof}
  We have $\rho^{\card{M}} f(M)\leq \T f(M) \leq \lambda g(M) + \eta = \eta$, and the result follows by rearranging.
\end{proof}

Next we show that the value of $f$ above a minterm of minimal size must also be close to
the same value.
\begin{claim}\label{claim:value_above_minterm_rhosmall}
  Let $M$ be a minterm of $g$ of minimal, and let $M'\supseteq M$.
  Then $\card{f(M') - \lambda \rho^{-\card{M}}}\leq \rho^{-3\card{M'}^2}\eta$.
\end{claim}
\begin{proof}
  The proof is by induction on $\card{M'\setminus M}$.
  The base case $M' = M$ is proved in Claim~\ref{claim:value_in_minterm_rhosmall}.
  Let $k\geq 0$, assume the statement holds when $\card{M'\setminus M}\leq k$,
  and prove for $\card{M'\setminus M} = k+1$.

  Let $T_1 = \sett{A\subsetneq M'}{A\supseteq M}$, $T_2 = \sett{A\subseteq M'}{A\not\supseteq M}$
  and write
  \begin{equation}\label{eq3}
  \T f(M') =
  \rho^{\card{M'}} f(M')
  + \sum\limits_{A\in T_1}{\rho^{\card{A}}(1-\rho)^{\card{M'\setminus A}} f(A)}
  + \sum\limits_{A\in T_2}{\rho^{\card{A}}(1-\rho)^{\card{M'\setminus A}} f(A)}.
  \end{equation}
  Our goal is to extract an approximation for $f(M')$, and for
  that we approximate the left-hand side as well as the two sums on the right-hand side.
  For the left-hand side, since $g$ is monotone we have $g(M') = 1$ and so $\card{\T f(M') - \lambda}\leq \eta$.

  For the first sum, for any $A\in T_1$ the induction hypothesis implies that $\card{f(A) - \lambda \rho^{-\card{M}}}\leq \rho^{-3\card{A}^2}\eta
  \leq \rho^{-3(\card{M'}-1)^2}\eta$,
  and therefore
  \[
  \card{\sum\limits_{A\in T_1}{\rho^{\card{A}}(1-\rho)^{\card{M'\setminus A}} f(A)} - \lambda \rho^{-\card{M}}\cdot \Prob{A\sim\mu_{\rho}}{A\in T_1}}
  \leq \rho^{-3(\card{M'}-1)^2}\eta.
  \]
  We also note that $\Prob{A\sim\mu_{\rho}}{A\in T_1} = \rho^{\card{M}}(1-\rho^{\card{M'\setminus M}})$.

  For the second sum, by non-negativity of $f$, plugging the approximations we have so far we get from~\eqref{eq3} that
  \[
    \sum\limits_{A\in T_2}{\rho^{\card{A}}(1-\rho)^{\card{M'\setminus A}} f(A)}
    \leq \rho^{\card{M'\setminus M}}\lambda + (1+\rho^{-3(\card{M'}-1)^2})\eta.
  \]
  We claim that this implies that for every $A\in T_2$ we have $g(A) = 0$. Indeed,
  otherwise there is $A^{\star}\in T_2$ that is a minterm of $g$, and by Claim~\ref{claim:value_in_minterm_rhosmall}
  its contribution to the left-hand side is at least
  \begin{align*}
  \rho^{\card{A^{\star}}}(1-\rho)^{\card{M'\setminus A^{\star}}}(\lambda \rho^{-\card{A^{\star}}} -  \rho^{-3\card{A^{\star}}}\eta)
  &\geq (1-\rho)^{\card{M'\setminus M}}\lambda - \rho^{-2\card{M'}}\eta \\
  &\geq (\rho^{\card{M'\setminus M}} + \zeta^{n})\lambda - \rho^{-2\card{M'}}\eta.
  \end{align*}
  In the first inequality we used $\card{A^{\star}}\geq \card{M}$ (since $M$ is minterm of minimal size),
  and in the second inequality we used $\rho\leq \half - \zeta$. Combining the two inequalities we get that
  $\zeta^{n}\lambda\leq \rho^{-3\card{M'}^2}\eta\leq \zeta^{-3n^2}\eta$, which is a contradiction to the choice of $\eta$.

  We conclude that $g(A) = 0$ for all $A\in T_2$, and by Claim~\ref{claim:value_in_zero_rhosmall} we conclude that
  $f(A)\leq \rho^{-\card{A}}\eta\leq \rho^{-\card{M'}}\eta$. Plugging all approximations we have in~\eqref{eq3}, by the triangle
  inequality we get that
  \[
  \card{\rho^{\card{M'}} f(M') + (1-\rho^{\card{M'\setminus M}})\lambda - \lambda}
  \leq (1+\rho^{-\card{M'}}+\rho^{-3(\card{M'}-1)^2})\eta
  \leq \rho^{-3\card{M'}^2+\card{M'}}\eta.
  \]
  Simplifying and dividing by $\rho^{\card{M'}}$ finishes the proof.
\end{proof}

\paragraph{Structure of $g$.}
We are now ready to prove that $g$ is an AND function, and we do it by the way of contradiction.
Suppose $g$ has at least two minterms, let $M_1$ be a minterm of $g$ of minimum size, and let $M_2\neq M_1$
be another minterm. Consider $M = M_1\cup M_2$, and let $A\subseteq M$ be sampled so that each $i\in M$
is included independently with probability $\rho$. Let $E_1$ be the event that $A\supseteq M_1$
and let $E_2$ be the event that $A = M_2$. We have that
\[
\T f(M) \geq
\Prob{A}{A\in E_1}\cdot\cExpect{A}{E_1}{f(A)}
+\Prob{A}{A\in E_2}\cdot f(M_2),
\]
and we lower-bound the right-hand side. The probability that $A$ is in $E_1$
is $\rho^{\card{M_1}}$, and by Claim~\ref{claim:value_above_minterm_rhosmall}, for every
$A$ satisfying $E_1$ we have $f(A)\geq \lambda \rho^{-\card{M_1}}\lambda - \rho^{-3n^2}\eta$.
The probability that $A$ is in $E_2$ is at least $\rho^n$, and by Claim~\ref{claim:value_in_minterm_rhosmall},
$f(M_2)\geq \lambda \rho^{-\card{M_2}} - \rho^{-3n}\eta$. Therefore we get that
$\T f(M)\geq (1+\rho^n)\lambda - \rho^{-3n^2 - 3n}\eta > \lambda + \eta$, by the choice of $\eta$.
This is a contradiction to $\T f(M)\leq \lambda g(M) + \eta = \lambda + \eta$.

It follows that there is $T\subseteq [n]$ such that $g = {\sf AND}_T$
(or $g\equiv 0$, a case that arose during the proof).

\paragraph{Structure of $f$.} If $g\equiv 0$, it follows by Claim~\ref{claim:value_in_zero_rhosmall}
that $f$ is close to the $0$ function. Otherwise, on any $M$ such that $g(M) = 1$, i.e.\ above
the minterm $T$, we have by Claim~\ref{claim:value_above_minterm_rhosmall} that $\card{f(M) - \rho^{-\card{T}}\lambda}\leq \rho^{-3n^2}\eta$,
and on any $M$ such that $g(M) = 0$ we have by Claim~\ref{claim:value_in_zero_rhosmall} that $f(M)\leq \rho^{-n}\eta$,
so $\norm{f - \rho^{-\card{T}}\lambda\cdot {\sf AND}_T}_{\infty}\leq \rho^{-3n^2}\eta$.
\qed

\begin{remark}
  A stronger bound on the distance of $f$ from $\rho^{-\card{T}}\lambda {\sf AND}_T$ may be proved using inversion,
  as in Section~\ref{sec:exact_rho_half}.
\end{remark}

\subsection{Proof of Theorem~\ref{thm:main_2func_smallrho}}
Fix $\zeta,\eps>0$
and let $J = J(\zeta)$ be from Lemma~\ref{lem:strong_junta_approx_smallrho}.
Take $\eta_0$ from Lemma~\ref{lemma:main_exact_smallrho} for $\zeta$ and $n=J$.
Set $\eps' = \zeta^{J + 3}\eps$, and pick $\eta_1$ from Lemma~\ref{lem:strong_junta_approx_smallrho}
for $\zeta,\eps'$.
We prove the statement for $\eta = \min(\zeta^{3n^2+n}\eps^2\eta_0,\eta_1)$.

By Lemma~\ref{lem:strong_junta_approx_smallrho} we get that $g$ is $\eps'$-close to
a junta $h$ depending on at most $J$ variables, say on $T\subseteq[n]$.
We write points $x\in \power{n}$ as $(\alpha,\beta)$, where $\alpha\in\power{T}$ and $\beta\in\power{[n]\setminus T}$,
and for each $\beta\in\power{[n]\setminus T}$, define $\tilde{f}_{\beta}\colon\power{T}\to[0,1]$ and $g_\beta\colon \power{T} \to \power{}$ by
\[
\tilde{f}_{\beta}(\alpha) = \Expectnoop{{\bf z}\sim\mu_{\rho}^{[n]\setminus T}}{f(\alpha, \beta\land{\bf z})}, \quad
g_\beta(\alpha) = g(\alpha,\beta).
\]
Let $B = \sett{\beta\in\power{[n]\setminus T}}{\norm{\tilde{f}_{\beta} - \lambda g_{\beta}}_{\infty}\leq \zeta^{3n^2}\eps\eta_0}$.
Since for any $\alpha,\beta$ we have $\T\tilde f_{\beta}(\alpha) = \T f(\alpha,\beta)$, it follows that
$\Expect{\bm{\beta}}{\norm{\T\tilde{f}_{\bm{\beta}} - \lambda g_{\bm{\beta}}}_1} = \norm{\T f - \lambda g}_1\leq \eta$.
Therefore Markov's inequality
implies that with probability at least $1-\frac{\eta}{\zeta^{3n^2+n}\eps\eta_0}\geq 1-\eps$ over $\bm{\beta}\sim\mu_p$ we have
$\norm{\T\tilde{f}_{\bm{\beta}} - \lambda g_{\bm{\beta}}}_1\leq \zeta^{3n^2+n}\eps\eta_0$,
in which case $\bm{\beta}\in B$. In particular, we conclude that $\Prob{\bm{\beta}\sim\mu_p^{[n]\setminus T}}{\bm{\beta}\in B}\geq 1-\eps$.

For each $\beta\in B$, Lemma~\ref{lemma:main_exact_smallrho} implies that $g_{\beta}$
is either the zero function or an AND function, i.e.\ there is $T'(\beta)$ which is either a subset of $T$ or $\bot$ such that $g_{\beta} = {\sf AND}_{T'(\beta)}$ (where ${\sf AND}_\bot = 0$).
Since $g$ is $\eps'$-close to a $T$-junta, if we choose $\bm{\beta},\bm{\beta'} \in B$ independently (according to $\mu_p$)
then on average $g_{\bm{\beta}}$ and $g_{\bm{\beta'}}$ are $\delta$-close, where $\delta \leq 2\eps'/\Pr[B] \leq 4\eps'$.
Thus there is $\beta^{\star} \in B$ such that $\Expect{\bm{\beta} \in B}{\norm{g_{\bm{\beta}} - g_{\beta^{\star}}}_1}\leq 4\eps'$,
and we denote $T^{\star} = T'(\beta^{\star})$.
By Markov's inequality, defining $B'\subseteq B$ by $B' = \sett{\beta\in B}{\norm{g_{\beta} - g_{\beta^{\star}}}_1\leq \zeta^{-J-1}}$, we have that
$\Prob{\bm{\beta}\sim\mu_p^{[n]\setminus T}}{\bm{\beta}\in B'}\geq 1-\frac{4\eps'}{\zeta^{-J-1}}\geq 1-\eps$. Note that if $\beta\in B'$, then
the functions ${\sf AND}_{T'(\beta)}$ and ${\sf AND}_{T^{\star}}$ agree on $\bm{\alpha}\sim\mu_p^T$ with probability $\geq 1-\zeta^{-J-1} > 1-\min_{\alpha\in\power{T}}\mu_p(\alpha)$, hence they must be the same function. In other words, we get that for every $\beta\in B'$ we have $g_{\beta} = {\sf AND}_{T^{\star}}$,
and so $g$ is $\eps$-close to ${\sf AND}_{T^{\star}}$.

To prove the structure for $f$, note that for each $\beta\in B'$, Lemma~\ref{lemma:main_exact_smallrho} implies that $\tilde{f}_{\beta}$ is $\eps$-close to
$2^{\card{T^{\star}}}\lambda\cdot {\sf AND}_{T^{\star}}$
in $L_{\infty}$, and averaging it over variables from $T\setminus T^{\star}$ can only decrease this distance. Thus, considering
$\tilde{f}\colon \power{T^{\star}}\to[0,1]$ defined by
\[
\tilde{f}(\alpha) = \Expectnoop{{\bf y}\sim\mu_{\rho p}^{[n]\setminus T^{\star}}}{f(\alpha,{\bf y})} =
\Expectnoop{\substack{\bm{\beta}\sim\mu_p^{[n]\setminus T}\\\bm{\alpha'}\sim\mu_{\rho p}^{T\setminus T^{\star}}}}{f_{\bm{\beta}}(\alpha,\bm{\alpha'})},
\]
we have that
$\norm{\tilde{f} - 2^{\card{T^{\star}}}\lambda\cdot {\sf AND}_{T^{\star}}}_{\infty}\leq
\mu_p(B')\cdot\eps  + (1-\mu_{p}(B')) \leq 2\eps$, and we are done.\qed.

\section{Open questions} \label{sec:open-questions}

Our work raises many open questions. Perhaps the most obvious is the quantitative aspect of our results:

\begin{open-question}
What is the optimal dependence between $\epsilon$ and $\delta$ in Theorem~\ref{thm:intro-main}?	
\end{open-question}
\noindent We can ask a similar question about the various results listed in Section~\ref{sec:main-results}.

Nehama~\cite{Nehama} showed that if we allow $\epsilon$ to depend on $n$, then we can choose $\epsilon = \Theta(\delta^3/n)$. Theorem~\ref{thm:intro-main} eliminates the dependence on $n$ in return for an exponential dependence on $\delta$. We conjecture that Theorem~\ref{thm:intro-one-sided} holds for $\epsilon = \delta^{\Theta(1)}$.

Nehama situates Theorem~\ref{thm:intro-main} in the larger context of approximate judgement aggregation, or equivalently, approximate polymorphisms. He considers not only functions satisfying
\[
 f(\mathbf{x_1} \land \cdots \land \mathbf{x_m}) \approx f(\mathbf{x_1}) \land \cdots \land f(\mathbf{x_m}),
\]
but also functions satisfying
\[
 f(\mathbf{x_1} \oplus \cdots \oplus \mathbf{x_m}) \approx f(\mathbf{x_1}) \oplus \cdots \oplus f(\mathbf{x_m}),
\]
showing (using linearity testing) that the latter must be close to XORs. More generally, we can replace $\land,\oplus$ with an arbitrary Boolean function (or even a function on a larger domain):

\begin{open-question}
Fix $\phi\colon	\{0,1\}^m \to \{0,1\}$. Suppose $f\colon \{0,1\}^n \to \{0,1\}$ satisfies
\[
 f(\phi(\mathbf{x_1},\ldots,\mathbf{x_m})) \approx \phi(f(\mathbf{x_1}),\ldots,f(\mathbf{x_m}))
\]
for random $\mathbf{x_1},\ldots,\mathbf{x_m} \in \{0,1\}^n$, where $\phi(x_1,\ldots,x_m)$ signifies elementwise application.

\noindent What can we say about $f$?
\end{open-question}

Dokow and Holzman~\cite{DH09} showed (essentially) that when $\phi$ is a non-trivial function which is not an AND or an XOR, then the only exact solutions are dictatorships.
We conjecture that when $\phi$ is such that the only exact solutions are dictatorships, then approximate solutions are approximate dictatorships.

Finally, let us mention the following tantalizing question:

\begin{open-question}
What can be said about functions $f\colon \{0,1\}^n \to \{0,1\}$ satisfying
\[
 \Pr[f(\mathbf{x} \land \mathbf{y}) = f(\mathbf{x}) \land f(\mathbf{y})] \geq \frac{3}{4} + \epsilon?
\]	
\end{open-question}
\noindent We remark that the $\frac{3}{4}$ bound on the right hand side is natural in light of
semi-random functions $f\colon\power{n}\to\power{}$, chosen by taking $f(x)$ to be a uniform bit when $\card{x}\approx \half n$,
and $f(x) = 0$ otherwise.

%\textcolor{red}{Is $3/4$ the correct threshold? If we are interested in both the $\mu_{1/2}$ and the $\mu_{1/4}$ structure of $f$, then it seems that $1/2$ is the correct threshold.}

\bibliographystyle{plain}
\bibliography{ref,mossel,all_mossel}
%\appendix

\end{document}